\documentclass[9pt,nocopyrightspace,preprint]{sigplanconf}

\usepackage{booktabs}
\usepackage[squaren]{SIunits}            
\usepackage{courier}            
\usepackage[scaled]{helvet} 
\usepackage{url}                  
\usepackage{enumitem}      
\usepackage[colorlinks=true,allcolors=blue,breaklinks,draft=false]{hyperref}   

\usepackage{multicol}
\usepackage{multirow}
\usepackage{pgf}
\usepackage[utf8x]{inputenc}
\usepackage{subfig}
\usepackage{amssymb}
\usepackage{amsfonts}
\usepackage{amsmath}
\usepackage{color}
\usepackage[boxed,linesnumbered]{algorithm2e}
\usepackage{textcomp}
\usepackage{tikz}
\usetikzlibrary{fit,arrows,automata,calc,backgrounds,shapes,snakes,positioning,decorations.pathreplacing}
\usepackage{graphicx}
\usepackage{amsmath}
\usepackage{amssymb}
\usepackage{mathpartir}
\usepackage{bm}
\usepackage{mathtools}
\usepackage{commath}
\usepackage{etoolbox}
\usepackage{amsthm}
\usepackage{times}
\usepackage{caption} 
\usepackage{colortbl}
\usepackage[compact]{titlesec}
\usepackage{mdframed}
\usepackage{wrapfig}
\usepackage{listings}
\usepackage{arydshln}
\usepackage{flushend}

\newcommand{\abegin}[1]{\texttt{deq}\texttt{#1}}
\newcommand{\aend}[1]{\texttt{end}\texttt{#1}}
\newcommand{\post}[1]{\texttt{post}\texttt{#1}}

\newcommand{\tinit}[1]{\texttt{tinit}\texttt{#1}}

\newcommand{\fork}[1]{\texttt{fork}\texttt{#1}}

\newcommand{\readX}[1]{\texttt{read}\texttt{#1}}

\newcommand{\acquire}[1]{\texttt{lock}\texttt{#1}}
\newcommand{\release}[1]{\texttt{unlock}\texttt{#1}}

\newcommand{\Threads}{\ensuremath{\mathit{Th}}}
\newcommand*{\tn}[2]{\tikz[baseline,remember picture]\node[inner sep=0pt,anchor=base] (#1) {#2};}
\newcommand{\eg}{\emph{e.g.},}
\newcommand{\ie}{\emph{i.e.},}

\newcommand{\threadof}{\mathit{thread}}

\newcommand{\nextTrans}{\ensuremath{\mathit{nextTrans}}}
\newcommand{\nextop}{\ensuremath{\mathit{next}}}
\newcommand{\postChain}{\ensuremath{\mathit{postChain}}}
\newcommand{\divergingPosts}{\ensuremath{\mathit{divergingPosts}}}
\newcommand{\reorderedPosts}{\ensuremath{\mathit{reorderedPosts}}}
\newcommand{\domof}{\ensuremath{\mathit{dom}}}
\newcommand{\eventof}{\ensuremath{\mathit{event}}}
\newcommand{\taskof}{\ensuremath{\mathit{task}}}

\newcommand{\destof}{\ensuremath{\mathit{dest}}}
\newcommand{\gtree}{\texttt{$\Gamma$-tree}}
\newcommand{\gidx}[1]{\texttt{$\Gamma$-idx}\texttt{#1}}
\newcommand{\getBegin}{\ensuremath{\mathit{getBegin}}}
\newcommand{\getEnd}{\ensuremath{\mathit{getEnd}}}

\newcommand{\emdpor}{{EM-DPOR}}
\newcommand{\emexplorer}{{EM-Explorer}}
\newcommand{\droidracer}{\textsc{DroidRacer}}

\newcommand{\myhline}{\arrayrulecolor{gray!40}\hdashline}

\newcommand{\defdepcovering}{\ref{def:dep-covering}}

\newcommand{\defhappensbefore}{\ref{def:happens-before}}
\newcommand{\defdependencerelation}{\ref{def:dependence-relation}}
\newcommand{\deflazypersistentset}{\ref{def:dc-set}}
\newcommand{\defdepcoverstatespace}{\ref{def:dep-cover-statespace}}
\newcommand{\defepmtdependence}{\ref{def:epmt-dependence}}

\newcommand{\algoexplore}{\ref{algo:explore}}
\newcommand{\algofindtarget}{\ref{algo:findtarget}}
\newcommand{\algobacktrackeager}{\ref{algo:backtrackeager}}
\newcommand{\algoreschedulepending}{\ref{algo:reschedulepending}}

\newcommand{\linecond}{\ref{line-cond}}
\newcommand{\lineunexploredthread}{\ref{line-add-unexplored-thread}}
\newcommand{\linecopyhb}{\ref{backtrack-line-copy-hb}}
\newcommand{\linefindtarget}{\ref{line-findtarget}}
\newcommand{\linediverging}{\ref{line-diverging}}
\newcommand{\linecandidateone}{\ref{line-candidate-1}}

\newcommand{\linereturnsuccess}{\ref{line-return-success}}
\newcommand{\linepending}{\ref{line-pending}}

\newcommand{\linerecursivefindtarget}{\ref{line-recursive-findtarget}}
\newcommand{\linecallbacktrackeager}{\ref{line-call-backtrackeager}}
\newcommand{\linereturnonhb}{\ref{line-return-on-hb}}
\newcommand{\bactracklinecandidateone}{\ref{bactrack-line-candidate-1}}
\newcommand{\backtracklineaddchoice}{\ref{backtrack-line-add-choice}}
\newcommand{\backtracklineaddall}{\ref{backtrack-line-add-all}}

\let\oldnl\nl
\newcommand{\nonl}{\renewcommand{\nl}{\let\nl\oldnl}}
\makeatother

\SetCommentSty{mycommfont}

\newcommand{\mypar}[1]{\vspace{2mm}\noindent\textbf{\emph{#1}.}}

\lstdefinestyle{code}{
  language={Java},
  basicstyle=\scriptsize\fontfamily{lmvtt},
  numbers=left,                                                                  
  xleftmargin=2.4em,
  frame=single,
  framexleftmargin=2.4em,
  columns=fullflexible,
  numberstyle=\scriptsize,
}

\theoremstyle{definition}
\newtheorem{definition}{Definition}[section]
\newtheorem{example}[definition]{\sc{Example}}
\newtheorem{theorem}[definition]{Theorem}
\newtheorem{lemma}[definition]{Lemma}
\newtheorem{assumption}{Assumption}[definition]

\newtheorem{construction}[definition]{Construction}
\newtheorem{innercustomthm}{Theorem}
\newenvironment{customthm}[1]
  {\renewcommand\theinnercustomthm{#1}\innercustomthm}
  {\endinnercustomthm}
  
\tikzset{
  mybox/.style = {
     minimum width=#1, 
     minimum height=0.5cm, 
     inner sep=3pt, 
     draw},
  mybox/.default=0.5cm,
}
 
\dontprintsemicolon
\incmargin{2pt}

\SetInd{1pt}{0.01pt}

\tikzset{
  mybbox/.style = {
     minimum width=#1, 
     minimum height=0.5cm, 
     inner sep=3pt, 
     very thick, draw},
  mybox/.default=0.5cm,
}

\begin{document}


\title{A Partial Order Reduction Technique for Event-driven Multi-threaded Programs}

\authorinfo{Pallavi Maiya}
           {Indian Institute of Science}
           {pallavih@iisc.ac.in}
\authorinfo{Rahul Gupta}
           {Indian Institute of Science}
           {rahulg@iisc.ac.in}
\authorinfo{Aditya Kanade}
           {Indian Institute of Science}
           {kanade@iisc.ac.in}
\authorinfo{Rupak Majumdar}
           {MPI-SWS}
           {rupak@mpi-sws.org}

\maketitle

\begin{abstract}
Event-driven multi-threaded programming is fast becoming a preferred style of
developing efficient and responsive applications. In this concurrency model,
multiple threads execute concurrently, communicating through shared objects
as well as by posting asynchronous events that are
executed in their order of arrival. In this work, we consider
partial order reduction (POR) for event-driven multi-threaded programs.
The existing POR techniques treat event queues
associated with threads as shared objects and thereby, reorder every pair
of events handled on the same thread even if reordering them does not
lead to different states. We do not treat event queues as shared objects
and propose a new POR technique based on a novel backtracking set called the 
dependence-covering set. 
Events handled by the same thread are reordered
by our POR technique only if necessary. 
We prove that exploring dependence-covering sets suffices to detect all deadlock 
cycles and assertion violations defined over local variables.
To evaluate effectiveness of our POR scheme,
we have implemented a dynamic algorithm to compute dependence-covering sets.
On execution traces obtained from a few Android applications, we demonstrate
that our technique explores many fewer transitions ---often orders of
magnitude fewer--- compared to exploration based on persistent sets,
wherein, event queues are considered as shared objects.
\end{abstract}

\section{Introduction}
\label{sec:intro}

Event-driven multi-threaded programming is fast becoming a preferred style of structuring
concurrent computations in many domains.
In this model, multiple threads execute concurrently, and each thread may
be associated with an event queue.
Threads may post events to each other's event queues, and a thread can
post an event to its own event queue. 
For each thread with an event queue, an event-loop processes the events
from its event queue in the order of their arrival. The event-loop runs the
handler of an event only after the previous handler finishes execution
but interleaved with the execution of all the other threads.
Further, threads can communicate 
through shared objects; even event handlers executing on the same thread 
may share objects.
Event-driven multi-threaded programming is a staple of developing
efficient and responsive smartphone
applications~\cite{mednieks2012programming};
a similar programming model is also used in distributed
message-passing applications, high-performance servers, and many other settings.

Stateless model checking~\cite{Godefroid:1997:MCP:263699.263717} is an
approach to explore the reachable state space of
concurrent programs by exploring different interleavings
systematically but without storing visited states.
In practice, the success of stateless model checking depends crucially
on partial order reduction 
(POR) techniques~\cite{Valmari:1991:SSR:647736.735461,Peled:1993:OOM:647762.735490,
Godefroid:1996:PMV:547238,DBLP:journals/sttt/ClarkeGMP99}.
Stateless search with POR defines an equivalence class on interleavings,
and explores only a representative interleaving from each equivalence
class (called a Mazurkiewicz trace~\cite{DBLP:conf:ac:Mazurkiewicz86}), but still provides
certain formal guarantees w.r.t.\ exploration of 
the complete but possibly much larger state space.
Motivated by the success of model 
checkers based on various POR strategies~\cite{SpinMC,Godefroid:1997:MCP:263699.263717,Flanagan:2005:DPR:1040305.1040315,Sen:2006:AST:2182061.2182094,
Palmer:2007:SDD:1273647.1273657,Coons:2013:BPR:2509136.2509556,Abdulla:2014:ODP:2535838.2535845,Abdulla:2015:SMC:2945565.2945622,
Zhang:2015:DPO:2737924.2737956},
in this work, we propose an effective POR strategy for event-driven multi-threaded programs.

\begin{figure*}[t]
\begin{tabular*}{\textwidth}{l@{\;\;\;}r}
\begin{minipage}{\dimexpr0.45\textwidth-2\tabcolsep}
\centering
\begin{tabular}{@{}r@{\;\;\;\;\;\;\;\;\;\;}l@{\;}l@{\;}l@{}}
 & \tn{t1}{$t_1$}  & \tn{t2}{\;\;\;\;\;\;\;\;\;$t_2$} & \tn{t3}{\;\;\;\;\;\;\;\;\;$t_3$} \\
\hline
\\
$r_1$ && \multicolumn{2}{@{}l}{\tn{n1}{\;\;\;\;\;\;\;\;\;\post{($e_1$)}}} \\\myhline
$r_2$ &&& \tn{n2}{\;\;\;\;\;\;\;\;\;\post{($e_2$)}} \\\myhline
\\\myhline
$r_3$ &  \multicolumn{2}{@{}l}{\tn{n3}{\post{($e_3$)}}} & \\\myhline
\\\myhline
$r_4$ & \multicolumn{2}{@{}l}{\tn{n4}{\texttt{y = 5}}} & \\\myhline
\\\myhline
$r_5$ & \multicolumn{2}{@{}l}{\tn{n5}{\texttt{x = 1}}} & \\\myhline
\\
\end{tabular}
\tikz[remember picture,overlay] \node[fit=(n3), draw] {};
\tikz[remember picture,overlay] \node[fit=(n4), draw] {};
\tikz[remember picture,overlay] \node[fit=(n5), draw] {};
\tikz[remember picture,overlay] \node[left of=n3,xshift=0cm] (e1) {$e_1$};
\tikz[remember picture,overlay] \node (e2) at (n4 -| e1) {$e_2$};
\tikz[remember picture,overlay] \node (e3) at (n5 -| e1) {$e_3$};
\end{minipage}%
&
\begin{minipage}{\dimexpr0.55\textwidth-2\tabcolsep}
\centering
\begin{tikzpicture}[auto,node distance=16 mm,semithick, scale=0.8, transform
shape]
\tikzstyle{stateNode} = [circle,draw=black,thick,inner sep=0pt,minimum 
size=6mm]
\tikzstyle{mystateNode} = [circle,draw=black,thick,dashed,inner sep=0pt,minimum 
size=6mm]
\tikzstyle{lazystateNode} = [circle,fill=yellow,draw=black,thick,inner sep=0pt,minimum 
size=6mm]
\tikzstyle{blankNode} = [thick,inner sep=0pt,minimum size=6mm]

\node[lazystateNode] (s0)  {$s_0$};

\node[lazystateNode] (s1) [below of=s0,yshift=.3cm,xshift=-2cm] 
{$s_1$}
edge [<-, very thick] node [left] {$r_1$} (s0);
\node[blankNode] (q1) [left of=s1, xshift=.65cm] {
\begin{tikzpicture}[node distance=-\pgflinewidth]
\node [mybox, fill=gray!20] (e1) {$e_1$};
\end{tikzpicture}};

\node[lazystateNode] (s2) [below of=s1,yshift=.3cm,xshift=-1cm] 
{$s_2$}
edge [<-, very thick] node [left] {$r_2$} (s1);
\node[blankNode] (q2) [left of=s2, xshift=.37cm] {
\begin{tikzpicture}[node distance=-\pgflinewidth]
\node [mybox, fill=gray!20] (e1) {$e_1$};
\node [mybox,fill=gray!20,right=of e1] (e2) {$e_2$};
\end{tikzpicture}};

\node[lazystateNode] (s3) [below of=s2,yshift=.3cm] 
{$s_3$} 
edge [<-, very thick] node [left] {$r_3$} (s2);
\node[blankNode] (q3) [left of=s3, xshift=.37cm] {
\begin{tikzpicture}[node distance=-\pgflinewidth]
\node [mybox, fill=gray!20] (e2) {$e_2$};
\node [mybox,fill=gray!20,right=of e2] (e3) {$e_3$};
\end{tikzpicture}};

\node[lazystateNode] (s4) [below of=s3,yshift=.3cm,xshift=1cm] 
{$s_4$} 
edge [<-, very thick] node [left] {$r_4$} (s3);
\node[blankNode] (q4) [left of=s4, xshift=.65cm] {
\begin{tikzpicture}[node distance=-\pgflinewidth]
\node [mybox, fill=gray!20] (e3) {$e_3$};
\end{tikzpicture}};

\node[lazystateNode] (s5) [below of=s4,yshift=.3cm,xshift=2cm] 
{$s_5$} 
edge [<-, very thick] node [left,xshift=-.1cm] {$r_5$} (s4);

\node[stateNode] (s6) [below of=s1,yshift=.3cm,xshift=1.5cm] 
{$s_6$}
edge [<-] node [left] {$r_3$} (s1);
\node[blankNode] (q6) [left of=s6, xshift=.65cm] {
\begin{tikzpicture}[node distance=-\pgflinewidth]
\node [mybox, fill=gray!20] (e3) {$e_3$};
\end{tikzpicture}};

\node[stateNode] (s7) [below of=s6,yshift=.3cm] 
{$s_7$} edge [<-] node [left] {$r_2$} (s6);
\node[blankNode] (q7) [left of=s7, xshift=.37cm] {
\begin{tikzpicture}[node distance=-\pgflinewidth]
\node [mybox, fill=gray!20] (e3) {$e_3$};
\node [mybox,fill=gray!20,right=of e3] (e2) {$e_2$};
\end{tikzpicture}};

\node[stateNode] (s8) [below of=s7,yshift=.3cm,xshift=1.5cm] 
{$s_8$}
edge [<-] node [left] {$r_5$} (s7)
edge [->] node [left] {$r_4$} (s5);
\node[blankNode] (q8) [left of=s8, xshift=.65cm] {
\begin{tikzpicture}[node distance=-\pgflinewidth]
\node [mybox, fill=gray!20] (e2) {$e_2$};
\end{tikzpicture}};

\node[mystateNode] (s9) [below of=s6,yshift=.3cm,xshift=1.5cm] 
{$s_9$}
edge [<-,dashed] node [left] {$r_5$} (s6)
edge [->,dashed] node [left] {$r_2$} (s8);

\node[stateNode] (s10) [below of=s0,yshift=.3cm,xshift=1.7cm] 
{$s_{10}$}
edge [<-] node [right] {$r_2$} (s0);
\node[blankNode] (q10) [right of=s10, xshift=-.65cm] {
\begin{tikzpicture}[node distance=-\pgflinewidth]
\node [mybox, fill=gray!20] (e2) {$e_2$};
\end{tikzpicture}};

\node[stateNode] (s11) [below of=s10,yshift=.3cm,xshift=1.5cm] 
{$s_{11}$}
edge [<-] node [right] {$r_1$} (s10);
\node[blankNode] (q11) [right of=s11, xshift=-.37cm] {
\begin{tikzpicture}[node distance=-\pgflinewidth]
\node [mybox, fill=gray!20] (e2) {$e_2$};
\node [mybox,fill=gray!20,right=of e2] (e1) {$e_1$};
\end{tikzpicture}};

\node[stateNode] (s12) [below of=s11,yshift=.3cm] 
{$s_{12}$}
edge [<-] node [right] {$r_4$} (s11);
\node[blankNode] (q12) [right of=s12, xshift=-.65cm] {
\begin{tikzpicture}[node distance=-\pgflinewidth]
\node [mybox, fill=gray!20] (e1) {$e_1$};
\end{tikzpicture}};

\node[stateNode] (s13) [below of=s12,yshift=.3cm,xshift=-1cm] 
{$s_{13}$}
edge [<-] node [right] {$r_3$} (s12)
edge [->] node [right,xshift=.2cm] {$r_5$} (s5);
\node[blankNode] (q13) [right of=s13, xshift=-.65cm] {
\begin{tikzpicture}[node distance=-\pgflinewidth]
\node [mybox, fill=gray!20] (e3) {$e_3$};
\end{tikzpicture}};

\node[mystateNode] (s14) [below of=s10,yshift=.3cm] 
{$s_{14}$}
edge [<-,dashed] node [right] {$r_4$} (s10)
edge [->,dashed] node [right] {$r_1$} (s12);
\end{tikzpicture}
\end{minipage}
\\
\begin{minipage}[t]{\dimexpr0.45\textwidth-2\tabcolsep}
\captionof{figure}{A partial trace of an event-driven program.}
\label{fig:trace-por-motivation}
\end{minipage}
&
\begin{minipage}[t]{\dimexpr0.55\textwidth-2\tabcolsep}
\captionof{figure}{
The state space reachable through all valid permutations of operations
in the trace given in Figure~\ref{fig:trace-por-motivation}. The leftmost event in an event queue is
the front of the queue.}
\label{fig:state-transition-motivation}
\end{minipage}
\end{tabular*}%
\end{figure*}

\mypar{Motivating example}
We first show why existing POR techniques may not be very effective in the combined 
model of threads and events.
Consider a partial execution trace of an event-driven program shown in 
Figure~\ref{fig:trace-por-motivation}. The operations are executed from top to bottom.
The operations in the trace are labeled $r_1$ to $r_5$
and those belonging to the same event handler are enclosed within a box
labeled with the corresponding event. These operations are executed by the 
threads $t_1$, $t_2$ or $t_3$. Figure~\ref{fig:trace-por-motivation} enumerates
all the operations executed by a thread on a vertical line below the thread. 
An operation \post{($e$)} under thread $t$ denotes the enqueuing of an event $e$
by thread $t$. The destination event queue can be identified by mapping the event 
posted with the corresponding event label against an event handler. For example, 
the operation $r_1$ executed by thread $t_2$ posts an event $e_1$ to thread $t_1$'s event queue. 
In this trace, threads $t_2$ and $t_3$ respectively post events 
$e_1$ and $e_2$ to thread $t_1$'s event queue.  
The event handler of $e_1$ in turn posts an event $e_3$ to $t_1$'s queue.
The event handlers of $e_2$ and $e_3$ respectively write to shared variables
\texttt{y} and \texttt{x}.

Figure~\ref{fig:state-transition-motivation} shows the state space reachable through
all valid permutations of operations in the trace in Figure~\ref{fig:trace-por-motivation}. 
Each node indicates a state of the program. An edge is labeled with an operation and
indicates the state transition due to that operation. 
The interleaving corresponding to the trace in Figure~\ref{fig:trace-por-motivation} 
is highlighted with bold lines and shaded states. 
For illustration purposes,
we explicitly show the contents of the event queue of thread $t_1$ at 
some states. Events in a queue are ordered from left to right.
Pictorially, an event is removed from the queue when it is dequeued for handling.

Existing POR techniques
(e.g.~\cite{Godefroid:1996:PMV:547238,Flanagan:2005:DPR:1040305.1040315,Sen:2006:AST:2182061.2182094,
Tasharofi:2012:TND:2366649.2366663,Abdulla:2014:ODP:2535838.2535845})
recognize that $r_2$ and $r_5$ (also $r_1$ and $r_4$) are independent
(or non-interfering) and that it is sufficient to explore 
any one of them at state $s_6$ (respectively, $s_{10}$). 
The dashed edges indicate the unexplored transitions. 
However, existing POR-based model checkers will explore all other
states and transitions. 
Since no two handlers executed on the thread $t_1$ modify a common object, all 
the interleavings reach the same state $s_5$. Thus, the existing techniques 
explore two redundant interleavings. This happens because these techniques 
treat \emph{event queues as shared objects} and so, mark any two \texttt{post}
operations that enqueue events to the event queue of the same thread as dependent.
Consequently, they explore both $r_1$ and $r_2$ at state $s_0$, and
$r_2$ and $r_3$ at state $s_1$. 
These result in unnecessary reorderings of events.

More generally, if there are $n$ events posted to an event queue,
these techniques may explore $O(n!)$ permutations among them, even if exploring only 
one of them may be sufficient. 
Therefore, a POR technique that can avoid redundant event
orderings can be significantly more scalable. 
We exploit this observation. For the state space in Figure~\ref{fig:state-transition-motivation}, 
our approach explores only the initial trace (the leftmost interleaving) and thus visits
substantially fewer states and transitions compared to existing techniques.

\begin{figure*}[t]
\begin{tabular*}{\textwidth}{l@{\;\;\;\;\;\;}r}
\begin{minipage}{\dimexpr0.5\textwidth-2\tabcolsep}
\centering
\begin{tabular}{@{}r@{\;\;\;\;\;\;\;\;\;\;}l@{\;}l@{\;}l@{\;}l@{}}
 & \tn{t1}{$t_1$}  & \tn{t2}{\;\;\;\;\;\;\;$t_2$} & \tn{t3}{\;\;\;\;\;\;\;$t_3$} & \tn{t4}{\;\;\;\;\;\;\;$t_4$} \\
\hline
\\
$r_1$ && \multicolumn{3}{@{}l}{\tn{n1}{\;\;\;\;\;\;\;\post{($e_1$)}}} \\\myhline
$r_2$ &&& \multicolumn{2}{@{}l}{\tn{n2}{\;\;\;\;\;\;\;\post{($e_2$)}}} \\\myhline
\\\myhline
$r_3$ &  \multicolumn{2}{@{}l}{\tn{n3}{\texttt{x = 1}}} & \\\myhline
\\\myhline
$r_4$ & \multicolumn{2}{@{}l}{\tn{n4}{\fork{($t_4$)}}} & \\\myhline
$r_5$ &&&& \tn{n5}{\;\;\;\;\;\;\;\tinit{($t_4$)}} \\\myhline
$r_6$ &&&& \tn{n6}{\;\;\;\;\;\;\;\texttt{x = 2}} \\\myhline
\\
\end{tabular}
\tikz[remember picture,overlay] \node[fit=(n3), draw] {};
\tikz[remember picture,overlay] \node[fit=(n4), draw] {};
\tikz[remember picture,overlay] \node[left of=n3,xshift=0cm] (e1) {$e_1$};
\tikz[remember picture,overlay] \node[left of=n4,xshift=-.3cm] (e2) {$e_2$};
\end{minipage}%
&
\begin{minipage}{\dimexpr0.5\textwidth-2\tabcolsep}
\centering
\begin{tikzpicture}[auto,node distance=16 mm,semithick, scale=.8, transform
shape]
\tikzstyle{stateNode} = [circle,draw=black,thick,inner sep=0pt,minimum 
size=6mm]
\tikzstyle{dporstateNode} = [circle,fill=yellow,draw=black,thick,inner sep=0pt,minimum 
size=6mm]
\tikzstyle{blankNode} = [thick,inner sep=0pt,minimum size=6mm]

\node[dporstateNode] (s0)  {$s_0$};
\node[dporstateNode] (s1) [below of=s0,yshift=.3cm,xshift=-1.5cm] 
{$s_1$} edge [<-, very thick] node [left] {$r_1$} (s0);
\node[dporstateNode] (s2) [below of=s1,yshift=.3cm] 
{$s_2$} edge [<-, very thick] node [left] {$r_2$} (s1);
\node[dporstateNode] (s3) [below of=s2,yshift=.3cm] 
{$s_3$} edge [<-, very thick] node [left] {\scalebox{2}{$r_3$}} (s2);
\node[dporstateNode] (s4) [below of=s3,yshift=.3cm] 
{$s_4$} edge [<-, very thick] node [left] {$r_4$} (s3);
\node[dporstateNode] (s5) [below of=s4,yshift=.3cm,xshift=1.5cm] 
{$s_5$} edge [<-, very thick] node [left] {$r_5$} (s4);
\node[dporstateNode] (s6) [below of=s5,yshift=.3cm] 
{$s_6$} edge [<-, very thick] node [left] {\scalebox{2}{$r_6$}} (s5);

\node[stateNode] (s7) [below of=s0,yshift=.3cm,xshift=1.5cm] 
{$s_7$} edge [<-] node [right] {$r_2$} (s0);
\node[stateNode] (s8) [below of=s7,yshift=.3cm] 
{$s_8$} edge [<-] node [right] {$r_1$} (s7);
\node[stateNode] (s9) [below of=s8,yshift=.3cm] 
{$s_9$} edge [<-] node [right] {$r_4$} (s8);
\node[stateNode] (s10) [below of=s9,yshift=.3cm] 
{$s_{10}$} edge [<-] node [right] {$r_5$} (s9)
edge [->] node [right] {$r_3$} (s5);

\node[stateNode] (s11) [below of=s10,yshift=.3cm,xshift=1.5cm] 
{$s_{11}$} edge [<-] node [right] {\scalebox{2}{$r_6$}} (s10);
\node[stateNode] (s12) [below of=s11,yshift=.3cm] 
{$s_{12}$} edge [<-] node [right] {\scalebox{2}{$r_3$}} (s11);

\node[blankNode] (b2) [right of=s0, xshift=1cm] 
{\texttt{PS:$\{r_1\}$} $\;$ \texttt{DCS:$\{r_1,r_2\}$}};
\node[blankNode] (b3) [right of=s10, xshift=.1cm] 
{\texttt{DCS:$\{r_3,r_6\}$}};

\node[blankNode] (q1) [left of=s1, xshift=.6cm] {
\begin{tikzpicture}[node distance=-\pgflinewidth]
\node [mybox, fill=gray!20] (e1) {$e_1$};
\end{tikzpicture} };
\node[blankNode] (q2) [left of=s2, xshift=.28cm] {
\begin{tikzpicture}[node distance=-\pgflinewidth]
\node [mybbox, fill=gray!40] (e1) {$e_1$};
\node [mybbox,fill=gray!40,right=of e1] (e2) {$e_2$};
\end{tikzpicture} };
\node[blankNode] (q3) [left of=s3, xshift=.6cm] {
\begin{tikzpicture}[node distance=-\pgflinewidth]
\node [mybox, fill=gray!20] (e2) {$e_2$};
\end{tikzpicture} };

\node[blankNode] (q7) [right of=s7, xshift=-.6cm] {
\begin{tikzpicture}[node distance=-\pgflinewidth]
\node [mybox, fill=gray!20] (e2) {$e_2$};
\end{tikzpicture} };
\node[blankNode] (q8) [right of=s8, xshift=-.28cm] {
\begin{tikzpicture}[node distance=-\pgflinewidth]
\node [mybbox, fill=gray!40] (e2) {$e_2$};
\node [mybbox,fill=gray!40,right=of e2] (e1) {$e_1$};
\end{tikzpicture} };
\node[blankNode] (q9) [right of=s9, xshift=-.6cm] {
\begin{tikzpicture}[node distance=-\pgflinewidth]
\node [mybox, fill=gray!20] (e1) {$e_1$};
\end{tikzpicture} };
\node[blankNode] (q10) [left of=s10, xshift=.6cm] {
\begin{tikzpicture}[node distance=-\pgflinewidth]
\node [mybox, fill=gray!20] (e1) {$e_1$};
\end{tikzpicture} };
\node[blankNode] (q11) [left of=s11, xshift=.6cm] {
\begin{tikzpicture}[node distance=-\pgflinewidth]
\node [mybox, fill=gray!20] (e1) {$e_1$};
\end{tikzpicture} };
\end{tikzpicture}

\end{minipage}
\\
\begin{minipage}[t]{\dimexpr0.5\textwidth-2\tabcolsep}
\captionof{figure}{A partial trace $w$ of an event-driven program involving
a multi-threaded dependence.}
\label{fig:motivation-mt}
\end{minipage}
&
\begin{minipage}[t]{\dimexpr0.5\textwidth-2\tabcolsep}
\captionof{figure}{A partial state space for some valid permutations of transitions 
in the trace given in Figure~\ref{fig:motivation-mt}.}
\label{fig:motivation-statespace-mt}
\end{minipage}
\end{tabular*}%
\end{figure*}

\mypar{Our approach}
Realizing a partial order reduction technique effective for event-driven programs
requires reviewing of various elements of POR and redesigning them to be suitable 
in the context of event-driven programs. 
To realize the reduction outlined through motivating example, we do \emph{not} consider event queues 
as shared objects. Equivalently, we treat a pair of \texttt{post}s even to
the same thread as \emph{independent}. The main question then is ``How to determine which events to reorder
and how to reorder them selectively?''. Surely, if two handlers 
executing on the same thread contain dependent transitions
then we should reorder their \texttt{post} operations, 
but this is not enough. To see this, consider a partial trace $w$
shown in Figure~\ref{fig:motivation-mt}.
The transitions $r_3$ and $r_6$ belong to two different threads and
are dependent as they write to the same variable. 
Figure~\ref{fig:motivation-statespace-mt} shows a
a partial state space obtained by different orderings of $r_3$ and $r_6$.
The contents of thread $t_1$'s event queue are shown next to each state, whenever
the queue is non-empty. 
As can be seen in the rightmost interleaving,
executing $r_6$ before $r_3$ requires posting the event $e_2$ \emph{before} $e_1$
even though their handlers do not have dependent transitions.
A state space exploration starting with sequence $w$ has to reorder 
$e_1$ and $e_2$ so as to explore a different ordering of $r_3$ and $r_6$
than that explored by $w$. 
Thus, operations posting events to the same thread may have
to be reordered even to reorder some multi-threaded dependences!
Our first contribution is to define a dependence relation that
captures both single-threaded as well as multi-threaded dependences.

We now discuss the implications of treating \texttt{post}s as
independent and only selectively reordering them.
For multi-threaded programs, or when \texttt{post}s are considered dependent, reordering
a pair of adjacent independent transitions in a transition sequence does not affect the reachable state.
Hence, the existing dependence relation~\cite{Godefroid:1996:PMV:547238} induces 
equivalence classes where transition sequences differing only in the order of 
executing independent transitions are in the same Mazurkiewicz trace~\cite{DBLP:conf:ac:Mazurkiewicz86}. 
However, our new dependence relation (where \texttt{post}s are considered independent)
may not induce Mazurkiewicz traces on an event-driven program. 
One reason is that reordering \texttt{post}s to the same thread affects the order of
execution of the corresponding handlers. 
If the handlers contain dependent transitions, it affects the reachable state. 
Second, one cannot rule out
the possibility of \emph{new} transitions (not present in the given transition sequence)
being pulled in when independent \texttt{post}s are reordered,
which is not admissible in a Mazurkiewicz trace.
We elaborate on this in Section~\ref{sec:dcset}.

\sloppy{
Our second contribution is to define a novel notion of \emph{dependence-covering sequence}
to provide the necessary theoretical foundation to reason about
reordering \texttt{post}s selectively.
Intuitively, a transition sequence $u$ is a dependence-covering sequence of a transition
sequence $v$ if the relative ordering of all the pairs of dependent transitions in $v$ is
preserved in $u$. While this sounds similar to the property of any pair of transition
sequences in the same Mazurkiewicz trace, the constraints imposed on a dependence-covering
sequence are more relaxed (as will be formalized in Definition~\ref{def:dep-covering}),
making it suitable to achieve better reductions.
For instance, $u$ is permitted to have \emph{new} transitions, that is,
transitions that are not in $v$, under certain conditions.
}

Given a notion of POR, a model checking algorithm such as
DPOR~\cite{Flanagan:2005:DPR:1040305.1040315} uses persistent
sets~\cite{Godefroid:1996:PMV:547238} to structure the state space
exploration to only explore representative transition sequences from
each Mazurkiewicz trace.
As we show now, DPOR based on persistent sets 
is unsound when used in conjunction with the dependence relation in
which \texttt{post}s are independent. 
Let us revisit the state space given in Figure~\ref{fig:motivation-statespace-mt}.
Assume DPOR to explore this state space starting with the leftmost branch of the state 
space in Figure~\ref{fig:motivation-statespace-mt}, which corresponds to sequence $w$
shown in Figure~\ref{fig:motivation-mt}. Then, DPOR identifies the set $\{r_1\}$ as persistent in state $s_0$, 
because exploring any transition other than $r_1$ from $s_0$ does not 
hit a transition dependent with $r_1$. This set is tagged as
\texttt{PS} in Figure~\ref{fig:motivation-statespace-mt}. However, a selective
exploration using this set explores only one ordering between $r_3$ and $r_6$,
even though the two orderings are not equivalent. 

Our third contribution is the notion of \emph{dependence-covering sets}
as an alternative to persistent sets.
A set of transitions $L$ at a state $s$ is said to be dependence-covering (formalized 
in Definition~\ref{def:dc-set}) if a dependence-covering sequence $u$ starting 
with some transition in $L$ can be explored for any sequence $v$ executed from $s$. 
We prove that selective state-space exploration based on 
dependence-covering sets is sufficient to detect all deadlock cycles and 
violations of assertions over local variables.
The dependence-covering sets at certain states 
are marked in Figure~\ref{fig:motivation-statespace-mt} as \texttt{DCS}. 
In contrast to \texttt{PS}, \texttt{DCS} at state $s_0$ contains \emph{both}
$r_1$ and $r_2$. The set $\{r_1,r_2\}$ at $s_0$ is a 
dependence-covering set because for any transition sequence $v$ starting 
from $s_0$, there exists a dependence-covering sequence $u$ 
starting with a transition in $\{r_1,r_2\}$. 
Let $v$ be the transition sequence 
along the rightmost interleaving in Figure~\ref{fig:motivation-statespace-mt}.
The sequence $w$ (the leftmost interleaving) 
is not a dependence-covering sequence of $v$ since 
the dependent transitions $r_3$ and $r_6$ appear in a different order. 
We therefore require $r_2$ to be explored at $s_0$. 
Note that, $\{r_2\}$ is another dependence-covering set at $s_0$ as both the orderings of 
dependent transitions $r_3$ and $r_6$ can be explored from $s_{10}$ reached 
on exploring $r_2$. 

Our final contribution is a dynamic algorithm called \emdpor\ to compute 
dependence-covering sets. EM refers to the \underline{E}vent-driven \underline{M}ulti-threaded 
concurrency model. 
\emdpor\ follows the DFS based exploration strategy of DPOR~\cite{Flanagan:2005:DPR:1040305.1040315} 
but the key steps of our algorithm are different.
In particular,  \emdpor\ incorporates several non-trivial steps
(1)~to reason about both multi-threaded dependences as well as dependent transitions
from different event handlers on the same thread (single-threaded dependences),
and (2)~to identify events for selective reordering and infer appropriate backtracking choices
to achieve the reordering.
We have implemented and evaluated this adaptation in a proof-
of-concept model checker. Further, we have provided a sketch outlining the proof of correctness
of our algorithm in Appendix~\ref{app:emdpor-proofs}.

We now briefly explain how EM-DPOR computes the dependence-covering sets
and explores the state space shown in Figure~\ref{fig:motivation-statespace-mt}
starting with sequence $w$.
On exploring a prefix of sequence $w$ and reaching state $s_5$, 
EM-DPOR identifies $r_6$ to be dependent with $r_3$ and hence tries to reorder $r_6$ 
w.r.t. $r_3$. It does so by attempting to 
execute transitions that happen before $r_6$ prior to $r_3$, ultimately leading to
the execution of $r_6$ prior to $r_3$. When attempting to
compute backtracking choices at state $s_2$ (the state where $r_3$ is explored)
to reorder $r_3$ and $r_6$, EM-DPOR finds 
$r_4$ to happens before $r_6$. However,
$r_4$ is not enabled at $s_2$ because both $r_3$ and $r_4$ execute on
\emph{the same thread} $t_1$, and $r_4$ is a transition of 
the handler of $e_2$ while $e_1$ is at the front of the queue
(see the event queue shown at $s_2$ in Figure~\ref{fig:motivation-statespace-mt}).
Because EM-DPOR is aware of the event-driven semantics and knows that
$r_3$ and $r_4$ come from handlers of two different events 
$e_1$ and $e_2$, it attempts to reorder the events themselves.
We call this a step to \emph{reschedule pending events} because
$e_2$ is pending in the queue of the thread $t_1$ at $s_2$.
EM-DPOR then starts \emph{another} backward search to identify the backtracking choices
that can reorder $e_1$ and $e_2$. It identifies that
the corresponding \post{} operations $r_1$ and $r_2$ can be reordered
to do so. It therefore adds $r_2$ to the backtracking set at $s_0$
(from Figure~\ref{fig:motivation-statespace-mt} $r_1$ is already in the 
backtracking set at $s_0$ since the exploration started with $w$),
exploring which leads to $s_8$ where event $e_2$ \emph{precedes}
$e_1$ in the event queue as required. 
EM-DPOR then reaches state $s_{10}$ where $r_3$ and $r_6$ are \emph{co-enabled}.
Being dependent, EM-DPOR explores both the ordering between $r_3$ and $r_6$
from $s_{10}$.
Note that even while considering only $r_3$ and $r_6$ from different threads as dependent, 
EM-DPOR is able to identify a  seemingly unrelated pair of \texttt{post}s at $r_1$ and $r_2$ for reordering.

\mypar{Experiments}
We have evaluated \emdpor\ on Android applications which are a class of
multi-threaded event-driven programs.
We have implemented a proof-of-concept model checking framework called EM-Explorer which simulates
the non-deterministic behaviour exhibited by Android applications given individual
execution traces.
We implemented \emdpor\ which performs a selective state-space exploration based on dependence-covering sets, in EM-Explorer.
For comparison, we also implemented DPOR which performs exploration based on 
persistent sets, where \texttt{post}s to the same thread are considered dependent. 
We performed experiments on traces obtained from $5$ Android applications.
Our results demonstrate that our POR technique explores many fewer transitions 
---often orders of magnitude fewer--- compared to using persistent sets.

\section{Formalization}
\label{sec:definitions}
We now formalize our notion of partial order reduction for event-driven
programs. Some of the definitions below follow the conventions 
in~\cite{Flanagan:2005:DPR:1040305.1040315}.
Any reference to persistent sets henceforth, assumes usage of the dependence 
relation defined in~\cite{Godefroid:1996:PMV:547238} as it is, 
which marks two \texttt{post} operations to the same event queue as dependent.

\subsection{Transition System}
\label{sec:transition-system-por}

We consider an event-driven multi-threaded \emph{program} $A$ which has the
usual sequential and multi-threaded operations such as assignments,
conditionals, synchronization through locks and thread creation.
In addition, the operation \post{($t_1,e,t_2$)} posts
an asynchronous event $e$ from the source thread $t_1$ 
to (the event queue of) a destination thread $t_2$. However in the execution traces 
given in the paper, we omit the source and destination threads of \post{} operation
(\eg\ Figure~\ref{fig:trace-por-motivation} and \ref{fig:motivation-mt}) when apparent from the diagram. 
Each event has a handler which 
runs to completion on the thread to whose event queue the event is posted.
However, the event handler of one thread may interleave with 
operations of other threads. Operation \abegin{($e$)} denotes the dequeuing of 
an event $e$, and \aend{($e$)} indicates the completion of execution of an event
handler. We consider \abegin{} and \aend{} as the first and the last operation
of an event handler. In the traces considered in this paper, all the operations 
belonging to the same event handler are grouped inside a box (\eg\ 
Figure~\ref{fig:trace-por-motivation} and \ref{fig:motivation-mt}). The operations 
\abegin{} and \aend{} are omitted but implicitly assumed as the first and the last
operation inside the box.
We omit the formal syntax and semantics of various operations relevant in the context
of a multi-threaded event-driven program; they can be found in \cite{Maiya:2014:RDA:2594291.2594311}.

An operation is \emph{visible} if it accesses an object shared between at least two threads or
two event handlers (possibly running on the same thread). The first operation 
(\texttt{deq}) of an event handler is also considered a visible operation.
All other operations are \emph{invisible}.

The \emph{local state of an event handler}  
is a valuation of the stack and the variables or
heap objects that are modified only within the event handler.
The local state of a thread is the local state of the currently executing event handler.
If a handler running on a thread has finished executing, but the thread has not
started executing the next handler (if any), we say that the
thread is \emph{idle}; the local state of an idle thread is undefined. 
A \emph{global state} of the program $A$ is a valuation to the variables
and heap objects that are accessed by multiple threads or multiple handlers.
Even though event queues are shared objects, we 
do not consider them in the global state (as defined above). 
Instead, we define a \emph{queue state of a thread} as 
an ordered sequence of events that have been posted to its event queue
but are yet to be handled. This separation allows us to analyze asynchronous
posts more precisely. 
Event queues are FIFO queues with unbounded capacity, that is, a \texttt{post} operation
never blocks.
For simplicity, we assume that every thread is associated with an event
queue. If a thread does not have an event queue in reality then its 
state is determined by the default procedure that runs on it
in response to some initial event, and no other events are enqueued to
its event queue subsequently.

Let $L$, $G$ and $Q$ be the set of all local states, global states and queue
states respectively. Let \Threads\ be the set of all threads in $A$. Then, a \emph{state} 
$s$ of an event-driven program $A$ is a triple $(l,g,q)$ where 
(1)~$l$ is a partial map from \Threads\ to $L$, 
(2)~$g$ is a global state and
(3)~$q$ is a total map from \Threads\ to $Q$. 
A \emph{transition} by a thread $t$ updates the state of $A$ by performing one
visible operation followed by a finite sequence of invisible
operations ending just before the next visible operation; 
all of which are executed on $t$. 
We identify a transition by its visible operation,
e.g., we say 
``\texttt{post} operation'' to mean a transition whose first 
operation is a \texttt{post}. 
Let $R$ be the set of all transitions in $A$. 
A transition $r_{t,\ell}$ of a thread $t$ at its 
local state $\ell$ is a partial function, $r_{t,\ell}: G \times Q \mapsto L \times G \times Q$.
A transition $r_{t,\ell} \in R$ is enabled at a state $s = (l,g,q)$ if $\ell = l(t)$
and $r_{t,\ell}(g,q)$ is defined. We may use $r_{t,\ell}(s)$ to 
denote application of a transition $r_{t,\ell}$, instead of the more precise use
$r_{t,\ell}(g,q)$.  
The first transition of the handler of an event $e$ enqueued to a thread $t$ is 
\emph{enabled} at a state $s$, if $e$ is at the front of $t$'s queue
at $s$ and $t$ is idle in $s$. 
We assume that if a transition is defined for a state then it
deterministically maps the state to a successor state. 

We formalize the state space of $A$ as a \emph{transition system} $\mathcal{S}_G = (\mathcal{S}, s_{\mathit{init}}, \Delta)$,
where $\mathcal{S}$ is the set of all states, $s_{\mathit{init}} \in \mathcal{S}$ is the initial state, and 
$\Delta \subseteq \mathcal{S} \times \mathcal{S}$ is the transition relation such that
$(s,s') \in \Delta$ iff $\exists r \in R$ and $s' = r(s)$.
We also use $s \in \mathcal{S}_G$ instead of $s \in \mathcal{S}$.
%
Two transitions $r_1$ and $r_2$ \emph{may be co-enabled} if there
may exist some state $s \in \mathcal{S}$ where they both are enabled.
Two events $e$ and $e'$ handled on the same thread $t$ 
\emph{may be reordered} if there exist states $s,s' \in \mathcal{S}$ such that 
$s = (l,g,q)$, $s' = (l',g',q')$, $q(t) = e \cdot w \cdot e' \cdot w'$ and 
$q'(t) = e' \cdot v \cdot e \cdot v'$.
In Figure~\ref{fig:state-transition-motivation}, events $e_1$ and $e_2$
may be reordered but not $e_1$ and $e_3$.

For simplicity, we assume that all threads and events in $A$ have unique IDs. 
We also assume that the state space is \emph{finite} 
and \emph{acyclic}.
This is a standard assumption for stateless model checking~\cite{Flanagan:2005:DPR:1040305.1040315}.
The transition system $\mathcal{S}_G$ collapses invisible
operations and is thus already reduced when compared to the transition
system in which even invisible operations are considered as
separate transitions. A transition system of this form is sufficient
for detecting deadlocks and assertion violations~\cite{Godefroid:1997:MCP:263699.263717}.
We note that the event dispatch semantics can be diverse in general.
For example, Android applications permit posting an 
event with a timeout or posting a specific event to the front of the queue. 
We over-approximate the effect of posting with timeout by forking
a new thread which does the \texttt{post} non-deterministically but
do not address other variants in this work. 
We leave a more general POR approach that allows such variants to
event dispatch, to future work.

\paragraph{Notation.}
Let $\nextop(s,t)$ give the next transition of a 
thread $t$ in a state $s$. Let $\threadof(r)$ return
the thread executing a transition $r$. 
If $r$ executes in the handler of an event $e$
on thread $t$ then the \emph{task} of $r$ is
$\taskof(r) = (t,e)$.
A transition $r$ on a thread $t$ is 
\emph{blocked} at a state $s$ if $r = \nextop(s,t)$ and 
$r$ is not enabled in $s$.
We assume that only visible operations may block.
Function $\nextTrans(s)$ gives the set of next transitions of all threads 
at state $s$. 
For a transition sequence $w : r_1 . r_2 \ldots r_n$ in $\mathcal{S}_G$,
let $\domof(w) = \{1,\ldots,n\}$. 
Functions $\getBegin(w,e)$ and $\getEnd(w,e)$ respectively return 
the indices of the first and the last transitions of an event $e$'s handler in $w$,
provided they belong to $w$.
For a transition $r$, $index(w,r)$ gives the position of $r$ in $w$. 

\paragraph{Deadlock cycles and assertion violations.} 
A pair $\langle DC,\rho \rangle$ in a state $s \in \mathcal{S}$ is said 
to form a \emph{deadlock cycle} if $DC \subseteq \nextTrans(s)$ is a set 
of $n$ transitions blocked in $s$, and $\rho$
is a one-to-one map from $[1,n]$ to $DC$
such that each $\rho(i) \in DC$, $i \in [1,n]$,
is blocked by some transition on a thread 
$t_{i+1} = \threadof(\rho(i+1))$ and may be enabled only by a transition on 
$t_{i+1}$, and the transition $\rho(n) \in DC$ is blocked and may be enabled 
by two different transitions of thread $t_1 = \threadof(\rho(1))$. A state $s$ in 
$\mathcal{S}_G$ is a \emph{deadlock} state if all the 
threads are blocked in $s$ due to a deadlock cycle.

An \emph{assertion} $\alpha$ is a predicate over local variables of an event handler 
and is considered visible. 
A state $s$ \emph{violates} an assertion
$\alpha$ if $\alpha$ is enabled at $s$ and evaluates to \textit{false}.

\subsection{Dependence Relation}
\label{sec:relations}

The notion of dependence between transitions is well-understood for 
multi-threaded programs. It extends naturally to event-driven programs if
event queues are considered as shared objects, thereby, marking
two \texttt{post}s to the same event queue as dependent.
To enable more reductions, we define
an alternative notion in which two \texttt{post} operations to the same event queue 
are \emph{not} considered dependent. One reason to selectively
reorder events posted to a thread is if their handlers contain dependent transitions.
This requires a new notion of
dependence between transitions of event handlers
executing on the same thread, which we refer to as \emph{single-threaded dependence}.

In order to explicate single-threaded dependences,
we first define an event-parallel transition system 
which over-approximates the transition system $\mathcal{S}_G$.
The \emph{event-parallel transition system} $\mathcal{P}_G$ of a 
program $A$ is a triple $(\mathcal{S}_P,s_{init},\Delta_P)$. 
In contrast to the transition system $\mathcal{S}_G = (\mathcal{S},s_{init},\Delta)$ of Section~\ref{sec:transition-system-por}
where events are dispatched in their order of arrival and execute till completion, 
a thread with an event queue in $\mathcal{P}_G$
removes \emph{any} event in its queue and spawns a fresh thread to execute its handler. This enables
concurrent execution of handlers of events posted to the same thread. Rest of the semantics remains the same.
Let \Threads\ and $\Threads_P$ be the sets of all threads in
$\mathcal{S}_G$ and $\mathcal{P}_G$ respectively. 
For each state $(l,g,q) \in \mathcal{S}$, there exists a state $(l',g',q') \in \mathcal{S}_P$
such that (1) for each thread $t \in \Threads$, if $l(t)$ is defined then there exists
a thread $t' \in \Threads_P$ where $l'(t') = l(t)$,
(2) $g = g'$, and (3) for each thread $t \in \Threads$, 
$q(t) = q'(t)$.
Let $R_P$ be the set of transitions in $\mathcal{P}_G$ and
$ep: R \rightarrow R_P$ be a total function which maps a transition $r_{t,\ell} \in R$ 
to an equivalent transition $r'_{t',\ell'}$ such that $\ell = \ell'$ and
either $t' = t$ or $t'$ is a fresh thread spawned by $t$ in $\mathcal{P}_G$ to
handle the event to whose handler $r_{t,\ell}$ belongs in $\mathcal{S}_G$.

\begin{figure*}[t]
\begin{tabular*}{\textwidth}{l@{\;\;\;}l}
\begin{minipage}{\dimexpr0.5\textwidth-2\tabcolsep}
\centering
\begin{tikzpicture}[auto,node distance=16 mm,semithick, scale=0.7, transform
shape]
\tikzstyle{sgNode} = [circle,draw=black,fill=yellow,thick,inner sep=0pt,minimum 
size=6mm]
\tikzstyle{pgNode} = [circle,draw=black,thick,inner sep=0pt,minimum 
size=6mm]
\tikzstyle{blankNode} = [thick,inner sep=0pt,minimum size=6mm]
\node[sgNode] (s0) {$s_0$};
\node[sgNode] (s1) [below of=s0,yshift=.4cm] 
{$s_1$} edge [<-,very thick] node [right] {$r_1$} (s0);
\node[sgNode] (s2) [below of=s1,yshift=.4cm] 
{$s_2$} edge [<-,very thick] node [right] {$r_2$} (s1);
\node[sgNode] (s3) [below of=s2,yshift=.4cm] 
{$s_3$} edge [<-,very thick] node [right] {$r_3$} (s2);
\node[sgNode] (s4) [below of=s3,yshift=.4cm,xshift=-1.25cm] 
{$s_4$} edge [<-,very thick] node [left] {$r_4$} (s3);
\node[sgNode] (s5) [below of=s4,yshift=.4cm] 
{$s_5$} edge [<-,very thick] node [left] {$r_5$} (s4);
\node[sgNode] (s6) [below of=s5,yshift=.4cm] 
{$s_6$} edge [<-,very thick] node [left] {$r_6$} (s5);

\node[sgNode] (s7) [below of=s0,yshift=.4cm,xshift=-1.25cm] 
{$s_7$} edge [<-, very thick] node [left,yshift=.2cm,xshift=.1cm] {$r_2$} (s0);
\node[sgNode] (s8) [below of=s7,yshift=.4cm] 
{$s_8$} edge [<-,very thick] node [left] {$r_1$} (s7);
\node[sgNode] (s9) [below of=s8,yshift=.4cm] 
{$s_9$} edge [<-,very thick] node [left] {$r_5$} (s8);
\node[sgNode] (s10) [below of=s9,yshift=.4cm,xshift=-1.25cm] 
{$s_{10}$} edge [<-,very thick] node [left] {$r_3$} (s9)
edge [->,very thick] node [left] {$r_4$} (s5);

\node[pgNode] (s20) [below of=s8,yshift=.4cm,xshift=-1.25cm] 
{$s_{20}$} edge [<-] node [left] {$r_3$} (s8)
edge [->] node [left] {$r_5$} (s10);

\node[pgNode] (s11) [below of=s3,yshift=.4cm] 
{$s_{11}$} edge [<-] node [right,xshift=-.05cm] {$r_5$} (s3)
edge [->] node [left,yshift=.2cm,xshift=.2cm] {$r_4$} (s5);
\node[pgNode] (s12) [below of=s11,yshift=.4cm] 
{$s_{12}$} edge [<-] node [right] {$r_6$} (s11)
edge [->] node [right] {$r_4$} (s6);

\node[pgNode] (s13) [below of=s3,yshift=.4cm,xshift=1.25cm] 
{$s_{13}$} edge [<-] node [right] {$r_6$} (s3);
\node[pgNode] (s14) [below of=s13,yshift=.4cm] 
{$s_{14}$} edge [<-] node [left] {$r_5$} (s13);
\node[sgNode] (s15) [below of=s14,yshift=.4cm,xshift=1.25cm] 
{$s_{15}$} edge [<-] node [left] {$r_4$} (s14);
\node[sgNode] (s16) [below of=s13,yshift=.4cm,xshift=1.25cm] 
{$s_{16}$} edge [<-] node [right] {$r_4$} (s13)
edge [->,very thick] node [right] {$r_5$} (s15);

\node[sgNode] (s17) [below of=s1,yshift=.4cm,xshift=1.25cm] 
{$s_{17}$} edge [<-,very thick] node [right] {$r_3$} (s1);
\node[sgNode] (s18) [below of=s17,yshift=.4cm] 
{$s_{18}$} edge [<-,very thick] node [right] {$r_2$} (s17);
\node[sgNode] (s19) [below of=s18,yshift=.4cm,xshift=1.25cm] 
{$s_{19}$} edge [<-,very thick] node [right] {$r_4$} (s18)
edge [->, very thick] node [right] {$r_6$} (s16);

\end{tikzpicture}
\end{minipage}
&
\begin{tabular*}{\textwidth}{l}
\begin{minipage}[t]{\dimexpr0.5\textwidth-2\tabcolsep}
\centering
\begin{lstlisting}[
    basicstyle=\scriptsize, %or \small or \footnotesize etc.
    mathescape=true
]
$r_1$: post($t_1$,$e_1$,$t$); // runs on thread $t_1$
$r_2$: post($t_2$,$e_2$,$t$); // runs on thread $t_2$
h1 := {$r_3$: post($t$,$e_3$,$t$); $r_4$: y = 2;}
h2 := {$r_5$: x = 5;}
h3 := {$r_6$: x = -5;}

\end{lstlisting}
\end{minipage}
\\
\begin{minipage}[t]{\dimexpr0.5\textwidth-2\tabcolsep}
\vspace{-10pt}
\captionof{figure}{Pseudo code of an event-driven program.}
\label{fig:listing}
\end{minipage}
\\
\begin{minipage}{\dimexpr0.5\textwidth-2\tabcolsep}
\centering
\begin{tikzpicture}[auto,node distance=16 mm,semithick, scale=0.7, transform
shape]
\tikzstyle{blankNode} = [thick,inner sep=0pt,minimum size=6mm]

\node[blankNode] (r1) [font=\boldmath] {\large $r_1$};
\node[blankNode] (r3) [below of=r1,yshift=.3cm, font=\boldmath] 
{\large $r_3$} edge [<-,blue] (r1);
\node[blankNode] (r4) [below of=r3,yshift=.3cm,xshift=-.75cm, font=\boldmath] 
{\large $r_4$} edge [<-,blue] (r3);
\node[blankNode] (r6) [below of=r3,yshift=.3cm,xshift=.75cm, font=\boldmath] 
{\large $r_6$} edge [<-,blue] (r3);
\node[blankNode] (r2) [right of=r1,xshift=-.25cm, font=\boldmath] {\large $r_2$};
\node[blankNode] (r5) [below of=r2,yshift=.3cm, font=\boldmath] 
{\large $r_5$} edge [<-,blue] (r2)
edge [->,blue,very thick] (r6);

\node[fit=(r2)(r5), draw, dashed, gray!95] {};

\node[blankNode] (dummynode1) at ($(r4)!0.5!(r2)$) {};
\node[blankNode] (captionA) [below of=dummynode1,yshift=-.65cm] {
\begin{tabular}{c}
\large (a) $w_1$: $r_1.r_2.r_3.r_4.r_5.r_6$, \\
\large \phantom{zzz}$w_2$: $r_2.r_1.r_5.r_3.r_4.r_6$ \\
\end{tabular}};

\node[blankNode] (br1) [right of=r2, xshift=2.3cm,font=\boldmath] {\large $r_1$};
\node[blankNode] (br3) [below of=br1,yshift=.3cm, font=\boldmath] 
{\large $r_3$} edge [<-,blue] (br1);
\node[blankNode] (br4) [below of=br3,yshift=.3cm,xshift=-.75cm, font=\boldmath] 
{\large $r_4$} edge [<-,blue] (br3);
\node[blankNode] (br6) [below of=br3,yshift=.3cm,xshift=.75cm, font=\boldmath] 
{\large $r_6$} edge [<-,blue] (br3);
\node[blankNode] (br2) [right of=br1,xshift=-.25cm, font=\boldmath] {\large $r_2$};
\node[blankNode] (br5) [below of=br2,yshift=.3cm, font=\boldmath] 
{\large $r_5$} edge [<-,blue] (br2)
edge [<-,blue,very thick] (br6);

\node[fit=(br2)(br5), draw, dashed, gray!95] {};

\node[blankNode] (dummynode2) at ($(br4)!0.5!(br2)$) {};
\node[blankNode] (captionB) [below of=dummynode2,yshift=-.5cm] {
\large (b) $w_3$: $r_1.r_3.r_2.r_4.r_6.r_5$};
\end{tikzpicture}
\end{minipage} 

\end{tabular*}
\\
\begin{minipage}[t]{\dimexpr0.5\textwidth-2\tabcolsep}
\captionof{figure}{
  Partial event-parallel state space of the program in Figure~\ref{fig:listing}.}
\label{fig:state-transition-pg}
\end{minipage}
&
\begin{minipage}[t]{\dimexpr0.5\textwidth-2\tabcolsep}
\captionof{figure}{Dependence graphs of some sequences in $\mathcal{S}_G$
of the program in Figure~\ref{fig:listing}.}
\label{fig:dependence-graph}
\end{minipage}

\end{tabular*}
\end{figure*}

We illustrate the event-parallel transition system for the example program
in Figure~\ref{fig:listing}.
Here, \texttt{x} and \texttt{y} are shared variables.
The transitions $r_1$ and $r_2$ respectively run on threads $t_1$ and $t_2$. 
The last three lines in Figure~\ref{fig:listing} give
definitions of handlers of the events \texttt{e1}, \texttt{e2} and \texttt{e3} respectively.
Figure~\ref{fig:state-transition-pg} shows a partial state space
of the program in Figure~\ref{fig:listing} according to
the event-parallel transition system semantics. The edges are labeled with
the respective transitions. The shaded states and
thick edges indicate part of the state space that is reachable in the
transition system semantics of Section~\ref{sec:transition-system-por} as well,
under the mapping between states and transitions described above.

\begin{definition}
\label{def:epmt-dependence}
\normalfont 
Let $R_P$ be the set of transitions in the event-parallel transition
system $\mathcal{P}_G$ of a program $A$.
Let $D_P \subseteq R_P \times R_P$ be a binary, reflexive and symmetric 
relation. The relation $D_P$ is a valid \textbf{\emph{event-parallel dependence 
relation}} iff for all $(r_1, r_2) \in R_P \times R_P$, $(r_1, r_2) \notin D_P$ 
implies that the following conditions hold for all states $s \in \mathcal{S}_P$:
\begin{enumerate}
\item If $r_1$ is enabled in $s$ and $s' = r_1(s)$ then 
$r_2$ is enabled in $s$ iff it is enabled in $s'$.
\item If $r_1$ and $r_2$ are both enabled in $s$ then 
there exists $s' = (l',g',q') = r_1(r_2(s))$ and $s'' = (l'',g'',q'') = r_2(r_1(s))$  
such that $l' = l''$ and $g' = g''$.
\end{enumerate}
\end{definition}

This definition is similar to the definition of dependence relation 
in~\cite{Godefroid:1997:MCP:263699.263717} except that we do \emph{not} require equality of 
the event states $q'$ and $q''$ in the second condition above.
Clearly, any pair of \texttt{post} transitions, even if posting to the same
event queue, are \emph{independent} according to the event-parallel dependence relation.

\begin{definition}
\label{def:dependence-relation}
\normalfont 
Let $R$ be the set of transitions in the transition system $\mathcal{S}_G$
of a program $A$.
Let $D_P$ be a valid 
event-parallel dependence relation for $A$ and 
$D \subseteq R \times R$ be a binary, reflexive and symmetric relation.
The relation $D$ is a valid \textbf{\emph{dependence relation}} 
iff for all $(r_1, r_2) \in R \times R$, $(r_1, r_2) \notin D$ implies 
that the following conditions hold:
\begin{enumerate}
 \item If $r_1$ and $r_2$ are transitions of handlers of two different events
 $e_1$ and $e_2$ executing on the \emph{same thread} then the following conditions hold:
\begin{enumerate}
\item[(A)] Events $e_1$ and $e_2$ may be reordered in $\mathcal{S}_G$.
\item[(B)] $ep(r_1)$ and $ep(r_2)$ are independent in $D_P$, \ie\ $(ep(r_1),ep(r_2))$ $\not\in D_P$.
\end{enumerate}
 \item Otherwise, conditions 1 and 2 in Definition~\ref{def:epmt-dependence}
hold for all states $s \in \mathcal{S}$.
\end{enumerate}
\end{definition}

In the definition above,
we use the event-parallel dependence relation $D_P$ to formalize single-threaded
dependence between transitions of two handlers in $\mathcal{S}_G$ and apply
the constraints in Definition~\ref{def:epmt-dependence} to states in $\mathcal{S}_G$
to define (1)~dependence among transitions of the same event handler and
(2)~multi-threaded dependence.
From the second condition in Definition~\ref{def:dependence-relation}, all 
\texttt{post}s are considered as \emph{independent} of each other
in $\mathcal{S}_G$.

\begin{example}
\label{ex:dependence}
\upshape
The transitions $r_5$ and $r_6$ in Figure~\ref{fig:listing}
run in two different event handlers but on the same thread \texttt{t}.
Since in the event-parallel transition system, the handlers execute
concurrently, we can inspect the effect of reordering $r_5$ and $r_6$
on a state where they are co-enabled.
In particular, at state $s_3$ in Figure~\ref{fig:state-transition-pg},
the sequence $r_{6} . r_{5}$ reaches state $s_{14}$, whereas,
$r_{5} . r_{6}$ reaches $s_{12}$ which differs from $s_{14}$ in the value of
\texttt{x}. Therefore, $(r_5,r_6) \in D_P$ and by condition 1.B of
Definition~\ref{def:dependence-relation}, $(r_5,r_6) \in D$.
\end{example}

The condition 1.A of Definition~\ref{def:dependence-relation} requires that
the ordering between $e_1$ and $e_2$ should not be fixed. Suppose the handler
of $e_1$ posts $e_2$ but the two handlers do not have any pair of transitions
that are in $D_P$. Recall that we do not track dependence through event queues.
Nevertheless, since a \texttt{post} transition in $e_1$ \emph{enables} $e_2$, 
the transitions in the two handlers should be marked as dependent.
This requirement is met through condition 1.A. 
Intuitively, it serves a purpose analogous to condition 1 of
Definition~\ref{def:epmt-dependence}.

If $(r_i,r_j) \in D$, we simply say that $r_i$ and $r_j$ are \emph{dependent}. 
In practice, we over-approximate the dependence relation, for example, by considering
all conflicting accesses to shared objects as dependent.

\subsection{Dependence-covering Sets}
\label{sec:dcset}

Mazurkiewicz trace~\cite{DBLP:conf:ac:Mazurkiewicz86} forms the
basis of POR for multi-threaded programs
and event-driven programs where \texttt{post}s are
considered dependent. Two transition sequences belong to the same
Mazurkiewicz trace if they can be obtained from each other by
reordering adjacent independent transitions. The objective of POR is
to explore a representative sequence from each Mazurkiewicz trace.
As pointed out in the Introduction, the reordering of 
\texttt{post}s (independent as per Definition~\ref{def:dependence-relation})
in a transition sequence $w$ may not yield another sequence
belonging to the same Mazurkiewicz trace (denoted $[w]$) for two reasons:
(1)~it may reorder dependent transitions from the corresponding event handlers and
(2)~some new transitions, not in $w$, may be pulled in.

We elaborate on the second point. Suppose in $w$, a handler
$h_1$ executes before another handler $h_2$, both on the same thread,
such that $h_2$ is executed only partially in $w$. 
Let us reorder the \texttt{post} operations
for these two and obtain a transition sequence $w'$. Since the handlers run to completion,
in order to include all the transitions of $h_1$ (executed in $w$) in $w'$,
we must complete execution of $h_2$. However, as $h_2$ is only partially executed in $w$,
this results in including \emph{new} ---previously unexplored---
transitions of $h_2$ in $w'$. This renders $w$ and $w'$ inequivalent
by the notion of Mazurkiewicz equivalence which expects the set of transitions
in two equivalent sequences to be identical.

We therefore propose an alternative notion,
suitable to correlate two transition sequences in event-driven programs, called the \emph{dependence-covering sequence}.
The objective of our reduction is to explore a dependence-covering sequence
$u$ at a state $s$ for any transition sequence $w$ starting at $s$.

Let $w: r_1.r_2 \ldots r_n$ and $u: r_1'.r_2' \ldots r_m'$ be two transition sequences
from the same state $s$ in $\mathcal{S}_G$ reaching states $s_n$ and $s_m'$ respectively. 
Let $R_w = \{r_1,\ldots,r_n\}$ and $R_u = \{r_1',\ldots,r_m'\}$.

\begin{definition}
\label{def:dep-covering}
\normalfont
The transition sequence $u$ is  called a \textbf{\emph{dependence-covering 
sequence}} of $w$ if (i)~all the transitions in $w$ are in $u$
but $u$ can have more transitions than $w$
(i.e., $R_w \subseteq R_u$) and (ii)~for each pair of dependent transitions 
$r_i',r_j' \in R_u$ such that $i < j$, any one among the following conditions holds:
\begin{enumerate}
\item $r_i'$ and $r_j'$ are executed in $w$ and their relative order in $u$ is 
consistent with that in $w$.
\item $r_i'$ is executed in $w$ and $r_j' \in nextTrans(s_n)$.
\item $r_i'$ is not executed in $w$, $r_j' \in nextTrans(s_n)$
and $w$ can be extended in $\mathcal{S}_G$ such that $r_i'$ executes before $r_j'$.
\item Irrespective of whether $r_i'$ is executed in $w$ or not, 
$r_j'$ is not in $R_w \cup nextTrans(s_n)$.
\end{enumerate}
\end{definition}

The condition (i) above allows new transitions, that are not in $w$,
to be part of $u$. The condition (ii) restricts how the new transitions may interfere
with the dependences exhibited in $w$ and also requires all the dependences in
$w$ to be maintained in $u$. These conditions permit dependence-covering
sequence to be a relaxation of Mazurkiewicz trace, making it
more suitable for stateless model checking of event-driven
programs where \texttt{post}s may be reordered selectively.

\begin{example}
\upshape 
As an example, let $w_1$, $w_2$ and $w_3$ be the three transition sequences
in Figure~\ref{fig:state-transition-pg} which correspond to valid sequences
in the transition system $\mathcal{S}_G$ of the program in Figure~\ref{fig:listing}.
The sequences of transitions in $w_1$, $w_2$ and $w_3$ are listed in
Figure~\ref{fig:dependence-graph}. To illustrate dependence-covering sequences,
we visualize the dependences in these
sequences as directed graphs, called dependence graphs, in Figure~\ref{fig:dependence-graph}.
The nodes in the dependence graph of a transition sequence $w$ represent transitions in $w$.
If a transition $r_i$ executes before another transition $r_j$ in $w$ such that
$r_i$ and $r_j$ are dependent then we draw an edge from $r_i$ to $r_j$.
The sequences $w_1$ and $w_2$ are dependence-covering sequences
of each other. As can be seen in Figure~\ref{fig:dependence-graph}(a),
their dependence graphs are identical.
Also, both $w_1$ and $w_2$ are dependence-covering sequences of
a sequence $w_4 = r_2.r_5$. The dependence graph of $w_4$ is isomorphic to
a subgraph (enclosed in a rectangular box) of Figure~\ref{fig:dependence-graph}(a).
For transitions $r_1$, $r_3$, $r_4$ and $r_6$ which do not belong to this subgraph,
there are no restrictions on dependences among themselves. However,
by Definition~\ref{def:dep-covering}, there can be
no incoming edge to the subgraph from nodes not in the subgraph.
Transition sequences $w_1$ and $w_2$ satisfy these criteria w.r.t. $w_4$ and 
hence are dependence-covering sequences of $w_4$.
However, we note that $w_4$ and $w_1$ (or $w_2$) do not belong to the same Mazurkiewicz trace.
The sequence $w_3$ is not a dependence-covering sequence
of $w_4$ since there is an interfering dependence $(r_6,r_5) \in D$ to
the transition $r_5$ executed in $w_4$. Pictorially, we can see an
incoming edge from $r_6$ to $r_5$ in Figure~\ref{fig:dependence-graph}(b).
\end{example}

\paragraph{Note.} An important takeaway from the above example is that a dependence-covering
sequence $u$ of a transition sequence $w$ can reorder event handlers seen in 
$w$ so long as the relative ordering of dependent transitions in $w$ are not 
altered. Hence, in addition to identifying similarities between thread schedules,
dependence-covering sequences enable identification of similar ordering between
events as well. Recognizing similar event orderings was not possible with the 
Mazurkiewicz way of identifying equivalence between transition sequences.
 
\begin{definition}
\label{def:dc-set}
\normalfont
A non-empty subset $L$ of transitions enabled at a state $s$ in $\mathcal{S}_G$ is a 
\textbf{\emph{dependence-covering set}} in $s$ iff,
for all non-empty sequences of transitions $w: r_1 \ldots r_n$ starting at $s$, 
there exists a dependence-covering sequence $u: r_1' \ldots r_m'$ of $w$ 
starting at $s$ such that $r_1' \in L$.
\end{definition}

\begin{example}
\label{ex:dcs}
\upshape
All the transition sequences connecting state $s_0$ to state $s_5$ in Figure~\ref{fig:state-transition-motivation} 
are dependence-covering sequences of each other. Thus, each of $\{r_1\}$, $\{r_2\}$ and $\{r_1,r_2\}$
are dependence-covering sets at $s_0$.
Even if we take a prefix $\sigma$ of any of these sequences, the shaded sequence
in Figure~\ref{fig:state-transition-motivation} is a dependence-covering sequence of $\sigma$.

In Figure~\ref{fig:motivation-statespace-mt}, $\{r_2\}$ and $\{r_1,r_2\}$ are individually
dependence-covering sets at state $s_0$, whereas, $\{r_1\}$ is not a dependence-covering set at $s_0$.
\end{example}

For efficient stateless model checking of event-driven programs,
we can explore a reduced state space using dependence-covering
sets.

\begin{definition}
\label{def:dep-cover-statespace}
\normalfont
A \textbf{\emph{dependence-covering state space}} of an
event-driven program $A$ is a 
reduced state space $\mathcal{S}_R \subseteq \mathcal{S}_G$ obtained by selectively 
exploring only the transitions in a dependence-covering set at each state in $\mathcal{S}_G$ reached from $s_{init}$.
\end{definition}

The objective of a POR approach is to show
that even while exploring a reduced state space,  
no concurrency bug is missed w.r.t. the complete but possibly
much larger state space.
The exploration of a dependence-covering state space satisfies
this objective.
The following theorem states this guarantee.
\begin{theorem}
\label{theorem:deadlock-reachability}
\normalfont
Let $\mathcal{S}_R$ be a dependence-covering state space of an event-driven
program $A$ with a finite and acyclic state space $\mathcal{S}_G$. 
Then, all deadlock cycles in $\mathcal{S}_G$ are 
reachable in $\mathcal{S}_R$. If there exists a state $v$ in $\mathcal{S}_G$
which violates an assertion $\alpha$ defined over local variables then
there exists a state $v'$ in $\mathcal{S}_R$ which violates $\alpha$.
\end{theorem}

The proof follows from the appropriate restrictions on allowed
dependences in a dependence-covering sequence $u$ compared to
the dependences in $w$ where $w$ is required to reach
a deadlock cycle or an assertion violation in the complete state space. We provide
a complete proof of the above theorem in Appendix~\ref{app:dc-properties}.

The set $\{r_1,r_2\}$ is both a persistent set and a dependence-covering
set at state $s_0$ in Figure~\ref{fig:state-transition-motivation}.
We observe that in general, a persistent set $P$ at a state $s \in \mathcal{S}_G$
is also a dependence-covering set at $s$. Here, persistent set
is defined using the dependence relation where \texttt{post}s to the same event queue are
dependent, whereas, dependence-covering set is defined using
the dependence relation where they are not (more formally, using
Definition~\ref{def:dependence-relation}). We present a proof
of this claim in Appendix~\ref{sec:ps-as-dcs}.
Note that a dependence-covering set need not be a persistent set.
As seen in Example~\ref{ex:dcs}, $\{r_1\}$ and $\{r_2\}$ individually are both
dependence-covering sets at $s_0$ in Figure~\ref{fig:state-transition-motivation}
but they are not persistent sets.

\section{Dynamic Algorithm to Compute Dependence-covering Sets}
\label{sec:algo}
This section describes the \emdpor\ algorithm for model checking event-driven
multi-threaded programs to explore a dependence-covering state space 
(see Definition~\defdepcoverstatespace). 
\emdpor\ extends DPOR~\cite{Flanagan:2005:DPR:1040305.1040315} to 
compute dependence-covering sets. However, it differs from DPOR in many key steps.

\subsection{Comparison between DPOR and EM-DPOR}
DPOR performs depth first traversal on the transition system of a program. Instead of exploring
all the enabled transitions at a state, it only explores transitions added as
backtracking choices by the steps of the algorithm which guarantees exploring
a persistent set at each visited state. On exploring a sequence $w$ reaching a state $s'$, and seeing dependence
between a transition $r' \in \nextTrans(s')$ and a transition $r$ executed at a 
state $s$ reached by a prefix of $w$, 
DPOR adds backtracking choices at state $s$, so as to reorder $r$ and $r'$ eventually. 
However, not every pair of
dependent transitions can be reordered. For example, a pair of dependent transitions
where one transition enables the other, cannot be reordered. DPOR uses a dependence
relation which implicitly considers every adjacent pair of transitions executed
on the same thread as dependent, because executing a transition on a thread enables
the execution of the next transition. Hence, DPOR only attempts to reorder dependent 
transitions which \emph{may be co-enabled}, \ie\ atleast executed on different threads.
However, a pair of dependent transitions executed on different threads may have a
strict ordering between them in a given execution, making them unsuitable for reordering
at any state reached in that execution. DPOR uses \emph{happens-before relation}, a partial
order relation on dependent transitions, to capture the ordering between dependent transitions
in a transition sequence. 
DPOR reorders only those may be co-enabled dependent transitions
which are not ordered by happens-before relation over the explored sequence. 

EM-DPOR, extends the DPOR~\cite{Flanagan:2005:DPR:1040305.1040315}
algorithm and computes dependence-covering sets. However, it differs from DPOR in several
ways. In particular, EM-DPOR incorporates several non-trivial steps
(1)~to reason about both multi-threaded dependences as well as dependent transitions
from different event handlers on the same thread (single-threaded dependences),
and (2)~to identify events for selective reordering and infer appropriate backtracking choices
to achieve the reordering. 
In order to perform these steps, EM-DPOR uses the dependence relation defined by Definition~\ref{def:dependence-relation}, to
identify dependent transitions. A happens-before relation based on this dependence
relation does not totally order all the transitions executed on the same thread,
and restricts the total ordering only within a task (due to the second condition in Definition~\ref{def:dependence-relation}).
A task refers to an event handler or a thread without an event queue.
Analogously, EM-DPOR attempts to reorder a pair of dependent transitions which
\emph{may be co-enabled} or executed in the handlers of \emph{may be reordered} events (see Section~\ref{sec:transition-system-por})
on the same thread. Typically, dynamic POR algorithms only reorder dependent transitions
\ie\ they add backtracking choices only at a state which executes a transition $r$ dependent with 
another transition $r'$ such that $r$ and $r'$ are identified for reordering. 
This is not the case with EM-DPOR. Due to atomic execution of event handlers and FIFO 
processing of events in a queue, reordering a pair of dependent transitions from different
handlers on the same thread would require reordering their corresponding \texttt{post}s. 
Transitions posting to the same event queue may have to be reordered even to reorder
dependent transitions on different threads, as shown for the state space 
in Figure~\ref{fig:motivation-statespace-mt}. Hence, EM-DPOR \emph{selectively} reorders
\texttt{post}s to the same event queue 
even though the dependence relation used by EM-DPOR considers all the pairs of \texttt{post}s 
to be independent. 
When attempting to reorder a transition $r$ executed at a state $s$ and a dependent transition $r'$,
if EM-DPOR fails to add backtracking choices at state $s$ 
then, EM-DPOR employs a recursive strategy to dynamically identify and reorder certain 
\texttt{post}s to the same event queue. 
As will be explained in Example~\ref{ex:reschedulepending}, EM-DPOR requires the enforced ordering between such selectively
reordered \texttt{post} operations to be captured. Hence, the happens-before relation
that we use with EM-DPOR is defined to be a partial order on dependent transitions
as well as selectively reordered \texttt{post}s. 

\subsection{Definitions}
\label{sec:emdpor-algo-defs}
We now define \emph{(selectively) reordered posts} and the happens-before
relation used by our algorithm. We also define a few functions that will be used in the
rest of the section, and a notion of \emph{diverging posts} that will be used by 
EM-DPOR to reorder a pair of transitions from different event handlers 
on the same thread.

\sloppy{
\paragraph{Reordered posts.} We define a function $\reorderedPosts(p,w)$ which takes
a transition $p$ posting an event to a thread $t$'s event queue and a sequence $w$ explored by EM-DPOR
where $p$ is executed in $w$, as input, and returns a set $P$ of transitions 
such that a transition $p'$ is a member of $P$ if the following conditions hold:
\begin{enumerate}
\item $p'$ posts an event to thread $t$'s event queue.
\item There exists a prefix $w_1$ of $w$ such that $w = w_1.w_2$, $w_1$ reaches 
a state $s$, $p$ is executed in $w_2$, and the following holds:
\begin{enumerate}
\item[(A)] EM-DPOR has already explored a sequence $w_1.w_3$ where $w_3 = p.a_1 \ldots a_i .p'. a_{i+1} \ldots a_m$, each 
 $a_i$ for $1 \leq\ i \leq m$ is a transition, and has added backtracking choices
 at state $s$ to reorder the \post{} transitions $p$ and $p'$, and
\item[(B)] $p'$ is a transition in $w_2$ such that $index(w_2,p') < index(w_2,p)$.
\end{enumerate}
\end{enumerate}
}

\paragraph{Happens-before relation.} 
In the concurrency model assumed, the events posted to the same event queue are 
handled in FIFO order. Hence, we extend the happens-before relation defined 
in~\cite{Flanagan:2005:DPR:1040305.1040315} with a rule to reason about FIFO ordering
and a rule to capture ordering between reordered \texttt{post}s. 

\begin{definition}
\label{def:happens-before}
\normalfont
For a transition sequence $w : r_1 . r_2 \ldots r_n$ in $\mathcal{S}_G$ explored by EM-DPOR, 
the \textbf{\emph{happens-before relation}} $\to_{w}$ is the 
smallest relation on $\domof(w)$ such that the following conditions hold:
\begin{enumerate}
\item If $i \leq j$ and $r_i$ is dependent with $r_j$ then $i \to_{w} j$.
\item If $r_i$ and $r_j$ are two different transitions posting events $e$ and $e'$
respectively to the same thread, such that 
$i \to_w j$ and the handler of $e$ has finished and that of
$e'$ has started in $w$, then $getEnd(w,e) \to_w getBegin(w,e')$.
This is the FIFO rule.
\item If $r_j$ is a \texttt{post} transition and $r_i \in \reorderedPosts(r_j,w)$
such that $i = max(\{l \mid r_l \in \reorderedPosts(r_j,w) \})$ then $i \to_{w} j$.
\item $\to_{w}$ is transitively closed.
\end{enumerate}
The relation $\to_w$ is defined over transitions in $w$.
We overload $\to_w$ to relate transitions in $w$ with those in the $nextTrans$
set in the last state, say $s$, reached by $w$.
For a task $(t,e)$ having a transition in $\nextTrans(s)$, $i\to_{w} (t,e)$ if either (a)~$\taskof(r_i)$ $=$ $(t,e)$ or
(b)~$\exists k \in dom(w)$ such that $i\to_{w} k$ and $\taskof(r_k) = (t,e)$.
\end{definition}

We note that unlike the happens-before relation defined in~\cite{Flanagan:2005:DPR:1040305.1040315},
the happens-before relation defined above captures some information related to 
sequences rooted at states reached by prefixes of $w$ explored by EM-DPOR prior
to exploring $w$. This is required to add happens-before mapping between reordered
\texttt{post}s. 

\paragraph{Diverging posts.} For a transition sequence $w$ in $\mathcal{S}_G$ reaching a state $s$
and a transition $r$ in $w$ or $\nextTrans(s)$, 
let $\postChain(r,w) = p_{m}.p_{m-1} \ldots p_1$ be the maximal sequence of \post\ transitions in $w$
such that $p_{i-1}$ is a transition in the handler of the event posted
by $p_i$ for $m \geq i > 1$, and $p_1$ posts the event whose
handler executes $r$. Let $r$ and $r'$ be transitions of two handlers
running on the same thread such that
$\postChain(r,w) = p_k \ldots p_1$ and $\postChain(r',w) = q_l \ldots q_1$. 
Then, $\divergingPosts(r,r',w)$ is a pair of \texttt{post}s $(p_i,q_i)$ 
where $i$ is 
the smallest index in the post-chains of $r$ and $r'$ in sequence $w$ 
such that $\threadof(p_i) \neq \threadof(q_i)$.
In Figure~\ref{fig:motivation-mt}, $\divergingPosts(r_3,r_4,r_1\ldots r_6) = (r_1,r_2)$.
Diverging posts are undefined if there exists an index $j$ such that 
$\taskof(p_j) = \taskof(q_j)$ and for all $i < j$, $\threadof(p_i) = \threadof(q_i)$. 

The order of execution of diverging posts of $r$ and $r'$ uniquely determines 
the order of execution of $r$ and $r'$. 
In Figure~\ref{fig:motivation-mt}, the order of execution of $r_1$ and $r_2$
uniquely determines the order of execution of $r_3$ and $r_4$. 
If $r$ and $r'$ do not have diverging \texttt{post}s, their relative order of execution is fixed.

\paragraph{Helper functions and data structures.} 
Function $enabled(s)$ gives the set of threads whose next transitions are
enabled at a state $s$. 
Consider a transition sequence $w : r_1 \ldots r_n$ from 
the initial state $s_{init}$ of a given event-driven multi-threaded program.
The function $last(w)$ gives the last state reached by $w$. 
If $w$ is empty, it is the initial state.
For an index $k \in \domof(w)$, $pre(w,k)$ is the state before executing transition $r_k$.
The function $getPost(w,e)$ gives the transition in $w$ which posted the event $e$. 
Function $\eventof(r)$ gives the event corresponding
to the handler which executes $r$ 
(this is $nil$ if $r$ is executed by a thread
without an event queue). For a thread $t$ with an event queue,
the function $executable(s,t)$ returns the event whose handler can 
perform the next transition on $t$ in a state $s$, whereas $blockedEv(s,t)$ 
returns the set of events present in $t$'s queue in state $s$ that are not executable.
We say that a task $(t,e)$ is executable at a state $s$ if $t$ is a thread without a queue
($e = nil$), or $e = executable(s,t)$. 
Function $execTasks(s)$ returns the set of tasks whose events are executable in state $s$,
whereas $blockedTasks(s)$ returns the set of tasks whose events are blocked in state $s$.
Function $\destof(r)$ takes a transition $r$ posting an event as input and returns the destination thread. 
Data structures $backtrack(s)$ and $done(s)$ respectively track the threads added as backtracking choices
at a state $s$, and the threads already explored from a state $s$ during the DFS traversal.
Another data structure the algorithm populates is the set $RP$ maintained at every 
visited state. The set $RP(s)$ corresponding to a state $s$ is a set of ordered 
pairs of transitions where a pair $(a,b) \in RP(s)$ is such that $a$ and $b$ are 
posts to the same thread such that $a$ and $b$ have been identified for reordering
in an execution where $a$ is executed prior to $b$. The set $RP$ will be implicitly
looked up to compute the set \reorderedPosts\ of a post operation, and in turn 
derive happens-before ordering between posts as per condition~3 
in Definition~\ref{def:happens-before}.

\subsection{Overview of EM-DPOR Algorithm}
\label{sec:emdpor-algo}
This section describes the \emdpor\ algorithm to explore a dependence-covering 
state space (see Definition~\ref{def:dep-cover-statespace}) of event-driven
programs obeying the concurrency model described in Section~\ref{sec:transition-system-por}.
 
The EM-DPOR algorithm has two components: (1)~a depth first search based state space explorer
called \texttt{Explore}, and (2)~a recursive routine called \texttt{FindTarget} to
compute backtracking points and choices for a pair
of reorderable dependent transitions.  
We note that the algorithms presented in this section assume dependence even 
between transitions reading from the same shared variable, even though the 
dependence relation defined by Definition~\ref{def:dependence-relation} 
considers such non-conflicting transitions to be independent. In Appendix~\ref{app:emdpor-opt}, 
we present modifications to the Algorithm \texttt{Explore} which makes \emdpor\ capable of 
treating such transitions including a few more types of transitions as independent.
We now give an overview of \texttt{Explore} and \texttt{FindTarget}.

{
\SetEndCharOfAlgoLine{}
\SetInd{6pt}{0.01pt}
\begin{algorithm*}[t]
\small
\SetKwFunction{findTarget}{FindTarget}
\SetKwFunction{Explore}{Explore}

\SetKwInput{KwInput}{Input}
\KwInput{a transition sequence \boldmath$w$\unboldmath$: r_1 \ldots r_n$ and a set 
\boldmath$rp$ \unboldmath of posts to be reordered}

Let $s$ $=$ $last(w)$; \phantom{xx} $RP(s) = rp$\;
\nllabel{line-init}

\ForEach{thread $t$}{
\nllabel{loop-start}
\Indp
\If{$\exists i = \max(\{ i \in dom(w) \mid r_i \text{ is dependent and }
  (\text{may be co-enabled or reordered with } next(s,t) ) \text{ and }  
  i\; {\not\to}_{w}\; task(next(s,t))\})$
\nllabel{line-cond}
}{
   \Indp
   \tcp{Identify backtracking point and choice to reorder $r_i$ and $next(s,t)$}
   \texttt{FindTarget}$(w, r_i, next(s,t))$\;
   \nllabel{line-findtarget}
}
}
\nllabel{loop-end}

\If{$\exists t \in enabled(s)$}{
\nllabel{line-dfs-start}
\Indp
  Let $backtrack(s)$ $=$ $\{t\}$ and $done(s)$ $=$ $\emptyset$\;
  \tcp{Perform selective depth-first traversal}
  \While{$\exists t \in (backtrack(s) \setminus done(s))$}{
\Indp
    Let $r = next(s,t)$; Execute transition $r$\;
    \If{$r$ is a \emph{\texttt{post}} operation \nllabel{line-ifpost} }{
\Indp
    \If{$\exists k = max(\{k \in dom(w) \mid r_k \in reorderedPosts(r,w.r)\})$  \nllabel{line-maxpost}}
    {  
    \Indp
    Add thread $t$ to $backtrack(pre(w,k))$\;
    \nllabel{line-addpost-to-backtrack}
    }
    $rp$ = $RP(s)\setminus\{(r,\_) \in RP(s)\}$\;
    \nllabel{line-remover-from-rp}
    }
    \nllabel{line-endifpost}
    Add $t$ to $done(s)$; $\texttt{Explore}(w \cdot r)$\;
    \nllabel{line-done}
  }
}
\nllabel{line-dfs-end}
\caption{\texttt{Explore}}
\label{algo:explore}
\end{algorithm*}
}

\paragraph{Explore.}
Algorithm \texttt{Explore}, given as Algorithm~\ref{algo:explore}, takes
a transition sequence $w$ and a set $rp$ of \texttt{post}s identified for reordering, as input and obtains the current
state $s = last(w)$ (line~\ref{line-init}). Also, the set $RP(s)$ corresponding
to state $s$ is initialized to $rp$.
Initially, \ie\ when \texttt{Explore} is invoked for the first time, $w$ is empty. 

The loop at lines~\ref{loop-start}--\ref{loop-end} iterates over all
threads $t$ and identifies transitions from $w$ that have a \emph{race} with $next(s,t)$. 
A transition $r_i$ has a \emph{race} with $next(s,t)$ if they are dependent and 
may be co-enabled (if $thread(r_i) \neq t$) or may be reordered (if $event(r_i)$ and 
$event(next(s,t))$ may be reordered),
and $r_i$ does not happen before any transition in the task that
executes $next(s,t)$. The algorithm selects a transition $r_i$ which satisfies
the above requirements and has the highest index in $w$.
It then invokes the recursive routine \texttt{FindTarget} at line~\ref{line-findtarget} 
to compute backtracking choices to reorder $r_i$ and $next(s,t)$, and if required, identify
\texttt{post}s to same thread for selective reordering. 

Lines~\ref{line-dfs-start}--\ref{line-dfs-end} perform a selective depth first 
traversal starting at state $s$ reached by $w$. The algorithm 
\texttt{Explore} is called recursively by extending the current transition sequence 
with an outgoing transition $r$ of a thread $t \in backtrack(s)$ from $s$, such that $t$
is not already explored from $s$ \ie\ $t \not\in done(s)$. 
Lines~\ref{line-ifpost}--\ref{line-endifpost} are effective only if the transition $r$
executed at state $s$ reached by $w$, is a \texttt{post} transition.
Line~\ref{line-remover-from-rp} removes those members from the set $RP(s)$ where 
the recently executed transition $r$ is the first transition in the ordered pair.
This is because after the execution of \post\ operation $r$, any remaining 
\post\ identified to be reordered w.r.t. $r$ cannot be reordered by extensions 
of the sequence $w.r$. Hence, we do not track such pairs anymore. 
We now explain intuitions for lines~\ref{line-maxpost} and \ref{line-addpost-to-backtrack} 
which add $r$'s thread as a backtracking choice at a state from where $r$'s nearest reordered \texttt{post}
is executed.

On inspecting the members of the form $(r,\_)$ in the set $RP(s)$ and checking the 
\post\ transitions in $w$, the \texttt{post}s which have been successfully 
reordered w.r.t. the \post\ transition $r$ can be identified, \ie\ $reorderedPosts(r,w.r)$
can be computed with the help of $RP$.
If the transition $r = next(s,t)$ has a \texttt{post} operation such that a transition 
$r_k$ in $w$ is its nearest reordered \texttt{post} then, condition~3 in Definition~\ref{def:happens-before}
adds a happens-before mapping from $r_k$ to $r$. The happens-before mapping from 
$r_k$ to $r$ initiates FIFO and transitive ordering 
between transitions across some of the handlers corresponding to post chains 
originating from $r_k$ and $r$; consequently, dependent transitions which could otherwise 
be identified by line~\ref{line-cond} for reordering may get ordered by happens-before.

Transition $r$ is enabled in $pre(w,k)$ --- state from which $r_k$ is executed, 
because $r_k \in reorderedPosts(r,w.r)$ which means \emdpor\ has already seen an 
execution where $r$ is executed from a state reached by a prefix of $w$ but prior
to or at $pre(w,k)$ which makes $r$ the next transition on its thread at 
$pre(w,k)$ (see definition of reordered \texttt{post}s in Section~\ref{sec:emdpor-algo-defs}). 
Hence, line~\ref{line-addpost-to-backtrack} adds thread $t$ as a backtracking choice at $pre(w,k)$,
so as to not miss alternate orderings between dependent transitions across post chains
of $r_k$ and $r$. 
For example, consider a sequence $w.r \ldots p_1 \ldots p_2 \ldots p_n \ldots w'$
explored by EM-DPOR where $w$ and $w'$ are transition sequences, and $p_i$ for $1 \leq i \leq n$
and $r$ are transitions posting to the same event queue. Assume that the handlers
of $p_1, \ldots p_{n-1}$ and $p_n$ contain transitions dependent with transitions
in $r$'s handler, and EM-DPOR identifies $p_1$, $p_2$, \ldots, $p_n$ to be reordered with $r$.
Let EM-DPOR eventually explore $v = w \ldots p_1 \ldots p_n \ldots r.w''$.
Since $p_n$ is the nearest reordered \texttt{post} w.r.t. $r$ in sequence $v$,
a happens-before mapping is added between $p_n$ and $r$. As a result the handlers
corresponding to $p_n$ and $r$ get ordered by FIFO rule, due to which the dependent
transitions in the handlers of $p_n$ and $r$ will not be selected for reordering
by line~\ref{line-cond} in Algorithm~\ref{algo:explore}. Since line~\ref{line-addpost-to-backtrack}
adds $r$'s thread to the backtracking set at the state prior to $p_n$ in sequence $v$,
EM-DPOR will still be able to explore a dependence-covering sequence for 
$w \ldots p_1 \ldots p_{n-1} \ldots r \ldots p_n.w'''$. This may be missed otherwise.
 
{
\SetEndCharOfAlgoLine{}
\SetInd{6pt}{0.01pt}
\begin{algorithm*}[!t]
\small
 \SetKwFunction{dom}{dom}
 \SetKwFunction{index}{index}
 \SetKwFunction{reschedulePending}{ReschedulePending}
 \SetKwFunction{findTarget}{FindTarget}

\SetKwInput{KwInput}{Input}
\KwInput{a transition sequence \boldmath$w$\unboldmath$: r_1 \ldots r_n$, a transition 
\boldmath$r$ \unboldmath from $w$ and 
a transition \boldmath$r'$ \unboldmath which may or may not belong to $w$}

Let $i$ $=$ $index(w,r)$ and $s$ $=$ $pre(w,i)$\;
\nllabel{line-i}
\lIf{$r' \not\in nextTrans(last(w)) \text{ and } i \to_w index(w,r')$ \nllabel{line-check-no-hb}}
{
  \textbf{return}
  \nllabel{line-return-on-hb}
}
 \tcp{Step 1: Recursively search for diverging posts} 
\If{$thread(r) = thread(r')$}{
  \Indp
  \findTarget{$w,getPost(w,event(r)),getPost(w,event(r'))$}; 
  \nllabel{line-diverging}
  \textbf{return}\;
  \nllabel{line-return-on-same-thread}
}

\tcp{Step 2: Reorder transitions from distinct threads}
\If{$r' \in nextTrans(last(w))$
\nllabel{line-conditional}
}
  {
  \Indp
  Let $candidates$ $=$ $\{task(p) \in execTasks(s)\cup blockedTasks(s) \mid thread(p) \in enabled(s)$\\ 
  \nonl\hspace{4mm} $\text{and } p = r' \text{ or } (\exists k \in dom(w): k > i$
    $\text{ and } k \rightarrow_w task(r') \text{ and } p = r_k)\}$\;
    \nllabel{line-candidate-1}
  }
\Else{
    \Indp
    Let $candidates$ $=$ $\{task(p) \in execTasks(s)\cup blockedTasks(s) \mid thread(p) \in enabled(s)$\\ 
    \nonl\hspace{4mm} $\text{and } p = r' \text{ or } (\exists k \in dom(w): k > i \text{ and } 
    k \rightarrow_w index(w,r') \text{ and } p = r_k)\}$\;
    \nllabel{line-candidate-2}
    }
\nllabel{line-postconditional}

Let $unexplored$ $=$ $\{t \mid (t,e) \in candidates\} \setminus done(s)$\;
\nllabel{line-compute-unexplored}
\If{$unexplored \neq \emptyset$ \nllabel{line-if-not-done}}{
    \Indp
    Add any $t \in unexplored$ to $backtrack(pre(w,i))$\; 
    \nllabel{line-add-unexplored-thread}
    \lIf{$r'$ and $r$ are \emph{\texttt{post}} operations}{$RP(s) = RP(s) \cup \{(r,r')\}$}
    \textbf{return}\;
}
\nllabel{line-return-success}

\tcp{Step 3: Recursively search for backtracking choices to make
a pending (blocked) event executable}
Let $pending$ $=$ $\{(t,e) \in candidates\,|\,e \in blockedEv(s,t)\}$\;
\nllabel{line-pending}
 
  \lIf{$pending \neq \emptyset$ \nllabel{line-recursive-findtarget}}{
     $\texttt{ReschedulePending}(w,pending,r)$ \tcp{See Algorithm~\ref{algo:reschedulepending}}
  }
  \Else(\tcp*[h]{Step 4: All the tasks in $candidates$ are executable, or $candidates = \emptyset$}){
\Indp
    Let $ts = \{t \mid (t,e) \in candidates\}$\;
    \nllabel{line-thread-of-tasks}
    \lIf{$ts \neq \emptyset$}{ 
      Add any $t \in ts$ to $backtrack(pre(w,i))$
      \nllabel{line-redundant-add}
    }
    \lElse{
        $\texttt{BacktrackEager}(w,i,r')$ \tcp*[h]{See Algorithm~\ref{algo:backtrackeager}}
        \nllabel{line-call-backtrackeager}
     }
  }
 \caption{\texttt{FindTarget}}
\label{algo:findtarget}
\end{algorithm*}
}

\paragraph{FindTarget.} 
\texttt{Explore} invokes \texttt{FindTarget} (Algorithm~\ref{algo:findtarget})
to compute backtracking choices to reorder a pair of dependent transitions
$r$ and $r'$.
Let $i$ be the index of $r$ in $w$ and $s$ be the state from which
$r = r_i$ is executed (line~\ref{line-i}). 
If \texttt{FindTarget} fails to identify backtracking choices to be added to $backtrack(s)$, 
then it identifies \texttt{post}s for selective reordering and recursively invokes 
itself to compute corresponding backtracking choices. 
Among other criteria, a recursive call terminates when a happens-before ordering between
$r$ and $r'$ is detected (line~\ref{line-check-no-hb}).
Transitions $r$ and $r'$ may be co-enabled or they may belong to different event handlers
on the same thread.
In the latter case, we first identify a pair of \texttt{post} operations executed 
on \emph{different} threads which need to be reordered so as to reorder $r$ and $r'$.
\texttt{FindTarget} operates in four main steps explained below, of which Steps~2 - 4
are applicable only when $thread(r) \neq thread(r')$ and Step~1 only when $thread(r) = thread(r')$.

\vspace{2mm}\noindent\emph{Step 1.}
Transitions $r$ and $r'$ may be from different tasks on the same thread.
Such transitions can only be reordered by reordering their diverging \texttt{post}s. 
Line~\ref{line-diverging} therefore recursively invokes \texttt{FindTarget} on
\post\ operations of $r$ and $r'$. This way it simultaneously walks up $postChain(r,w)$ and 
$postChain(r',w)$ on each recursive call to \texttt{FindTarget} till it finds $divergingPosts(r,r',w)$.
On reaching the diverging \texttt{post}s, 
the condition $thread(r) = thread(r')$ ---where $r$ and $r'$ are
diverging \texttt{post}s--- evaluates to false and the control goes to Step 2.

\vspace{2mm}\noindent\emph{Step 2.}
This step is reached only when $thread(r) \neq thread(r')$. Similar to the algorithm
DPOR's~\cite{Flanagan:2005:DPR:1040305.1040315}
computation of backtracking choices, this step computes threads to be added to 
$backtrack(s)$ to facilitate executing $r'$ before $r$ in a future run.
Lines~\ref{line-conditional}--\ref{line-candidate-2} compute a set $candidates$
consisting of $task(r')$ and tasks that have a transition, executed after $r$, 
with a happens-before ordering with $r'$. Tasks in set $candidates$ are restricted 
to only those which are either executable or blocked in state $s$. Additionally,
only those tasks 
whose threads are enabled at $s$ are added, so that one such thread can be explored 
from $s$ to eventually achieve the reordering. 

Threads whose transitions are already explored from state $s$ are added to
$done$ set at $s$ by line~\ref{line-done} in Algorithm~\ref{algo:explore}.
For a task $(t,e) \in candidates$, it is possible that 
its thread $t$ is already in  $done(s)$. If all the tasks in the set $candidates$
are in $done(s)$ then in case of a purely multi-threaded program, 
this would imply that the intended order between $r'$ and $r$ has already been explored. 
However, this reasoning need not hold in the presence of events.
This is because for a task $(t,e) \in candidates$ such that $t \in done(s)$, 
event $e$ may be blocked on its queue in state $s$ --- which means $t \in done(s)$ 
due to exploration of the executable task on $t$ in a prior run. However, the executable task
on $t$ may not even have any happens-before ordering with $r'$. 
In which case exploring it from state $s$ would either not have explored the required order between $r'$ and $r$,
or would not have preserved the required order between other pairs of dependent transitions
when $r'$ is executed before $r$ in a prior run. 

Hence, lines~\ref{line-compute-unexplored}--\ref{line-return-success} compute $unexplored$
to be a set of threads corresponding to tasks in $candidates$ which are not 
in $done(s)$, and add some thread in $unexplored$ to $backtrack(s)$ if
$unexplored \neq \emptyset$.
In addition, if $r$ and $r'$ are \post\ transitions then the algorithm tracks that
these two posts have been identified for reordering and backtracking choices have
been added correspondingly at state $s$ to execute $r'$ prior to $r$. This information
is tracked by adding the ordered pair $(r,r')$ to the set $RP(s)$.

If $unexplored = \emptyset$, \ie\ all the threads with transitions that \emph{happen-before} $r'$ 
are already explored from $s$, does not imply that 
$r'$ cannot be reordered with $r$ or EM-DPOR has already seen a run where $r'$ is 
explored before $r$. Rather it indicates that we need to adopt a different 
strategy to achieve the reordering. This is illustrated through an example below.

\begin{example}
\label{ex:pending}
\upshape
In sequence $w$ of Figure~\ref{fig:motivation-mt}, transitions $r_3$ and $r_6$ are 
dependent, may be co-enabled and do not have a happens-before ordering. 
When \texttt{Explore} invokes \texttt{FindTarget} to compute backtracking choices 
to reorder $r_3$ and $r_6$, Step~1 is skipped as $thread(r_3)$ 
$\neq thread(r_6)$. Step~2 computes $candidates = \{(t_1,e_2)\}$ 
as $t_1$ is enabled at $s_2$ (see Figure~\ref{fig:motivation-statespace-mt}), 
and $r_4$ executed in $(t_1,e_2)$ forks $t_4$ and thus happens before $r_6$. 
However, $t_1$ is already executed from $s_2$ and is in $done(s_2)$. Yet, as can be 
seen in Figure~\ref{fig:motivation-statespace-mt}, $r_3$ and $r_6$ can be reordered; but by reordering
$r_1$ and $r_2$ posting events $e_1$ and $e_2$ respectively. 
But adding thread $t_1$ corresponding to the only task $(t_1,e_2)$ in $candidates$
will not achieve this reordering. Step~3 explains our technique to handle such cases.
\end{example}

\noindent\emph{Step 3.} In this step, line~\ref{line-pending} computes a set $pending$ which is a subset 
of tasks in $candidates$ whose events are blocked in their event queues in state $s$. 
If set $pending$ is not empty, line~\ref{line-recursive-findtarget} invokes
\texttt{ReschedulePending}. 
Intuitively, \texttt{ReschedulePending} identifies a set of events blocked in $s$ to 
be reordered with their corresponding executable events \ie\ it performs selective reordering
of \texttt{post}s to same thread so as to eventually reorder $r$ and $r'$ executed on different threads. 
We present its details in Section~\ref{sec:reschedulepending}.

\vspace{2mm}\noindent\emph{Step 4.} Finally, the set $pending$ being empty implies that all the tasks in $candidates$ 
are executable at state $s$ or $candidates$ itself is empty. \texttt{FindTarget} computes a set of threads $ts$ 
corresponding to each task in $candidates$. 
If the set $ts$ is non-empty, it only means that another ordering of $r$ and $r'$ 
is already explored in a past run as all the threads in $ts$ are already in $done(s)$ 
(due to lines~~\ref{line-compute-unexplored}--\ref{line-return-success}), 
and the algorithm trivially adds any thread from $ts$ to $backtrack(s)$ (line~\ref{line-redundant-add}). 
If $ts = \emptyset$ which means $candidates = \emptyset$, \texttt{FindTarget} 
invokes \texttt{BacktrackEager} (see Algorithm~\ref{algo:backtrackeager}) at line~\ref{line-call-backtrackeager}. 

\subsection{Selective Reordering of Blocked and Executable Events} 
\label{sec:reschedulepending}
\texttt{ReschedulePending} (Algorithm~\ref{algo:reschedulepending}) is invoked by Algorithm~\ref{algo:findtarget}
on line~\ref{line-recursive-findtarget} in Step~3 of \texttt{FindTarget} when a transition 
$r$ executed from a state $s$
in sequence $w$ explored by EM-DPOR has to be reordered with a transition $r'$ on another thread, 
and Step~2 of \texttt{FindTarget} fails to add backtracking choices to $backtrack(s)$. \texttt{ReschedulePending} is
called only if the \emph{candidate} set of tasks computed by Step~2 has a set of tasks with their
events blocked in state $s$ such that their corresponding executable tasks are already explored
from $s$.
Then, Algorithm~\ref{algo:reschedulepending} identifies suitable events blocked in $s$ to be 
reordered with executable events on their corresponding queues, attempting to co-enable $r$ and 
$r'$ facilitating their reordering. 

We present some intuitions on scenarios where relevant pairs of events enqueued to 
the same event queue should be reordered to explore different orderings between 
a pair of transitions executed on
different threads. A pair of event handlers executed on the same thread may have to be reordered
so as to reorder a pair of transitions, say $p_1$ (assumed to be executed at a state $s$ in a sequence $v$)
and $p_n'$ (may or may not be executed in $v$) on different threads, typically in the following
scenarios. 

\vspace{1mm}
\noindent (a) Even though there exists a sequence in $\mathcal{S}_G$ where $p_n'$
is executed prior to $p_1$, in sequence $v$ however $p_1$ must be executed to eventually
execute $p_n'$. This may be the case if a transition that enables $p_n'$ is in a task whose
event is blocked in $p_1$'s thread in state $s$ (similar to the scenario presented 
for Figure~\ref{fig:motivation-statespace-mt}).

\vspace{1mm}
\noindent (b) Any transition sequence rooted at state $s$ cannot preserve the relative
ordering between a set of pairs of dependent transitions when reordering $p_1$ and $p_n'$,
even though this can be achieved by reordering some relevant pairs of events. 
This may be the case if a transition that happens before $p_n'$ 
is in a task whose event is blocked in $p_1$'s thread in state $s$. 
In such a case executing $p_n'$ prior to $p_1$ 
by adding backtracking choices at state $s$ breaks the ordering between transitions
in the blocked task on $p_1$'s thread and $p_n'$. More generally case~(b) can occur if a transition
in a task blocked on $p_1$'s thread in state $s$ happens before a transition
in the executable task on another thread, say $t_n$, such that a transition in a task blocked in $s$ on $t_n$
happens before $p_n'$. In general there may be any number of such 
blocked -- executable tasks between $p_1$ and $p_n'$, with
happens-before mapping from transitions in blocked tasks to transitions in executable 
tasks on different threads, as depicted in Figure~\ref{fig:reschedule-structure}(a). 
Clearly, reordering $p_1$ and $p_n'$ by exploring thread $t_n$ (see Figure~\ref{fig:reschedule-structure})
from state $s$ breaks the happens-before ordering between a transition in a blocked task on
thread $t_{n-1}$ and transition $p_n$ in the executable task on $t_n$.

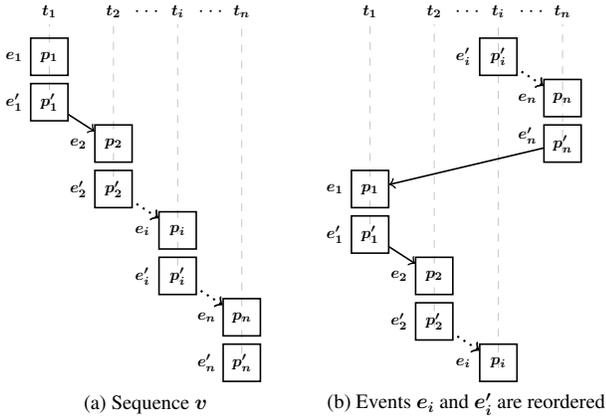
\begin{figure}[h]
\centering
\begin{tikzpicture}[auto,node distance=16 mm,semithick, scale=0.55, transform
shape, font=\boldmath]
\tikzstyle{taskNode} = [rectangle,draw=black,inner sep=0pt,minimum 
size=9mm]
\tikzstyle{plainNode} = [thick,inner sep=0pt,minimum size=6mm]

\node[plainNode] (t')  {\large $t_1$};
\node[taskNode] (rl) [below of=t',yshift=.5cm] {\large $p_1$};
\node[plainNode] (e') [left of=rl,xshift=.75cm] {\large $e_1$};
\node[taskNode] (rj1) [below of=rl,yshift=0.5cm] {\large $p_1'$}; 
\node[plainNode] (ej1) [left of=rj1,xshift=.75cm] {\large $e_1'$};
\begin{scope}[on background layer]
\node[plainNode] (inv1)  [left of=t',xshift=1.6cm] {};
\node[plainNode] (inv2)  [below of=inv1,yshift=-1cm] {}
edge [-,dashed,very thin,color=gray!45] (inv1);
\end{scope}

\node[plainNode] (t1) [right of=t',xshift=-.05cm] {\large $t_2$};
\node[plainNode] (dummynode1) at (rj1 -| t1) {};
\node[taskNode] (rk1) [below of=dummynode1,yshift=.6cm] {\large $p_2$}
edge [<-] (rj1);
\node[plainNode] (ek1) [left of=rk1,xshift=.75cm] {\large $e_2$};
\node[taskNode] (rj2) [below of=rk1,yshift=0.5cm] {\large $p_2'$}; 
\node[plainNode] (ej2) [left of=rj2,xshift=.75cm] {\large $e_2'$};

\begin{scope}[on background layer]
\node[plainNode] (inv11)  [left of=t1,xshift=1.6cm] {};
\node[plainNode] (inv22)  [below of=inv11,yshift=-3.1cm] {}
edge [-,dashed,very thin,color=gray!45] (inv11);
\end{scope}

\node[plainNode] (tellipse1) [right of=t1,xshift=-.8cm] {$\mathbf{\cdotp \, \cdotp \, \cdotp}$};

\node[plainNode] (t2) [right of=t1,xshift=-.05cm] {\large $t_i$};
\node[plainNode] (dummynode2) at (rj2 -| t2) {};
\node[taskNode] (rk2) [below of=dummynode2,yshift=.6cm] {\large $p_i$}
edge [<-, dotted, thick] (rj2);
\node[plainNode] (ek2) [left of=rk2,xshift=.75cm] {\large $e_i$};
\node[taskNode] (rj3) [below of=rk2,yshift=0.5cm] {\large $p_i'$}; 
\node[plainNode] (ej3) [left of=rj3,xshift=.75cm] {\large $e_i'$};

\begin{scope}[on background layer]
\node[plainNode] (inv111)  [left of=t2,xshift=1.6cm] {};
\node[plainNode] (inv222)  [below of=inv111,yshift=-5.25cm] {}
edge [-,dashed,very thin,color=gray!45] (inv111);
\end{scope}

\node[plainNode] (tellipse2) [right of=t2,xshift=-.8cm] {$\mathbf{\cdotp \, \cdotp \, \cdotp}$};

\node[plainNode] (t) [right of=t2,xshift=-.05cm] {\large $t_n$};
\node[plainNode] (dummynode4) at (rj3 -| t) {};
\node[taskNode] (rkk') [below of=dummynode4,yshift=.6cm] {\large $p_{n}$}
edge [<-, dotted, thick] (rj3);
\node[plainNode] (ekk') [left of=rkk',xshift=.75cm] {\large $e_n$};
\node[taskNode] (r) [below of=rkk',yshift=0.5cm] {\large $p_n'$}; 
\node[plainNode] (e) [left of=r,xshift=.75cm] {\large $e_n'$};

\begin{scope}[on background layer]
\node[plainNode] (inv111k')  [left of=t,xshift=1.6cm] {};
\node[plainNode] (inv222k')  [below of=inv111k',yshift=-7cm] {}
edge [-,dashed,very thin,color=gray!45] (inv111k');
\end{scope}

\node[plainNode] (dummynode5) at ($(t1)!0.5!(t2)$) {}; 
\node[plainNode] (captionA) [below of=dummynode5,yshift=-7.9cm] {\Large (a) Sequence $v$};


\node[plainNode] (bt') [right of=t,xshift=1.5cm] {\large $t_1$};
\begin{scope}[on background layer]
\node[plainNode] (binv1)  [left of=bt',xshift=1.6cm] {};
\node[plainNode] (binv2)  [below of=binv1,yshift=-4.25cm] {}
edge [-,dashed,very thin,color=gray!45] (binv1);
\end{scope}

\node[plainNode] (bt1) [right of=bt',xshift=-.05cm] {\large $t_2$};
\begin{scope}[on background layer]
\node[plainNode] (binv11)  [left of=bt1,xshift=1.6cm] {};
\node[plainNode] (binv22)  [below of=binv11,yshift=-6.35cm] {}
edge [-,dashed,very thin,color=gray!45] (binv11);
\end{scope}

\node[plainNode] (btellipse1) [right of=bt1,xshift=-.8cm] {$\mathbf{\cdotp \, \cdotp \, \cdotp}$};

\node[plainNode] (bt2) [right of=bt1,xshift=-.05cm] {\large $t_i$};
\begin{scope}[on background layer]
\node[plainNode] (binv111)  [left of=bt2,xshift=1.6cm] {};
\node[plainNode] (binv222)  [below of=binv111,yshift=-7.35cm] {}
edge [-,dashed,very thin,color=gray!45] (binv111);
\end{scope}

\node[plainNode] (btellipse2) [right of=bt2,xshift=-.8cm] {$\mathbf{\cdotp \, \cdotp \, \cdotp}$};

\node[plainNode] (bt) [right of=bt2,xshift=-.05cm] {\large $t_n$};
\begin{scope}[on background layer]
\node[plainNode] (binv111k')  [left of=bt,xshift=1.6cm] {};
\node[plainNode] (binv222k')  [below of=binv111k',yshift=-2.15cm] {}
edge [-,dashed,very thin,color=gray!45] (binv111k');
\end{scope}

\node[taskNode] (brj3) [below of=bt2,yshift=0.5cm] {\large $p_i'$}; 
\node[plainNode] (bej3) [left of=brj3,xshift=.75cm] {\large $e_i'$};

\node[plainNode] (bdummynode4) at (brj3 -| bt) {};
\node[taskNode] (brkk') [below of=bdummynode4,yshift=.6cm] {\large $p_{n}$}
edge [<-, dotted, thick] (brj3);
\node[plainNode] (bekk') [left of=brkk',xshift=.75cm] {\large $e_n$};
\node[taskNode] (br) [below of=brkk',yshift=0.5cm] {\large $p_n'$}; 
\node[plainNode] (be) [left of=br,xshift=.75cm,yshift=.2cm] {\large $e_n'$};

\node[plainNode] (bdummynode3) at (br -| bt') {};
\node[taskNode] (brl) [below of=bdummynode3,yshift=.5cm] {\large $p_1$}
edge [<-] (br);
\node[plainNode] (be') [left of=brl,xshift=.75cm] {\large $e_1$};
\node[taskNode] (brj1) [below of=brl,yshift=0.5cm] {\large $p_1'$}; 
\node[plainNode] (bej1) [left of=brj1,xshift=.75cm] {\large $e_1'$};

\node[plainNode] (bdummynode1) at (brj1 -| bt1) {};
\node[taskNode] (brk1) [below of=bdummynode1,yshift=.6cm] {\large $p_2$}
edge [<-] (brj1);
\node[plainNode] (bek1) [left of=brk1,xshift=.75cm] {\large $e_2$};
\node[taskNode] (brj2) [below of=brk1,yshift=0.5cm] {\large $p_2'$}; 
\node[plainNode] (bej2) [left of=brj2,xshift=.75cm] {\large $e_2'$};

\node[plainNode] (bdummynode2) at (brj2 -| bt2) {};
\node[taskNode] (brk2) [below of=bdummynode2,yshift=.6cm] {\large $p_i$}
edge [<-, dotted, thick] (brj2);
\node[plainNode] (bek2) [left of=brk2,xshift=.75cm] {\large $e_i$};

\node[plainNode] (bdummynode5) at ($(bt1)!0.5!(bt2)$) {}; 
\node[plainNode] (captionB) at (captionA -| bdummynode5) {\Large (b) Events $e_i$ and $e_i'$ are reordered };

\end{tikzpicture} 
\caption{Partial dependence structure of sequence $v$. Directed edges indicate dependence
or relation by $\to_v$. Even though $p_n'$ is indicated inside task $(t_n,e_n')$,
it may only have a happens-before relation with a transition in $(t_n,e_n')$.}
\label{fig:reschedule-structure}
\end{figure}

In both cases~(a) and (b) it is intuitive to identify the event corresponding to 
the blocked task that happens before $p_n'$ 
for reordering with its corresponding executable event. Also, this blocked task will be in
set $candidates$ computed by Step~2 of Algorithm~\ref{algo:findtarget} invoked to reorder $p_1$ and $p_n'$ when exploring
sequence $v$. From the structure given in Figure~\ref{fig:reschedule-structure}(a), reordering tasks
$(t_n,e_n)$ and $(t_n,e_n')$ seems to reorder $p_1$ and $p_n'$ without disturbing the 
happens-before ordering between any $p_i'$ and $p_{i+1}$, for $1 \leq i < n$. Now assume 
tasks $(t_n,e_n)$ and $(t_n,e_{n}')$ to contain a pair of dependent transitions,
say $q$ and $q'$, in which case reordering these tasks so as to reorder $p_1$ and $p_n'$
breaks the ordering between $q$ and $q'$. In such a scenario reordering events
$e_i$ -- $e_i'$, for some $i \in [1,n-1]$, such that the corresponding
tasks of these event pairs do not have dependent transitions, would aid in reordering 
$p_1$ and $p_n'$ without affecting any other pairs of dependent transitions (see Figure~\ref{fig:reschedule-structure}(b)). 
However, identifying one right pair of events for reordering among various 
available relevant pairs of events is hard, as the dependent transitions that may
be affected by the reordering of a pair of events may not even be present
in the handlers of these events. Hence, we have designed EM-DPOR to reorder
all the relevant pairs of events.

\paragraph{Insights on reordering relevant event pairs.} 
In case of scenario presented for Figure~\ref{fig:reschedule-structure}(a) EM-DPOR
eventually explores every thread $t_i$ for $1 \leq i \leq n$ from state $s$.
This is because exploring any thread $t_i$ from $s$ eventually explores a sequence 
where the order between transitions $p_{i-1}'$ and $p_i$ (or $p_n'$ and $p_1$) is reversed
compared to what is required, while the remaining blocked to executable task happens-before
mapping is as required. As a result \texttt{FindTarget} adds thread $t_{i-1}$ to $backtrack(s)$
eventually exploring it. Even after exploring every $t_i$, $1 \leq i \leq n$, 
from state $s$, one pair of transitions from executable and blocked tasks respectively
on different threads are out of order. \texttt{FindTarget} invoked to reorder this pair
finds threads corresponding to all tasks in $candidates$ to be explored from $s$
resulting in a call to \texttt{ReschedulePending}. Then, \texttt{ReschedulePending}
identifies relevant blocked -- executable event pairs for reordering by checking
for happens-before mapping from blocked tasks to executable tasks such that the
threads corresponding to these tasks are already explored from $s$. The details 
of this process is explained below.

{ 
\SetEndCharOfAlgoLine{}
\SetInd{6pt}{0.01pt}
\begin{algorithm*}[t]
\SetKwFunction{dom}{dom}
\SetKwFunction{index}{index}
\SetKwFunction{findTarget}{FindTarget}

\SetKwInput{KwInput}{Input}
\KwInput{a transition sequence \boldmath$w$\unboldmath$: r_1 \ldots r_n$,
a set of tasks called \boldmath$pending$ \unboldmath whose events are posted in $w$, and
a transition \boldmath$r$ \unboldmath from $w$}

Let $i$ $=$ $index(w,r)$ and $s$ $=$ $pre(w,i)$ \tcp*[h]{Step 3a: Initialization}\;

Let $(t_k,e_k)$ be any task in $pending$\;
\nllabel{line-rp-pick-a-pending}
Let $worklist = \{(t_k, executable(s,t_k))\}$ and $swapMap[t_k] := \{e_k\}$\;
\nllabel{line-rp-init-worklist}
\tcp{Step 3b: Identify blocked events to be reordered with executable events at state $s$ }      
\While{$worklist \neq \emptyset$ \nllabel{rp-line-while}}{
  \Indp
  Remove a task $(t_j,e_j)$ from $worklist$\; 
  \nllabel{rp-line-remove-exec-task}
  Let $C = \{(t,e) \mid \exists l \in dom(w): (t,e) = task(r_l) \text{ and }
  e \in blockedEv(s,t)$ $\text{ and } l \to_w getEnd(w,e_j) \text{ and } t \in done(s) \}$\;
  \nllabel{rp-line-compute-c}
         
  \ForEach{$(t,e) \in C$ \nllabel{line-iterate-c}}{
       \Indp
       $worklist = worklist \cup \{(t,executable(s,t))\}$\;
       \nllabel{rp-line-update-worklist}
       Add event $e$ to the set $swapMap[t]$\;  
       \nllabel{rp-line-update-swapmap}
   }
   \nllabel{line-end-iterate-c}
}
\nllabel{rp-line-end-while}
\tcp{Step 3c: Reorder blocked and executable events identified by Step 3b}
\ForEach{thread $t$ such that $swapMap[t] \neq \emptyset$ \nllabel{rp-line-process-swapMap} }{
\Indp
  Let $e$ be any event in $swapMap[t]$\;
  Let $r$ $=$ $getPost(w,executable(s,t))$ and $r'$ $=$ $getPost(w,e)$\;
  \nllabel{rp-line-redefine}
  \findTarget{$w,r,r'$}\;
  \nllabel{rp-line-recursive-findTarget}
}
\nllabel{rp-line-end-process-swapMap}
\caption{\texttt{ReschedulePending}}
\label{algo:reschedulepending}
\end{algorithm*}
}

\paragraph{Algorithm ReschedulePending.} 
Algorithm~\ref{algo:reschedulepending} takes a sequence $w$ explored by EM-DPOR, 
a set of tasks $pending$ (same as $pending$ computed by \texttt{FindTarget}), 
and a transition $r = r_i$ identified by \texttt{FindTarget} to be reordered with
a transition $r'$ as input. 
Since \texttt{ReschedulePending} is invoked by the step~3 of \texttt{FindTarget}
(Algorithm~\ref{algo:findtarget}), we refer to the steps of \texttt{ReschedulePending}
as 3a, 3b and 3c.
In Algorithm~\ref{algo:reschedulepending}, variable $worklist$ stores a subset of 
executable tasks in state $s$, and $swapMap$ maintains a map from threads to
a subset of events blocked on their respective queues at $s$.
Lines~\ref{line-rp-pick-a-pending} and~\ref{line-rp-init-worklist} in Step~3a pick any task $(t_k,e_k)$ 
from set $pending$ passed as argument,
initialize $worklist$ with the executable task on thread $t_k$ and add $e_k$ to
the set of blocked events maintained for thread $t_k$ in $swapMap$. 
Step~3b (lines~\ref{rp-line-while}--\ref{rp-line-end-while}) initiated
by a non-empty $worklist$ identifies other relevant blocked events for reordering. 
This is required as it is hard to pick exactly one pair of relevant \emph{blocked -- executable}
events for reordering, as explained earlier.
Line~\ref{rp-line-remove-exec-task} removes some executable task $(t_j,e_j)$ from 
the $worklist$. Line~\ref{rp-line-compute-c} computes a set $C$ of tasks blocked in $s$ 
such that, a blocked task $(t,e)$ is added to $C$ if there exists a transition $r_l$ 
in the handler of $e$ which happens before a transition in $(t_j,e_j)$. This
essentially checks for the blocked task on one thread to executable task on another thread
happens-before pattern, illustrated through Figure~\ref{fig:reschedule-structure}.
Additionally, line~\ref{rp-line-compute-c} only retains those blocked tasks whose 
threads are already explored from state $s$. 
Lines~\ref{line-iterate-c}--\ref{line-end-iterate-c} iterate on each blocked task in $C$, add
corresponding executable task to $worklist$ for further processing
and store the event corresponding to blocked task in $swapMap$.
We note that in case of scenario presented for sequence $v$ in Figure~\ref{fig:reschedule-structure},
if $t_n \in done(s)$ then, \texttt{FindTarget} called to reorder dependent transitions
$p_1$ and $p_n'$ reach Step~3, compute $pending = \{(t_n,e_n')\}$ and invoke
\texttt{ReschedulePending}. Step~3a of Algorithm~\ref{algo:reschedulepending} adds 
event $e_n'$ corresponding to a \emph{pending} task $(t_n,e_n')$ to the set $swapMap[t_n]$ 
and initializes $worklist$ with the executable task $(t_n,e_n)$. Initiated 
by the executable task $(t_n,e_n)$, Step~3b iteratively adds $e_i'$ to $swapMap[t_i]$
and $(t_i,e_i)$ to $worklist$ starting from $i = n-1$ to $i = 1$. The \texttt{while}
loop exits on processing executable task $(t_1,e_1)$ and not finding any more blocked events
satisfying the constraints in line~\ref{rp-line-compute-c}.

Lines~\ref{rp-line-process-swapMap}--\ref{rp-line-end-process-swapMap} (Step~3c) 
iterate over each thread $t$ for which the set of blocked events $swapMap[t]$ is non-empty,
pick an event among events in set $swapMap[t]$, 
and invoke \texttt{FindTarget} to reorder the \texttt{post} transition
for the executable event at state $s$ on thread $t$ 
with that of the selected blocked event. 

\begin{example}
\upshape
Continuing Example~\ref{ex:pending}, Step~2 in \texttt{FindTarget} called to
reorder $r_3$ and $r_6$ (Figure~\ref{fig:motivation-mt}) fails to 
add any backtracking choices at state $s_2$ (Figure~\ref{fig:motivation-statespace-mt}).
Then, Step~3 computes $pending = \{(t_1,e_2)\}$ as $e_2$ is blocked in $s_2$, and invokes 
\texttt{ReschedulePending($r_1\ldots r_5,\{(t_1,e_2)\},r_3$)}. Line~\ref{line-rp-init-worklist}
in Algorithm~\ref{algo:reschedulepending} adds $e_2$ to $swapMap[t_1]$.
Step~3b adds no more blocked events to $swapMap$.
Step~3c calls \texttt{FindTarget($r_1\ldots r_5,r_1,r_2$)} to reorder blocked event $e_2$
with executable event $e_1$ at state $s_2$ on $t_1$. 
In the \emph{recursive call}, state $s_0$ (where $r_1$ is executed) is identified 
as the backtracking point and Step~2 adds thread $t_3$ 
to $backtrack(s_0)$ as $t_3$ executes $r_2$. Thus in a future run where $r_2$ is explored before $r_1$, 
$r_3$ and $r_6$ get reordered as shown in Figure~\ref{fig:motivation-statespace-mt}.
\end{example}

{ 
\SetEndCharOfAlgoLine{}
\SetInd{6pt}{0.01pt}
\begin{algorithm*}[t]
\small
 \SetKwFunction{dom}{dom}
 \SetKwFunction{index}{index}
 \SetKwFunction{trans}{trans}
 
 \SetKwInput{KwInput}{Input}
\KwInput{a transition sequence \boldmath$w$\unboldmath$: r_1 \ldots r_n$, 
an index \boldmath$i$\unboldmath$\;\in dom(w)$ and a transition \boldmath$r'$ \unboldmath such that 
\texttt{FindTarget} failed to reorder transitions $r_i$ and $r'$}

Let $\rightsquigarrow$ $=$ $\to_w$ \tcp*[h]{Initialize a new relation with existing members in $\to_w$}\;
\nllabel{backtrack-line-copy-hb}
\ForEach{$k \in \{2,\cdots,i-1\}$ in the increasing order}{
\Indp
\tcp{Reorder nearest co-enabled \post\ operations}
Let $\hat{w} = r_1 \ldots r_{k-1}$\;
\nllabel{backtrack-begin-iterate}
\If{$r_k$ is a \emph{\texttt{post}} operation and
$\exists j = max(\{ j \in dom(\hat{w}) \mid 
r_j \text{ is a \emph{\texttt{post}} operation and } dest(r_j) = dest(r_k) 
  \text{ and } j \not\rightsquigarrow k\})$}{
\Indp
  \nllabel{line-post-post-dep}
  Let $ts = \{ t \in enabled(pre(w,j)) \mid t = thread(r_k) \text{ or }
(\exists l \in dom(\hat{w}): l > j \text{ and } l \rightsquigarrow k \text{ and } t = thread(r_l))\}$\;
  \nllabel{line-dpor-like-ts}
  \lIf{$ts \neq \emptyset$}{
    Add any $t \in ts$ to $backtrack(pre(w,j))$ 
    \textbf{else} Add all $t \in enabled(pre(w,j))$ to $backtrack(pre(w,j))$\;
   \nllabel{line-end-adding-post-backtrack}}
   $RP(pre(w,j)) = RP(pre(w,j)) \cup \{(r_j,r_k)\}$\;
   \nllabel{backtrack-line-update-rp-1}
Add $j \rightsquigarrow k$ to the relation $\rightsquigarrow$ and close it by transitivity and FIFO rules in Definition~\ref{def:happens-before}\;
\nllabel{line-add-new-hb}
\lIf{$r_i$ happens before $r'$ by $\rightsquigarrow$ }{\textbf{return}}
\nllabel{line-return-on-hb}
}
}
\tcp{Compute backtracking choices to reorder $r$ and $r'$ using extended happens-before relation $\rightsquigarrow$}
Let $s = pre(w,i)$\;
\nllabel{line-final-part}
\eIf{$r' \in nextTrans(last(w))$
\nllabel{backtrack-line-conditional}
}
    {
\Indp
$candidates = \{ t \in enabled(s)\,|\, t = thread(r') \text{ or } (\exists l \in  
    dom(w): l > i \text{ and } l \rightsquigarrow task(r') \text{ and } t = thread(r_l)) \}$\;
    \nllabel{bactrack-line-candidate-1}
    }
    {
\Indp
$candidates = \{ t \in enabled(s) \mid t = thread(r') \text{ or } 
    (\exists l \in dom(w): l > i \text{ and } l \rightsquigarrow index(w,r')
    \text{ and } t = thread(r_l))\}$\;
    \nllabel{backtrack-line-candidate-2}
    }
\nllabel{backtrack-line-postconditional}

\lIf{$candidates \neq \emptyset$}{
  Add any $t \in candidates$ to $backtrack(s)$\;
  \nllabel{backtrack-line-add-choice}
  }
\lElse{ Add all $t \in enabled(s)$ to $backtrack(s)$\;
  \nllabel{backtrack-line-add-all}
}
\nllabel{line-end}
\lIf{$r_i$ and $r'$ are \emph{\texttt{post}} operations}{$RP(s) = RP(s) \cup \{(r_i,r')\}$\;}
\nllabel{backtrack-line-update-rp-2}

 \caption{\texttt{BacktrackEager}} 
\label{algo:backtrackeager}
\end{algorithm*}
}

\subsection{Simulating DPOR}
\label{sec:backtrackeager}

Call to \texttt{BacktrackEager($w,i,r'$)} is performed by line~\ref{line-call-backtrackeager} in Algorithm~\ref{algo:findtarget} 
when Steps~2 and~3 of Algorithm~\ref{algo:findtarget} fail to identify backtracking choices to reorder transitions 
$r$ (same as $r_i$) executed at a state $s$ and a transition $r'$. 
When the DPOR algorithm fails to identify candidate threads using the HB relation so as to 
reorder a pair of racing transitions in the multi-threaded setting, it includes
all the threads enabled at $s$ as backtracking choices, initiating exploration of
all thread interleavings rooted at $s$. In our event-driven setting, in addition,
EM-DPOR must initiate all possible reordering of events in each queue which are posted prior to reaching state $s$.
\texttt{BacktrackEager} (Algorithm~\ref{algo:backtrackeager}) achieves the same.

It initializes a temporary HB relation $\rightsquigarrow$ which will only be used
in the current invocation of \texttt{BacktrackEager}, with the HB ordered pairs
in the relation $\to_w$.
Given a transition sequence $w$, an index $i$ and a transition $r'$,
\texttt{BacktrackEager} treats every nearest pair $(r_j,r_k)$ of transitions 
with no happens-before between them as per $\rightsquigarrow$, and posting to the same event queue
as \emph{dependent}, provided $j,k < i$ (Algorithm~\ref{algo:backtrackeager} line~\ref{line-post-post-dep}). 
We consider $r_j$ to be \emph{nearest} to $r_k$ if $j < k$ and $r_j$ has the highest index
in $w$ among all other transitions satisfying the given constraints.
\texttt{BacktrackEager} then simulates the DPOR approach with this dependence relation from 
the initial state along $w$ up to $r_i$.
Note that dependence through shared objects is already considered in Algorithm~\ref{algo:explore}.
Lines~\ref{line-dpor-like-ts}--\ref{line-add-new-hb}
add backtracking choices at state $pre(w,j)$ to reorder $r_j$ and $r_k$, and mark 
$r_j$ to happen before $r_k$. The new happens-before mapping added to $\rightsquigarrow$ induces 
additional transitive and FIFO mappings to be added to $\rightsquigarrow$ 
(see line~\ref{line-add-new-hb} in Algorithm~\ref{algo:backtrackeager}).
Hence, we call $\rightsquigarrow$ as the \emph{extended} HB relation.
If $r_i$ is established to happens before $r'$ as per $\rightsquigarrow$, then \texttt{BacktrackEager} returns 
(line~\ref{line-return-on-hb}), because $r_i$ and $r'$ have got related by 
happens-before by considering a pair of \post\ operations $(r_j,r_k)$
as dependent. Thus, $r_i$ and $r'$ will get reordered when $r_j$ and $r_k$ 
get reordered on exploring backtracking choices added by line~\ref{line-end-adding-post-backtrack}. 
Otherwise, the algorithm iterates until $r_i$ is reached, and computes backtracking choices 
to reorder $r_i$ and $r'$ similar to DPOR (lines~\ref{line-final-part}--\ref{line-end})
using the extended HB relation $\rightsquigarrow$. 
Lines~\ref{backtrack-line-update-rp-1} and \ref{backtrack-line-update-rp-2} update the $RP$ sets of different states since the \post{} 
transitions executed from these states were identified to be reordered w.r.t. \texttt{post}s
executed later. As explained earlier $RP$ sets will be queried to identify the set of
\reorderedPosts\ in subsequent explorations.
Below is an example illustrating the working of \texttt{BacktrackEager}.

\begin{example}
\label{ex:backtrackeager}
\upshape
For the purpose of this example, consider an implementation of EM-DPOR which does not track happens-before 
ordering between a \texttt{fork} operation and the initialization of the spawned thread.
Assume exploring a sequence
$w$ given in Figure~\ref{fig:motivation-mt} with such an implementation of EM-DPOR. On
reaching state $s_5$ (see Figure~\ref{fig:motivation-statespace-mt}) \texttt{Explore} invokes \texttt{FindTarget}
to reorder dependent transitions $r_3$ and $r_6$. 
As thread $t_4$ executing $r_6$ is not enabled at $s_2$
and missing happens-before mapping between $r_4$ and $r_5$
causes $candidates$ computed on line~\ref{line-candidate-1} of Algorithm~\ref{algo:findtarget}  
to be an empty set. Set $pending$ 
is also empty as it is a subset of $candidates$. This causes the control flow of \texttt{FindTarget}
to reach Step~4 invoking \texttt{BacktrackEager($r_1\ldots r_5, 3, r_6$)}. Then,
lines~\ref{line-post-post-dep}--\ref{line-end-adding-post-backtrack} in Algorithm~\ref{algo:backtrackeager}
pick transitions $r_1$ and $r_2$ posting events to the same event queue, as the 
nearest co-enabled \texttt{post}s not ordered by $\rightsquigarrow$,
and add $t_3$ executing $r_2$ to $backtrack(s_0)$. This is because $r_1$ is explored at $s_0$ in $w$.
On backtracking to $s_0$, EM-DPOR explores a run where events $e_1$ and $e_2$ are reordered
which eventually reorders $r_3$ and $r_6$ as shown in Figure~\ref{fig:motivation-statespace-mt}.
\end{example}

\begin{figure*}[t]
\begin{tabular*}{\textwidth}{l@{\;\;\;}r}
\begin{minipage}{\dimexpr0.45\textwidth-2\tabcolsep}
\centering
\begin{tabular}{@{}r@{\;\;\;\;\;}l@{\;}l@{\;}l@{\;}l@{\;}l@{\;}l@{\;}l@{\;}l@{}}
 & \tn{t0}{$t_0$} & \tn{t1}{\;\;\;$t_1$}  & \tn{t2}{\;\;\;$t_2$} & \tn{t3}{\;\;\;$t_3$} & \tn{t4}{\;\;\;$t_4$} 
 & \tn{t5}{\;\;\;$t_5$} & \tn{t6}{\;\;\;$t_6$}\\
\hline
\\
$r_1$ &&&& \multicolumn{4}{@{}l}{\tn{n1}{\;\;\;\post{($e_1$)}}} \\\myhline
$r_2$ &&&&& \multicolumn{3}{@{}l}{\tn{n2}{\;\;\;\post{($e_2$)}}} \\\myhline
$r_3$ &&&&&&  \multicolumn{2}{@{}l}{\tn{n3}{\;\;\;\post{($e_3$)}}} \\\myhline
$r_4$ &&&&&&& \tn{n4}{\;\;\post{($e_4$)}} \\\myhline
\\\myhline
$r_5$ &&& \multicolumn{5}{@{}l}{\tn{n5}{\;\;\;\post{($e_5$)}}} \\\myhline
\\\myhline
$r_6$ && \multicolumn{6}{@{}l}{\tn{n6}{\;\;\;\post{($e_6$)}}} \\\myhline
$r_7$ && \multicolumn{6}{@{}l}{\tn{n7}{\;\;\;\fork{($t_0$)}}} \\\myhline
\\\myhline
$r_8$ & \multicolumn{7}{@{}l}{\tn{n8}{\texttt{b = 1}}} \\\myhline
\\\myhline
$r_9$ && \multicolumn{6}{@{}l}{\tn{n9}{\;\;\;\texttt{y = 5}}} \\\myhline
$r_{10}$ && \multicolumn{6}{@{}l}{\tn{n10}{\;\;\;\texttt{b = 10}}} \\\myhline
\\\myhline
$r_{11}$ &&& \multicolumn{5}{@{}l}{\tn{n11}{\;\;\;\readX{(y)}}} \\\myhline
\\
\end{tabular}
\tikz[remember picture,overlay] \node[] (fit1) at (t2 |- n5) {};
\tikz[remember picture,overlay] \node[] (fit2) at (t4 |- n5) {};
\tikz[remember picture,overlay] \node[fit=(fit1)(fit2), draw] {};
\tikz[remember picture,overlay] \node[] (fit3) at (t2 |- n11) {};
\tikz[remember picture,overlay] \node[] (fit4) at (t4 |- n11) {};
\tikz[remember picture,overlay] \node[fit=(fit3)(fit4), draw] {};
\tikz[remember picture,overlay] \node[] (fit5) at (t1 |- n6) {};
\tikz[remember picture,overlay] \node[] (fit6) at (t3 |- n7) {};
\tikz[remember picture,overlay] \node[fit=(fit5)(fit6), draw] {};
\tikz[remember picture,overlay] \node[] (fit7) at (t1 |- n9) {};
\tikz[remember picture,overlay] \node[] (fit8) at (t3 |- n10) {};
\tikz[remember picture,overlay] \node[fit=(fit7)(fit8), draw] {};
\tikz[remember picture,overlay] \node[left of=fit1,xshift=.4cm] (e3) {$e_3$};
\tikz[remember picture,overlay] \node[left of=fit3,xshift=.4cm] (e4) {$e_4$};
\tikz[remember picture,overlay] \node[left of=fit5,xshift=.4cm] (e1) {$e_1$};
\tikz[remember picture,overlay] \node[left of=fit7,xshift=.4cm] (e2) {$e_2$};
\end{minipage}%
&
\begin{minipage}{\dimexpr0.55\textwidth-2\tabcolsep}
\centering

\begin{tikzpicture}[auto,node distance=16 mm,semithick, scale=0.675, transform
shape]
\tikzstyle{eventNode} = [thick,inner sep=0pt,minimum size=6mm]
\tikzstyle{myarrows}=[line width=.5mm,draw=black,-triangle 45,postaction={draw, line width=1.2mm, shorten >=3mm, -}]
\tikzstyle{stateNode} = [circle,draw=black,thick,inner sep=0pt,minimum 
size=6mm]
\tikzstyle{sequenceNode} = [circle,fill=yellow,draw=black,thick,inner sep=0pt,minimum 
size=6mm]
\tikzstyle{blankNode} = [thick,inner sep=0pt,minimum size=6mm]

\node[sequenceNode] (s0) {$s_0$};


\node[sequenceNode] (s1) [below of=s0,yshift=-.6cm] 
{$s_1$} edge [<-] node [left] {$r_1.r_2.r_3.r_4$} (s0);

\node[blankNode] (q1t1) [left of=s1, xshift=.35cm,yshift=.5cm] {
\begin{tikzpicture}[node distance=-\pgflinewidth]
\node [mybox, fill=gray!20] (e1) {$e_1$};
\node [mybox,fill=gray!20,right=of e1] (e2) {$e_2$};
\end{tikzpicture}};
\node[blankNode] (q1t1label) [left of=q1t1,xshift=.7cm] {$t_1$:};

\node[blankNode] (q1t2) [below of=q1t1, yshift=.9cm] {
\begin{tikzpicture}[node distance=-\pgflinewidth]
\node [mybox, fill=gray!20] (e3) {$e_3$};
\node [mybox,fill=gray!20,right=of e3] (e4) {$e_4$};
\end{tikzpicture}};
\node[blankNode] (q1t2label) [left of=q1t2,xshift=.7cm] {$t_2$:};

\node[sequenceNode] (s2) [below of=s1,yshift=.3cm] 
{$s_2$} edge [<-] node [left] {$r_5$} (s1);
\node[sequenceNode] (s3) [below of=s2,yshift=.3cm] 
{$s_3$} edge [<-] node [left] {$r_6.r_7$} (s2);
\node[sequenceNode] (s4) [below of=s3,yshift=.3cm] 
{$s_4$} edge [<-] node [left] {$r_8$} (s3);
\node[sequenceNode] (s5) [below of=s4,yshift=.3cm] 
{$s_5$} edge [<-] node [left] {$r_9.r_{10}$} (s4);
\node[sequenceNode] (s6) [below of=s5,yshift=.3cm] 
{$s_6$} edge [<-] node [left] {$r_{11}$} (s5);

\node[stateNode] (s7) [below of=s0,yshift=-.6cm,xshift=-4.7cm] 
{$s_7$} edge [<-] node [left, yshift=.2cm] {$r_1.r_2.r_4.r_3$} (s0);
\node[blankNode] (q7t1) [left of=s7, xshift=.35cm,yshift=.75cm] {
\begin{tikzpicture}[node distance=-\pgflinewidth]
\node [mybox, fill=gray!20] (e1) {$e_1$};
\node [mybox,fill=gray!20,right=of e1] (e2) {$e_2$};
\end{tikzpicture}};
\node[blankNode] (q7t1label) [left of=q7t1,xshift=.7cm] {$t_1$:};

\node[blankNode] (q7t2) [below of=q7t1, yshift=.9cm] {
\begin{tikzpicture}[node distance=-\pgflinewidth]
\node [mybox, fill=gray!20] (e4) {$e_4$};
\node [mybox,fill=gray!20,right=of e4] (e3) {$e_3$};
\end{tikzpicture}};
\node[blankNode] (q7t2label) [left of=q7t2,xshift=.7cm] {$t_2$:};
 
\node[stateNode] (s9) [below of=s7,yshift=.3cm,xshift=-1cm] 
{$s_9$} edge [<-] node [left] {\scalebox{2}{$r_6$}$.r_7$} (s7);
\node[stateNode] (s10) [below of=s9,yshift=.3cm,xshift=-1cm] 
{$s_{10}$} edge [<-] node [left] {$r_8$} (s9);
\node[stateNode] (s11) [below of=s10,yshift=.3cm,xshift=-.7cm] 
{$s_{11}$} edge [<-] node [left] {$r_9.r_{10}$} (s10);
\node[stateNode] (s12) [below of=s11,yshift=.3cm] 
{$s_{12}$} edge [<-] node [left] {$r_{11}$} (s11);
\node[stateNode] (s13) [below of=s12,yshift=.3cm] 
{$s_{13}$} edge [<-] node [left] {\scalebox{2}{$r_5$}} (s12);

\node[blankNode] (s14) [below of=s10,yshift=.3cm,xshift=.7cm] 
{\scalebox{2}{$\mathbf{\cdotp \, \cdotp \, \cdotp}$}} edge [<-, dashed] (s10);

\node[blankNode] (s17) [below of=s9,yshift=.3cm,xshift=.8cm] 
{\scalebox{2}{$\mathbf{\cdotp \, \cdotp \, \cdotp}$}} edge [<-, dashed] (s9);
%

\node[stateNode] (s25) [below of=s7,yshift=.3cm,xshift=1cm] 
{$s_{14}$} edge [<-] node [right] {\scalebox{2}{$r_{11}$}} (s7);
\node[stateNode] (s26) [below of=s25,yshift=.3cm] 
{$s_{15}$} edge [<-] node [left] {$r_5$} (s25);
\node[stateNode] (s27) [below of=s26,yshift=.3cm] 
{$s_{16}$} edge [<-] node [left] {$r_6.r_7$} (s26);
\node[stateNode] (s28) [below of=s27,yshift=.3cm,xshift=-.6cm] 
{$s_{17}$} edge [<-] node [left] {$r_8$} (s27);
\node[stateNode] (s29) [below of=s28,yshift=.3cm] 
{$s_{18}$} edge [<-] node [left] {\scalebox{2}{$r_9$}$.r_{10}$} (s28);

\node[blankNode] (s30) [below of=s27,yshift=.3cm,xshift=.6cm] 
{\scalebox{2}{$\mathbf{\cdotp \, \cdotp \, \cdotp}$}} edge [<-, dashed] (s27);


\node[stateNode] (s8) [below of=s0,yshift=-.6cm,xshift=2.1cm] 
{$s_8$} edge [<-] node [right, yshift=.2cm] {$r_2.r_1.r_4.r_3$} (s0);

\node[blankNode] (q8t1) [right of=s8, xshift=.2cm,yshift=.6cm] {
\begin{tikzpicture}[node distance=-\pgflinewidth]
\node [mybox, fill=gray!20] (e2) {$e_2$};
\node [mybox,fill=gray!20,right=of e2] (e1) {$e_1$};
\end{tikzpicture}};
\node[blankNode] (q8t1label) [left of=q8t1,xshift=.7cm] {$t_1$:};

\node[blankNode] (q8t2) [below of=q8t1, yshift=.9cm] {
\begin{tikzpicture}[node distance=-\pgflinewidth]
\node [mybox, fill=gray!20] (e4) {$e_4$};
\node [mybox,fill=gray!20,right=of e4] (e3) {$e_3$};
\end{tikzpicture}};
\node[blankNode] (q8t2label) [left of=q8t2,xshift=.7cm] {$t_2$:};

\node[stateNode] (s32) [below of=s8,yshift=.3cm] 
{$s_{19}$} edge [<-] node [left] {$r_9.$\scalebox{2}{$r_{10}$}} (s8);
\node[stateNode] (s33) [below of=s32,yshift=.3cm] 
{$s_{20}$} edge [<-] node [left] {$r_{11}$} (s32);
\node[stateNode] (s34) [below of=s33,yshift=.3cm] 
{$s_{21}$} edge [<-] node [left] {$r_5$} (s33);
\node[stateNode] (s35) [below of=s34,yshift=.3cm] 
{$s_{22}$} edge [<-] node [left] {$r_6.r_7$} (s34);
\node[stateNode] (s36) [below of=s35,yshift=.3cm] 
{$s_{23}$} edge [<-] node [left] {\scalebox{2}{$r_8$}} (s35);
\node[blankNode] (s37) [below of=s33,yshift=.3cm,xshift=1.5cm] 
{\scalebox{2}{$\mathbf{\cdotp \, \cdotp \, \cdotp}$}} edge [<-,dashed] (s33);
%
\node[blankNode] (s40) [below of=s8,yshift=.3cm,xshift=1.5cm] 
{\scalebox{2}{$\mathbf{\cdotp \, \cdotp \, \cdotp}$}} edge [<-,dashed] (s8);
%

\node[blankNode] (z1) [below of=s13, yshift=.7cm] {\scalebox{1.6}{$z_1$}};
\node[blankNode] (z2) [below of=s29, yshift=.7cm] {\scalebox{1.6}{$z_2$}};
\node[blankNode] (z) [below of=s6, yshift=.7cm] {\scalebox{1.6}{$z$}};
\node[blankNode] (z) [below of=s36, yshift=.7cm] {\scalebox{1.6}{$z_3$}};

\end{tikzpicture}

\end{minipage}
\\
\begin{minipage}[t]{\dimexpr0.45\textwidth-2\tabcolsep}
\captionof{figure}{A partial trace $z$ of an event-driven program involving
a multi-threaded dependence.}
\label{fig:rp-hb-role-trace}
\end{minipage}
&
\begin{minipage}[t]{\dimexpr0.55\textwidth-2\tabcolsep}
\captionof{figure}{A partial state space for some valid permutations of transitions
in the trace given in Figure~\ref{fig:rp-hb-role-trace}.}
\label{fig:fig:rp-hb-role-state}
\end{minipage}
\end{tabular*}
\end{figure*}

\subsection{Role of HB Order Induced Between \post\ Transitions}
We now give another example to illustrate the end-to-end 
working of EM-DPOR along with highlighting the role played by
happens-before mappings added between reordered \texttt{post} operations by
rule~3 in Definition~\ref{def:happens-before}.

\begin{example}
\label{ex:reschedulepending}
\upshape
Consider an execution trace $z$ shown in Figure~\ref{fig:rp-hb-role-trace},
of a program in which two threads $t_1$ and $t_2$ have event queues. 
Transitions $r_1$ and $r_2$ respectively post events $e_1$ and $e_2$ to the event queue of
the thread $t_1$, and the transitions $r_3$ and $r_4$ respectively post events 
$e_3$ and $e_4$ to the event queue of the thread $t_2$.
Transitions $r_5$ and $r_6$ post events $e_5$ and $e_6$
respectively to the same event queue. However, the event handlers 
corresponding to $e_5$ and $e_6$ are not shown in the figure.
We assume that the event handlers 
of $e_5$ and $e_6$ contain dependent transitions. 
Figure~\ref{fig:fig:rp-hb-role-state} shows a partial state space explored by various
permutations of transitions in $z$. For economy of space, we merge prefixes of 
certain transition sequences and represent them by single edges.
Event queue state of threads $t_1$ and $t_2$ are indicated for some of the states reached on executing
the \post\ operations in various orders. The events in an event queue are ordered 
from left to right, which makes the leftmost event the front of the queue.
The sequences of interest are labeled as $z$, $z_1$, $z_2$ and $z_3$ in Figure~\ref{fig:fig:rp-hb-role-state}.
The shaded states correspond to states explored by $z$. Sequence $z$ has two pairs 
of \emph{may be co-enabled} dependent transitions --- $(r_8,r_{10})$ and $(r_9,r_{11})$,
and a pair of \emph{may be reordered} dependent transitions in the handlers of $e_5$ and $e_6$.


Assume that EM-DPOR initially explores sequence $z_1$ in which the relative order of 
events $e_3$ and $e_4$ is reversed compared to that in $z$. 
We show how EM-DPOR eventually explores a dependence-covering sequence of $z$, rather $z$ itself,
when the model checking starts with $z_1$. 
A dependence-covering sequence of $z$ must maintain the relative ordering of all pairs of dependent transitions in 
$z$ (see Definition~\ref{def:dep-covering}).
Clearly, $z_1$ is not a dependence-covering sequence of $z$ as the relative order of dependent
transitions in the event handlers of $e_5$ and $e_6$ posted respectively by the 
transitions $r_5$ and $r_6$, is reversed w.r.t. that in $z$. We will be showing the
pair of dependent transitions or \post{} transitions in a transition sequence $z_i$,
whose order is problematic for $z_i$ to be a dependence-covering sequence of $z$, in an enlarged form.

When exploring $z_1$, Algorithm~\ref{algo:explore} invokes \texttt{FindTarget} (Algorithm~\ref{algo:findtarget}) to
compute backtracking choices to reorder dependent transitions in the handlers of $e_6$ and $e_5$ (not shown in Figure~\ref{fig:fig:rp-hb-role-state}). 
Step~1 of \texttt{FindTarget} identifies $r_6$ and $r_5$ as corresponding diverging \texttt{post}s and recursively
invokes \texttt{FindTarget} to reorder $r_6$ and $r_5$. In the recursive call,
Step~2 of \texttt{FindTarget} adds thread $t_2$ to $backtrack(s_7)$ since $r_6$ is executed 
from state $s_7$, and EM-DPOR eventually explores a sequence $z_2$.  
Since $r_5 \in reorderedPosts(r_6,z_2)$, $r_5$ and $r_6$ are related by $\to_{z_2}$.
Again, $z_2$ is not a dependence-covering sequence of $z$ as the relative order of dependent transitions 
$r_9$ and $r_{11}$ is reversed compared to that in $z$. 
On exploring a prefix of $z_2$ till state $s_{17}$ where $r_9 = next(s_{17}, t_1)$, 
\texttt{FindTarget} is invoked to reorder $r_{11}$ and $r_9$. Step~2 of 
\texttt{FindTarget} computes $candidates = \{(t_1,e_2)\}$. Since $t_1$ is 
in $done(s_7)$ due to sequence $z_1$, Step~3 of \texttt{FindTarget} is reached which computes
$pending = \{(t_1,e_2)\}$. Then, \texttt{ReschedulePending} is invoked by
line~\ref{line-recursive-findtarget} of \texttt{FindTarget} to reorder relevant blocked events with executable events at state $s_7$.
Event $e_2$ is added to $swapMap[t_1]$ and $(t_1,e_1)$ to 
$worklist$ (line~\ref{line-rp-init-worklist} in Algorithm \texttt{ReschedulePending}). 
On processing $(t_1,e_1)$ in $worklist$, Step~3(b) of 
\texttt{ReschedulePending} adds blocked event $e_3$ to $swapMap[t_2]$
and $(t_2,e_4)$ to $worklist$, as $r_5$ in the task $(t_2,e_3)$ blocked at state $s_7$ 
happens before $r_6$ in the task $(t_1,e_1)$ executable at state $s_7$, and $t_2 \in done(s_7)$. 
No task is added to $worklist$ on processing $(t_2,e_4)$.
Then, Step~3c invokes \texttt{FindTarget} to reorder \texttt{post}s of
events $e_1$ and $e_2$ and \texttt{post}s of $e_4$ and $e_3$. 
Reordering $e_4$ and $e_3$ allows us to explore $z$ --- our target sequence. 

As mentioned earlier, arbitrarily selecting a blocked event for reordering w.r.t.
an executable event, among the 
set of blocked events identified by Steps~3a - 3b of \texttt{ReschedulePending} may not 
yield a dependence-covering sequence for a target sequence.
For example,  any sequence explored after reordering events $e_1$ and 
$e_2$ reverses the order of dependent transitions $r_8$ (executed by the thread $t_0$) and $r_{10}$
(executed by the handler of $e_2$ on $t_1$)
as shown in sequence $z_3$, making such sequences \emph{non} dependence-covering
w.r.t. $z$. This example also demonstrated the necessity to capture the ordering
between reordered \texttt{post}s. The happens-before mapping from $r_5$ to $r_6$ 
helped in identifying event $e_3$ as a relevant blocked event to be reordered with 
its corresponding executable event $e_4$, leading to the exploration of a dependence-covering
sequence of $z$.
\end{example}

\subsection{Formal Guarantees and Variants of EM-DPOR}
\label{sec:emdpor-soundness}
in Appendix~\ref{app:emdpor-proofs} we provide a sketch outlining the 
proof of correctness  of \emdpor. Through this proof sketch we show that whenever \texttt{Explore}
backtracks from a state $s$ to a prior state in the search stack, it must have 
explored a dependence-covering sequence (see Definition~\defdepcovering) for any 
sequence $w$ in $\mathcal{S}_G$ from state $s$. 
This equivalently proves that \emdpor\ explores a dependence-covering set at each visited
state $s$. 

Appendix~\ref{app:emdpor-opt} discusses a few variants of the Algorithm \texttt{Explore}
capable of identifying more pairs of independent transitions than assumed in this section (see the beginning of 
Section~\ref{sec:emdpor-algo}). We have incorporated 
these optimizations in our \emdpor\ implementation used for experimental evaluation of \emdpor.
%
%

\section{Implementation}
\label{sec:emdpor-impl}

This section describes a vector clock based implementation of EM-DPOR on a prototype 
stateless model checking framework called \emexplorer. Since we evaluate \emdpor\
over Android application traces, \emexplorer\ has been designed to handle the concurrency
behavior of Android applications.

\subsubsection*{Vector Clock Based Implementation of \emdpor}
Happens-before relation (see Definition~\ref{def:happens-before}) over a  
given transition sequence which in turn captures the order between dependent transitions  
in the sequence, plays a vital role in various steps of \emdpor\ such as identifying 
unordered dependent operations to be reordered, computing backtracking choices
and so on. We use vector clocks data structure to compute the happens-before 
relation. We have designed the implementation of \emdpor\ similar to the implementation
of the DPOR~\cite{Flanagan:2005:DPR:1040305.1040315} algorithm which too uses vector clocks to capture
the HB relation over traces of multi-threaded programs to 
dynamically computes persistent sets~\cite{Godefroid:1997:MCP:263699.263717}.
In a multi-threaded setting where all the operations executed on 
the same thread are totally ordered, 
each component (or clock) of a vector clock corresponds to a thread. Hence, the 
vector clock timestamp of an operation $z$ denotes the last known operation (as known by $z$) performed
by each thread of the program.
In an event-driven program, the operations from different event
handlers on the same thread need not be totally ordered. Hence in the vector clocks
we use, each clock corresponds to a task in the program where a task is either an event 
or a thread. In order to compute the
vector clock timestamps of operations of a task, we maintain a vector clock with
each task.
Most of the computations on vector clocks described in~\cite{Flanagan:2005:DPR:1040305.1040315} are lifted in a straightforward
manner to task-based vector clocks.
As defined by rule~(2) of Definition~\ref{def:happens-before}, \emdpor\ orders event handlers 
executed on the same thread if their corresponding \texttt{post}s 
have a happens-before ordering, so as to respect the FIFO ordering of events.
FIFO ordering is specific to the event-driven concurrency model 
considered in this work and is not handled by the vector clock based implementation of DPOR.
The treatment of FIFO closure 
requires a special design explained below.

\paragraph{Computing FIFO closure.}
Initially all the components (scalar clocks) of the vector clocks of all the tasks are initialized to zero.
Let $V_1$ be the vector clock of a task in which the transition with visible operation
\post{($\_$,$e$,$t$)} is executed. Let $V_2$ be the vector clock of the task $(t,e)$.
On executing \post{($\_$,$e$,$t$)}, the component $(t,e)$ in the vector clock $V_1$,
\ie\ $V_1((t,e))$, is incremented 
making this component of $V_1$ non-zero, and the vector clock $V_2$ of task $(t,e)$ 
is initialized with the same value as that of $V_1$.
After initialization $V_2$ remains unmodified till event $e$ is dequeued. When dequeuing
event $e$ we check the value of each component corresponding to events posted to the thread $t$, in vector clock $V_2$.
If the value of any such component of $V_2$, say $(t,e')$, is non-zero, we update $V_2$ by performing 
a vector clock join between $V_2$ and the 
vector clock of the task $(t,e')$. A non-zero component
value for a task $(t,e')$ in $(t,e)$'s vector clock $V_2$ indicates that \post{($\_$,$e'$,$t$)} \emph{happens-before}
\post{($\_$,$e$,$t$)}, and thus FIFO rule in Definition~\ref{def:happens-before} is applicable.
Since the event $e'$ is handled prior to $e$ on the thread $t$, the vector clock
of $(t,e')$ has a value corresponding to the VC timestamp of \aend{($t$,$e'$)} 
when it is used to update $V_2$. Thus the event handler of $e$ gets ordered w.r.t. 
that of $e'$.

\subsubsection*{\emexplorer\ Framework}
The order of execution of operations in an Android application is influenced not only by the
sources of non-determinism in the application, but also by the Android 
framework and the inter-process communication between the applications running in 
different processes on an Android device. 
Interpreting or modeling various concurrency relevant APIs and operations from
application/framework code, makes building a full fledged model checker for 
Android applications a challenge in itself.
Tools such as JPF-Android~\cite{DBLP:journals/sigsoft/MerweMV12} and
AsyncDroid~\cite{DBLP:conf/cav/OzkanET15} take promising steps in this direction. However,
presently they either explore only a limited number of sources of
non-determinism~\cite{DBLP:conf/cav/OzkanET15}
or require a lot of framework libraries to be modeled~\cite{DBLP:journals/sigsoft/MerweMV12,DBLP:conf/icse/Merwe15}.
We have therefore built a prototype exploration 
framework called EM-Explorer, which emulates the semantics of visible 
operations like \texttt{post}, \texttt{read}, \texttt{acquire} and so on.

Our framework takes an execution trace generated 
by an automated testing and race detection tool for Android applications, called 
\droidracer~\cite{Maiya:2014:RDA:2594291.2594311}, as input. 
Since \droidracer\ has the capability to run on real-world applications, we can 
experiment on real concurrency behaviors seen in Android applications and evaluate 
different POR techniques on them.
\droidracer\ records all concurrency relevant operations and memory \texttt{read}s and \texttt{write}s.
EM-Explorer \emph{emulates} such a trace based on their operational semantics and
explores all interleavings of the given execution trace permitted by the semantics. 
Android permits user and system-generated events apart from programmatically generated events 
by the application. 
\emexplorer\ only explores the non-determinism between program and system generated events 
while keeping the order of user events fixed.
This is analogous to model checking w.r.t. a fixed data input.
EM-Explorer does not track 
variable values and is incapable of evaluating conditionals on a different interleaving of the trace. 
\emexplorer\ is a stateless model checker, \ie\ it does not store program states which
can be restored when backtracking to a state. Hence, backtracking is performed by re-executing the prefix
of the last explored sequence upto the backtracking point.

Android supports different types of component classes, \eg\ Activity class
for user interface, and enforces a happens-before ordering between 
handlers of lifecycle events of component classes. \emexplorer\ seeds the 
happens-before relation for such events in each trace before starting the model checking,
to avoid exploring invalid interleavings of lifecycle events.
Android applications may post events in different modes such as associating a delay with an event
or posting an event to the front of the queue.  
We over-approximate the effect of posting with delay by forking
a new thread which does the \post{} non-deterministically, 
as mentioned in Section~\ref{sec:transition-system-por}. We leave handling of other
variants of posting events as future work.
We subject the execution trace generated by \droidracer\ to post-processing. 
Specifically, we recursively remove \emph{empty} event handlers (event handlers which only
execute \abegin{} and \aend{} with either no other visible operations in between
or only posting events whose event handlers are empty) 
from the traces obtained from \droidracer\ before model checking. This is done to 
facilitate fair comparison with DPOR which does not
inspect the contents of the handlers before reordering events. DPOR would otherwise 
unnecessarily reorder even such events.

\section{Experimental Evaluation}
\label{sec:eval}

\renewcommand{\arraystretch}{1}
\begin{table}[t]
\caption{Statistics on execution traces from Android applications}
\centering
\scalebox{0.8}{
\begin{tabular}{|@{\;}c@{\;}||@{\;}r@{\;\;}@{\;\;}r@{\;\;}@{\;\;}r@{\;\;}@{\;\;}r@{\;}@{\;}|}
\hline
{Application} & {Trace length} & {Threads} & {Events} & {Memory locations}\\
\hline
Remind Me             & $444$ & $4$ & $9$ & $89$ \\
My Tracks             & $453$ & $10$ & $9$ & $108$ \\
Music Player          & $465$ & $6$ & $24$ & $68$  \\
Character Recognition & $485$ & $4$ & $22$ & $40$ \\
Aard Dictionary       & $600$ & $5$ & $30$ & $30$ \\
\hline
\end{tabular}
}
\label{tab:stats}
\end{table}

\begin{table*}[t]
\vspace{-12pt}
\caption{Statistics on model checking runs using different POR techniques}
\centering
\scalebox{0.8}{
\begin{tabular}{|@{\;}c@{\;}||r@{\;\;\;}r@{\;\;\;}r@{\;\;}|r@{\;\;\;}r@{\;\;\;}r@{\;\;}|}
\hline
\multirow{2}{*}{{Application}} & \multicolumn{3}{c|}{\textbf{DPOR}} & 
\multicolumn{3}{c|}{\textbf{EM-DPOR}}\\
 & Traces & Transitions & Time & Traces & Transitions & Time \\
\hline
Remind Me             & $24$ & $1864$ & $0.18$s & $3$  & $875$ & $0.05$s \\
My Tracks             & $1610684^*$ & $113299092^*$ & $4$h$^*$ & $405013$  & $26745327$ & $101$m $30$s \\
Music Player          & $1508413^*$ & $93254810^*$ & $4$h$^*$ & $266$ & $34333$ & $4.15$s \\
Character Recognition & $1284788$ & $67062526$ & $199$m $28$s & $756$ & $39422$ & $6.58$s \\
Aard Dictionary        & $359961^*$ & $14397143^*$ & $4$h$^*$ & $14$ & $4772$ & $1.4$s \\
\hline
\end{tabular}
}
\label{tab:results}
\end{table*}

We evaluate the performance of EM-DPOR
which computes dependence-covering sets, by comparing with DPOR~\cite{Flanagan:2005:DPR:1040305.1040315} which computes 
persistent sets. DPOR is designed to use a dependence relation in
which transitions with operations posting to the same event queue are considered 
dependent. Whereas, \emdpor\ uses the dependence relation given in Definition~\ref{def:dependence-relation}.  
Both the algorithms are implemented in the EM-Explorer framework described in Section~\ref{sec:emdpor-impl}
and evaluated on post-processed execution traces of Android applications obtained
by running \droidracer.

We evaluated these two POR techniques on execution traces generated by DroidRacer 
on $5$ Android applications obtained from the Google Play Store~\cite{playStore}. 
Table~\ref{tab:stats} 
presents statistics like the number of \emph{visible} operations in the trace 
(which is same as the count of concurrency relevant operations logged by \droidracer), 
threads, events, threads with event loops and (shared) memory locations in the 
collected execution trace of each of these applications.
We only report the threads created by the application, and the number of events
excluding events with empty event handlers.

We analyzed each of the traces described in Table~\ref{tab:stats} using both the POR techniques.
Table~\ref{tab:results} gives the number of interleavings (listed as ``\emph{Traces}'')
and distinct transitions explored by DPOR and EM-DPOR.
It also gives the time taken for exploring the reduced state space
for each execution trace.
If a model checking run did not terminate within $4$ hours, we force-kill it and 
report the statistics for $4$ hours. The statistics for force-killed runs
are marked with $*$ in Table~\ref{tab:results}.
Since EM-Explorer does not track variable values,
it cannot prune executions that are infeasible due to conditional sequential execution.
However, both DPOR and EM-DPOR are implemented on top of EM-Explorer and therefore
operate on the same set of interleavings. The difference in their performance
thus arises from the different POR strategies.

In our experiments, DPOR's model checking run terminated only on 
two execution traces among the five, whereas, EM-DPOR terminated on all
of them. Except for the execution trace from My Tracks application, EM-DPOR finished state space exploration within
a few seconds. As can be seen from Table~\ref{tab:results},
DPOR explores a much larger number of interleavings and transitions, often 
\emph{orders of magnitude} larger compared to EM-DPOR.
While this is a small evaluation, it does show that significant reduction
can be achieved for real-world multi-threaded event-driven programs by avoiding
unnecessary reordering of events.

\paragraph{Performance.}
Both the techniques used about the same memory and the maximum peak
memory consumed by EM-DPOR across all traces, as reported by Valgrind, was less than $50$MB.
The experiments were performed on a machine with Intel Core i5 3.2GHz CPU 
with 4GB RAM, and running Ubuntu 12.04 OS.

\section{Related Work}
\label{sec:related}

Exploring all possible interleaving of transitions executed by threads (or 
processes) is one of the causes of state explosion problem faced by state space 
exploration based verification techniques. 
Partial order reductions consisting of techniques like \emph{stubborn sets,
persistent sets and sleep 
sets}~\cite{Valmari:1991:SSR:647736.735461,Godefroid:1996:PMV:547238} alleviate 
this problem by trying to explore only a representative interleaving of each 
Mazurkiewicz trace~\cite{DBLP:conf:ac:Mazurkiewicz86} (an equivalence class on 
thread interleavings).
Traces are partial orders of a dependency 
relation~\cite{Katz:1992:DCI:136373.136380,Godefroid:1996:PMV:547238} over  
transitions which classifies a pair of non-interfering transitions as 
independent. A POR enabled state space explorer only reorders dependent transitions, 
and this has been proved to visit all deadlocks and safety violations present
in the original \emph{non-reduced} space of thread interleavings~\cite{Godefroid:1996:PMV:547238}.
Practically, dependent transitions are identified based on the operations 
performed on communication objects like shared memory, FIFO buffers and so on. 
Dynamic partial order reduction (DPOR)~\cite{Flanagan:2005:DPR:1040305.1040315}
is an algorithm to compute persistent sets by checking for dependences 
during runtime, thus improving the precision of the persistent sets computed 
and resulting in greater reductions in state space explored, while the older 
techniques~\cite{Godefroid:1997:MCP:263699.263717} inspect static program structures.

A few works~\cite{musuvathi2007partial, Coons:2013:BPR:2509136.2509556} in the past 
have combined POR with bounded exploration~\cite{Musuvathi:2007:ICB:1250734.1250785,Emmi:2011:DS:1926385.1926432} of the state space.
Coons et al.~\cite{Coons:2013:BPR:2509136.2509556,coonsPhdThesis} have extended persistent sets to account for various bound functions such 
as context bounding and preemption bounding, and have soundly combined the DPOR algorithm 
with various search bounding techniques. They achieve this by conservatively 
identifying more backtracking points where backtracking choices computed to reorder 
a pair of dependent transitions can be added than the default one computed by DPOR,
so that a partial order between transitions which could be explored within the bound
is not missed. Their algorithm which performs bounded POR dynamically is integrated 
with the \textsc{Chess}~\cite{Musuvathi:2008:FRH:1855741.1855760} model checker.

Recent algorithms guarantee optimality in POR~\cite{Abdulla:2014:ODP:2535838.2535845,DBLP:conf/concur/RodriguezSSK15},
\ie\ they explore exactly one transition sequence per Mazurkiewicz trace~\cite{DBLP:conf:ac:Mazurkiewicz86}.
Whereas prior POR techniques guarantee exploring atleast one member from each equivalence
class of execution traces and provided no such optimality guarantees. 
Abdulla et al.~\cite{Abdulla:2014:ODP:2535838.2535845} have devised an optimal DPOR
technique based on a novel backtracking set
called \emph{source set} and a data structure called \emph{wakeup tree}. 
However, the notion of source sets and the optimal DPOR algorithm assume 
total ordering between transitions executed on the same thread. Hence, 
integrating our new dependence relation with source sets will involve significant 
changes to the definitions and algorithms presented in~\cite{Abdulla:2014:ODP:2535838.2535845}. 
Rodr{\'{\i}}guez et al.~\cite{DBLP:conf/concur/RodriguezSSK15} describe unfolding semantics 
parametrized on the commutativity based classical independence relation~\cite{Godefroid:1996:PMV:547238}, and present 
an unfolding based optimal POR algorithm. The unfolding semantics identifies dependent 
transitions with no ordering relation between them to be in conflict.
Their POR algorithm backtracks and explores a new transition sequence $w$ from a state $s$
only if every prior transition explored from $s$ is in conflict with some transition
in $w$. This is problematic in our setting where posts are considered independent and hence
trivially non-conflicting, causing unfolding based POR to miss reordering \texttt{post}s when required.
Establishing optimality in our setting is an interesting but non-trivial future direction.

Huang~\cite{Huang:2015:SMC:2737924.2737975} has developed a state space reduction technique for multi-threaded programs
based on a notion called maximal causality~\cite{Serbanuta2013,Huang:2014:MSP:2594291.2594315}, where an explored 
thread interleaving
is guaranteed to have an operation that reads a value different from all the prior interleavings.
Whereas Mazurkiewicz trace based conventional POR techniques explore different
thread interleaving so as to explore different partial order of dependent
transitions without any constraints on the values observed. Hence, exploration 
based on maximal causality are capable of reducing the number of equivalence
classes over execution traces even further, compared to Mazurkiewicz trace based equivalence. 
Unlike dynamic POR based techniques which explore the thread interleavings using 
depth-first search of the state space,
this technique identifies the interleavings by starting from a seed
interleaving and generate other interleavings by encoding the interleaving and
the allowed variations as a quantifier-free first-order logic formula. Solving
the constraints of the generated formula using an SMT solver identifies an interleaving
from another equivalence class. 
While the number of explorations by maximal causality based reduction
technique can be much smaller, the constraint solving may be time consuming.
However this technique is shown to be parallelized where multiple
interleavings are explored parallely, and the constraint solving corresponding
to various interleavings can also be carried out parallely.

Sen and Agha~\cite{Sen:2006:AST:2182061.2182094} and Tasharofi 
et al.~\cite{Tasharofi:2012:TND:2366649.2366663} describe 
dynamic POR techniques for distributed programs with actor semantics where 
actors execute concurrently. Actors do not have shared memory 
and communicate only via asynchronous message exchanges. Both the POR techniques
for the actor model explore all possible interleavings of messages sent to the 
same process. 
Sen and Agha~\cite{Sen:2006:AST:2182061.2182094} present a way of combining 
concolic execution~\cite{Sen:2005:CCU:1081706.1081750} with partial order reduction 
in the context of actor based systems, thus being able to reason about various data input as 
well as thread interleavings. The dynamic partial order reduction technique 
outlined in~\cite{Sen:2006:AST:2182061.2182094} is adapted in a tool called Basset~\cite{Lauterburg:2009:FSE:1747491.1747541} 
which is a model checker for actor programs built on top of Java PathFinder~\cite{Visser:2003:MCP:641151.641186}.
Tasharofi et al.~\cite{Tasharofi:2012:TND:2366649.2366663} identify 
the dependence relation defined in the context of actor programs to be transitive, 
which is not the case for dependence relation over transitions of multi-threaded
programs. The authors have adapted the DPOR algorithm~\cite{Flanagan:2005:DPR:1040305.1040315}
given for multi-threaded programs to be 
sensitive to this transitive dependence relation, causing it to explore fewer transitions
than a na\"ive adaptation of DPOR for actor programs.
Reduction techniques and model checking algorithms for MPI programs are 
described in~\cite{Palmer:2007:SDD:1273647.1273657,Vakkalanka:2008:DVM:1427782.1427794}. 
MPI programs too use message-passing constructs like non-blocking send and receive 
to exchange data between processes, and use global synchronization constructs 
like barriers. However, the message processing semantics of actor programs and 
MPI programs are quite different compared to the event handling semantics of 
event-driven programs such as Android applications.

$R^4$~\cite{Jensen:2015:SMC:2814270.2814282} is a stateless model checker 
for event-driven programs such as client-side web applications. $R^4$
adapts persistent sets~\cite{Godefroid:1997:MCP:263699.263717} and the DPOR 
algorithm to the domain of single-threaded event-driven programs where enqueued
events are non-deterministically dequeued in any order and each event handler 
is atomically executed to completion without interference from other handlers.
As described in~\cite{Jensen:2015:SMC:2814270.2814282}, the concurrency
model handled by $R^4$ allows an entire event handler to be considered as a single 
transition. 
In contrast, the focus of our POR technique is on  
multi-threaded programs with event queues, and thus needs to be sensitive to 
interference from multiple threads.
Mirzaei et al.~\cite{jpfsymandroid} and Merwe et al.~\cite{vanderMerwe:2012:VAA:2382756.2382797} 
model Android libraries and extend Java PathFinder~\cite{Visser:2003:MCP:641151.641186} to model check 
Android applications. However, these works do not model various concurrency aspects of Android
present in real-world applications.
AsyncDroid~\cite{DBLP:conf/cav/OzkanET15} is a systematic concurrency testing tool for Android applications
which explores various thread schedules for a given sequence of UI events.

While most of the state space reduction techniques in the literature assume the 
target programs to be run under a sequentially consistent (SC) memory model, recently, 
many efficient stateless model checking techniques have been developed for weaker memory
models as well~\cite{Abdulla:2015:SMC:2945565.2945622,Zhang:2015:DPO:2737924.2737956,
Demsky:2015:SSS:2814270.2814297,DBLP:conf/cav/AbdullaAJL16,Huang:2016:MCR:2983990.2984025}.
The challenges faced when developing efficient exploration techniques for event-driven 
programs are orthogonal to those faced when handling different memory models.

\section{Conclusions and Future Work}
\label{sec:conclusions}

The event-driven multi-threaded style of programming concurrent applications
is becoming increasingly popular. We considered the problem of 
POR-based efficient stateless model checking for this concurrency model. The key insight of
our work is that more reduction is achievable by treating 
operations that post events to the same thread as independent and only
reordering them if necessary. 

Towards this, we presented new formulations
of dependence-covering sequences and sets such that exploring only
dependence-covering sets suffices to provide certain formal guarantees.
We also presented EM-DPOR ---a dynamic algorithm to perform POR by computing
dependence-covering sets for event-driven
multi-threaded programs. 
Our experiments provide empirical evidence that EM-DPOR explores orders of magnitude fewer transitions compared
to DPOR for event-driven multi-threaded programs.

In future, we plan to develop further optimizations and a practical tool
to model check these programs.
Also, we aim to achieve better reductions by defining a notion of sleep sets suitable for this
concurrency model and combining it with dependence-covering sets. 
Another non-trivial but interesting problem would be to establish
optimality in our event-driven setting on the similar lines as~\cite{Abdulla:2014:ODP:2535838.2535845,
DBLP:conf/concur/RodriguezSSK15}. A few other directions are to 
extend~\cite{Huang:2015:SMC:2737924.2737975} to develop maximal causality based state space exploration technique
for event-driven programs, and to explore bounded POR for event-driven programs.

\newpage
\bibliographystyle{abbrvnat}
\bibliography{paper}

\appendix
\section{Properties of Dependence-covering Sets}
\label{app:dc-properties}

In this section, we prove that a selective state space exploration using
the dependence-covering sets (see Definition~\ref{def:dc-set}) is sufficient to detect all 
deadlock cycles in $\mathcal{S}_G$ (see Section~\ref{sec:transition-system-por}), 
and if there is a state 
$s$ in $\mathcal{S}_G$ where a local assertion $\alpha$ fails then some state 
$s'$ where $\alpha$ fails is reached in the reduced state space as well.
We then give a theorem relating dependence-covering sets and persistent sets~\cite{Godefroid:1996:PMV:547238}.

\subsection{Deadlock Cycles}

In the following discussion,
let $w$ be a transition sequence from a state $s$ in $\mathcal{S}_G$ to reach
a deadlock cycle $\langle DC,\rho \rangle$. 
Let $u$ be a dependence-covering sequence (see Definition~\ref{def:dep-covering}) of $w$ starting from $s$. 
Further, $R_w$ and $R_u$ be the sets of transitions executed in $w$
and $u$ respectively, and $s_n$ and $s_m'$
be the last states reached by $w$ and $u$ respectively. 

\begin{lemma}
\label{lmm:blocking}
\normalfont
Let $s_i$ be a state reached by
a prefix of $w$ where a transition $b \in DC$ is blocked and
not enabled later in $w$. 
Then, there exists a prefix of $u$ which reaches a state $s_j'$
where $b$ is blocked and not enabled later in $u$.
\end{lemma}
\begin{proof}
Let $R_b \subseteq R_w$ denote the transitions
which have a direct dependence with $b$ or a dependence with some other transition
which directly or transitively has a dependence with $b$. Clearly, all the transitions
in $R_b$ are executed prior to $b$ since the transition $b$ is blocked by a prefix
of $w$ and never enabled as per the premise of the lemma.
By the definition of dependence-covering sequence (see Definition~\ref{def:dep-covering}), $R_w \subseteq R_u$
and hence, $R_b \subseteq R_u$. Further, the relative ordering of dependent transitions
in $R_b$ is maintained in $u$. Let $s_j'$ be the state reached after executing
all transitions in 
$R_b$ in $u$. 
Since $b \in nextTrans(s_n)$, there can be a transition $r_k'$ in $u$ such that 
$r_k' \not\in R_w$ and $r_k'$ is dependent with $b$, in particular, $r_k'$ enables $b$. 
Since $u$ is a dependence-covering sequence of $w$, $r_k'$ exists only if $w$ can be extended
so that $r_k'$ executes before $b$. By the definition of deadlock cycle, this is
not possible. Hence, the transition $b$ will be blocked at $s_j'$ and
there is no transition in $u$ which enables $b$ after $s_j'$.
\end{proof}

\begin{lemma}
\label{lmm:deadlock}
\normalfont
The pair $\langle DC,\rho \rangle$ is a deadlock cycle at the state $s_m'$ reached by $u$.
\end{lemma}
\begin{proof}
By Lemma~\ref{lmm:blocking}, for any $b \in DC$,
there exists a state $s_j'$ reachable from $s$
by some prefix of $u$ such that $b$ is blocked at $s_j'$ and not enabled
later in $u$. Thus, all the transitions in $DC$ are blocked at $s_m'$.

Let $DC$ contain $k$ transitions and $b = \rho(a)$ for some $a \in [1,k]$.
Let $t$ be the thread blocked on the transition $b' = \rho(a+1)$ at $s_n$ where
$k+1$ is taken to be $1$. In $w$, let $r_i$ be the transition of $t$ that
blocks $b$ after which it is never enabled. 
Clearly, $r_i \in R_b$ for the set $R_b$ defined in the proof of 
Lemma~\ref{lmm:blocking}. Since the state $s_j'$ is reached in $u$ once
all the transitions in $R_b$ are executed and in the same relative order
between themselves, 
$r_i$ blocks $b$ before or at $s_j'$ in $u$. 
Since $b$ remains blocked from $s_j'$ onwards (Lemma~\ref{lmm:blocking}), 
there is no other transition
in $u$ that can enable $b$.
By Lemma~\ref{lmm:blocking}, the thread $t$ itself 
subsequently blocks on $b' \in DC$ in $u$.
Thus, $\langle DC,\rho \rangle$ is also a deadlock cycle at $s_m'$.
\end{proof}

\begin{customthm}{\ref{theorem:deadlock-reachability}.1}
\label{th:preserve-deadlock}
\normalfont
[\textbf{Part of Theorem~\ref{theorem:deadlock-reachability}}]
Let $\mathcal{S}_R$ be a 
dependence-covering state space of a program $A$
with a finite and acyclic state space $\mathcal{S}_G$. 
Then, all deadlock cycles in $\mathcal{S}_G$ are 
reachable in $\mathcal{S}_R$.
\end{customthm}
\begin{proof}
Let $\langle DC,\rho \rangle$ be a deadlock cycle at 
a state $d$ in $\mathcal{S}_G$, reachable from $s_{init}$.
Let $s$ be a state which is common to 
both $\mathcal{S}_G$ and $\mathcal{S}_R$ such that there exists a transition
sequence $w$ from $s$ to $d$ in $\mathcal{S}_G$. In the least, the initial 
state $s_{init}$ is such a state.

Let $L$ be a dependence-covering set at $s$. By definition~(see 
Definition~\ref{def:dc-set}), there exists a transition sequence
$u$ from $s$, starting with a transition $r \in L$ such that $u$ is a 
dependence-covering sequence of $w$. By Lemma~\ref{lmm:deadlock}, $u$ eventually
reaches the deadlock cycle $\langle DC,\rho \rangle$. Let $s' = r(s)$. Since
$r \in L$, $s'$ is in $\mathcal{S}_R$. If $u = r.u'$ then $u'$ is a transition
sequence from $s'$ in $\mathcal{S}_G$ to a state with deadlock cycle 
$\langle DC,\rho \rangle$. There exists a dependence-covering sequence for $u'$
from $s'$ in $\mathcal{S}_R$. With a similar argument, there exists a 
successor state $s''$ of $s'$ in 
$\mathcal{S}_R$ from which the same deadlock cycle can be reached and so on.
Since the state space is finite and acyclic, eventually a state $d'$ is reached
in $\mathcal{S}_R$ where $\langle DC,\rho \rangle$ is a deadlock cycle.
\end{proof}

A dependence-covering state space only preserves all the deadlock cycles and not 
deadlock states present in $\mathcal{S}_G$. Suppose $w$ is a transition sequence
in $\mathcal{S}_G$ reaching a deadlock state $d$. 
Let $u$ be a dependence-covering sequence of $w$.
Since $u$ may contain some transitions not in $w$ (Definition~\ref{def:dep-covering}) and those may
modify some shared objects, $u$ may reach another state $d'$ with the same
deadlock cycle as in $d$. But $d$ and $d'$ may not be the same.
Note that exploration of the dependence-covering state space does
detect the set of transitions involved in \emph{each} deadlock in $\mathcal{S}_G$.

\subsection{Assertion Violations}

\begin{customthm}{\ref{theorem:deadlock-reachability}.2}
\label{th:preserve-assertion}
\normalfont
[\textbf{Part of Theorem~\ref{theorem:deadlock-reachability}}]
Let $\mathcal{S}_R$ be a
dependence-covering state space of an event-driven multi-threaded program $A$
with a finite and acyclic state space $\mathcal{S}_G$. 
If there exists a state $v$ in $\mathcal{S}_G$
which violates an assertion $\alpha$ defined over local variables then
there exists a state $v'$ in $\mathcal{S}_R$ which violates $\alpha$.
\end{customthm}
\begin{proof}
The state $v$ is reachable from the initial state $s_{init}$ in $\mathcal{S}_G$.
Let $s$ be a state which is common to both $\mathcal{S}_G$ and 
$\mathcal{S}_R$ such that there exists a transition sequence $w$ from $s$
to $v$ in $\mathcal{S}_G$. In the least, the initial state $s_{init}$ is such a state. 

Let $L$ be a dependence-covering set at $s$. By the definition of dependence-covering
set (see Definition~\ref{def:dc-set}), there exists a transition sequence
$u$ from $s$, starting with a transition $r \in L$ such that $u$ is a 
dependence-covering sequence of $w$. Let $R_w$ and $R_u$ be the sets of transitions
executed in $w$ and $u$ respectively. Let $R_{\alpha} \subseteq R_w$ denote the
set of transitions which have a direct dependence with $\alpha$ or a dependence with
some other transition which directly or transitively has a dependence with $\alpha$.
By definition (see Definition~\ref{def:dep-covering}), $R_w \subseteq R_u$ and hence, $R_{\alpha} \subseteq R_u$. Further,
the relative ordering of dependent transitions in $R_\alpha$ is maintained in $u$.
Let $v'$ be the state reached after executing all transitions in 
$R_\alpha$ in $u$.
Since $\alpha$
is an assertion on local variables, no new transition $r_k' \in R_u$ \ie\ 
$r_k' \in R_u \setminus R_w$ can have a dependence with $\alpha$.
Thus state $v'$ violates the assertion $\alpha$.

Let $s' = r(s)$. Since $r \in L$, $s'$ is in $\mathcal{S}_R$. If $u = r.u'$
then $u'$ is a transition sequence from $s'$ in $\mathcal{S}_G$ to a state
which violates $\alpha$. With a similar argument, there
exists a successor state $s''$ of $s'$ in $\mathcal{S}_R$ from which a state
which violates $\alpha$ is reachable and so on.
Since the state space is finite and acyclic, eventually a state is
reached in $\mathcal{S}_R$ which violates $\alpha$.
\end{proof}

\subsection{Relation between Persistent Sets and Dependence-covering Sets}
\label{sec:ps-as-dcs}
\begin{theorem}
\normalfont
 If $P$ is a persistent set in a state $s \in \mathcal{S}_G$ according to 
 the standard dependence relation which considers \texttt{post}s to the same
 event queue to be dependent,
 then $P$ is a dependence-covering set in $s$ according to the dependence relation
 of Definition~\ref{def:dependence-relation}.
\end{theorem}
\begin{proof}
 Let  $w: r_1.r_2 \ldots r_n$ 
 be any transition sequence in $\mathcal{S}_G$ from a state $s$.
As $w$ is a dependence-covering sequence of itself, if $r_1 \in P$ then
 $P$ is also a dependence-covering set in $s$. 
 
 If $r_1 \not\in P$ then by Lemma~6.8 in~\cite{Godefroid:1996:PMV:547238} we can infer that 
 either (a)~there exists a sequence $w' \in [w]$ (where $[w]$ is the Mazurkiewicz
 trace of $w$) such that the first transition in $w'$, say $w_1'$, is in the 
 persistent set $P$, or (b)~all the transitions in $P$ are independent with all the transitions in $w$.
 We prove the lemma for the two cases (a) and (b) identified.

 \vspace{2mm}\noindent
 \textbf{Case~(a)}: We show that $w'$ is a dependence-covering sequence of $w$ 
 even according to dependence relation of Definition~\ref{def:dependence-relation}.
 Since $w' \in [w]$, the relative ordering of each pair of dependent transitions 
 in $w'$ is the same as that in $w$. 
 The only difference between the dependence relation of Definition~\ref{def:dependence-relation}
 and the standard dependence relation resulting in Mazurkiewicz traces is that, 
 Definition~\ref{def:dependence-relation} considers \texttt{post}s to be independent
 and does not totally order transitions executed by different event handlers on the
 same thread. However, if interfering (non-\texttt{post})
 transitions are executed on two different threads,
 then both these dependence relations identity such pairs to be dependent.
 Since \texttt{post}s are considered dependent as per the dependence relation 
 resulting in $[w]$, the relative ordering of all \texttt{post}s in $w'$ posting 
 to the same event queue is consistent with that in $w$. As a result, the relative
 ordering of operations across event handlers executed on the same thread is the same
 in both $w'$ and $w$. Thus, the relative orderings of all dependent transitions
 in $w$ are preserved in $w'$ even according to Definition~\ref{def:dependence-relation}.
 Additionally, $R_w = R_{w'}$ because of the property of Mazurkiewicz trace. 
 Thus, $w'$ is a dependence-covering sequence of $w$ such that $w_1' \in P$ (assumption of this case).
 Therefore, $P$ is a dependence-covering set in $s$ as per Definition~\ref{def:dc-set}.

\vspace{2mm}\noindent
 \textbf{Case~(b)}: Consider a state $s' = r(s)$ such that $r \in P$. As per the 
 assumptions of this case, $r$ is independent with all the transitions in $w$ as per the
 standard dependence relation which considers \texttt{post}s to the same event queue dependent.
 Then, sequence $w$ is enabled at $s'$ making $r.w$ a valid sequence in $\mathcal{S}_G$. 
 If any transition $r_i$ in $w$ has a \texttt{post} operation, then $r$ cannot be 
 a transition posting to the same event queue as $r_i$. Otherwise, $r_i$ would be 
 dependent with $r$, contradicting the assumption of this case. Also, if $r$ is executed
 on a thread $t$ then, no transition in $w$ is executed on thread $t$, because 
 $next(s,t)$ is unique. A pair of transitions from different threads considered
 independent by the standard dependence relation, are considered independent even
 by Definition~\ref{def:dependence-relation} (see condition~2 of the definition).
 Hence, $r.w$ is a dependence-covering sequence of $w$ in conjunction with
 dependence relation of Definition~\ref{def:dependence-relation}. Thus, $P$ is a 
 dependence-covering set in $s$.
\end{proof}

\section{Correctness of EM-DPOR}
\label{app:emdpor-proofs}

This section presents a sketch to prove the correctness of the algorithm 
\emdpor\ to dynamically compute dependence-covering sets (see Definition~\deflazypersistentset), 
presented in Section~\ref{sec:algo}. 
Algorithm \texttt{Explore} (Algorithm~\algoexplore) performs a
depth first traversal of the state space. We want to prove that whenever \texttt{Explore}
backtracks from a state $s$ to a prior state in the search stack, it must have 
explored a dependence-covering sequence (see Definition~\defdepcovering) for any 
sequence $w$ in $\mathcal{S}_G$ from state $s$. 
We equivalently prove that EM-DPOR explores a dependence-covering set at each visited
state $s$. Theorem~\ref{thm:emdpor-correctness} given towards the end of this section  
formally states this property.

We organize this section as follows. Section~\ref{sec:proof-strategy} gives the 
proof strategy for the Theorem~\ref{thm:emdpor-correctness}. 
Section~\ref{sec:lemma-em-dpor} provides a complete proof or a proof sketch for 
the lemmas related to the cases introduced in the
proof strategy, and Section~\ref{sec:em-dpor-proof} presents the main proof.
The variables and notation introduced in Section~\ref{sec:proof-strategy} will be used
in the rest of this section.

Even though in Section~\ref{sec:emdpor-algo-defs} we had defined helper 
functions such as $last()$, $pre()$ and a few others over a transition sequence
starting from the initial state $s_{init}$, we may abuse the notation to use these
functions over transition sequences starting from an intermediate state in the 
state space $\mathcal{S}_G$ as well. 
In the rest of the section, \emph{happens-before} relation ($\to_w$) used in the context of
a transition sequence $w$ in $\mathcal{S}_G$ which is not assumed to be explored by EM-DPOR
is defined as follows: 
\begin{definition}
\label{def:happens-before-no-emdpor}
\normalfont
The \textbf{\emph{happens-before relation}} $\to_{w}$ for a transition
sequence $w : r_1 . r_2 \ldots r_n$ in $\mathcal{S}_G$ is the 
smallest relation on $dom(w)$ such that the following conditions hold:
\begin{enumerate}
\item If $i < j$ and $r_i$ is dependent with $r_j$ then $i \to_{w} j$.
\item If $r_i$ and $r_j$ are transitions posting events $e$ and $e'$
respectively to the same thread, such that 
$i \to_w j$ and the handler of $e$ has finished and that of
$e'$ has started in $w$, then $getEnd(w,e) \to_w getBegin(w,e')$.
\item $\to_{w}$ is transitively closed.
\end{enumerate}
\end{definition}

While the above happens-before relation is similar to that defined in Definition~\ref{def:happens-before},
it does not reason about reordered \texttt{post}s in $w$. This is because if we do not assume $w$ to be explored
by EM-DPOR then the notion of reordered \texttt{post}s is irrelevant in the context of $w$. 
Hence, the above happens-before relation can be derived for a transition sequence
in $\mathcal{S}_G$ without any prior information and purely with the help of dependence relation on $\mathcal{S}_G$.
However, if $w$ is assumed to be explored by EM-DPOR then happens-before relation ($\to_w$) referred in 
its context is the one defined in Definition~\ref{def:happens-before}.

\paragraph{Transition in a sequence.}
Given a transition sequence $w$ and another transition sequence $z$ we say that
a transition $c \in R_z \cup nextTrans(last(z))$ to be a transition executed in 
$w$, \ie\ $c \in R_w$, if at a state $s'$ reached by a prefix $w'$ of $w$ such 
that $c \in nextTrans(s')$ and a state  $s''$ reached by a prefix $z'$ of $z$ such 
that $c \in nextTrans(s'')$ we have $index(w',b) \to_{w'} \taskof{(c)}$ iff $index(z',b) \to_{z'} \taskof{(c)}$
for any transition $b$. 

This intuitively means that the transition prior to $c$ in $\taskof(c)$ has an 
identical incoming direct and transitive dependence edges in the dependence graph
of $w$ as well as $z$, due to which transition $c$ in $w$ is guaranteed to be discovered 
(may or may not be enabled) in the transition sequence $z$ too.

\subsection{Proof Strategy and Notation}
\label{sec:proof-strategy}
\subsubsection{Inductive Reasoning}
\label{sec:inductive-reason}
\emdpor\ consists of four algorithms --- \texttt{Explore} (Algorithm~\algoexplore),
\texttt{FindTarget} (Algorithm~\algofindtarget), \texttt{ReschedulePending} (Algorithm~\algoreschedulepending)
and \texttt{BacktrackEager} (Algorithm~\algobacktrackeager).
The proof is by induction on the order in which states visited by \texttt{Explore} (Algorithm~\algoexplore) are backtracked. 
This is similar to the inductive strategy used to prove Theorem~1 in~\cite{Flanagan:2005:DPR:1040305.1040315}, 
which states that the DPOR algorithm computes persistent sets at each explored state. 
However, we cannot directly borrow the structure of DPOR's proof, as we additionally need to
consider the effect of event-driven semantics and reason about the recursive nature of 
\texttt{FindTarget} (Algorithm~\algofindtarget).

Let $S$ be a sequence explored by Algorithm \texttt{Explore} of EM-DPOR, starting from
an initial state $s_{init} \in \mathcal{S}_G$. Let $s = last(S)$, and $L = \{next(s,t) \mid t \in backtrack(s)\}$
where $backtrack(s)$ is the backtracking set computed by EM-DPOR before backtracking to a state prior to $s$ in the 
search stack $S$.
Assume $\mathcal{S}_R \subseteq \mathcal{S}_G$ to be the state space explored by EM-DPOR starting from state $s_{init}$.
State $s$ is in $\mathcal{S}_R$ as sequence $S$ is explored by Algorithm~\algoexplore.

\begin{mdframed}
 {\textbf{Claim C1.}} The EM-DPOR algorithm explores a dependence-covering 
 sequence for every transition sequence $w$ in $\mathcal{S}_G$ from a state $s$ reached
 by \texttt{Explore($S$)}.
\end{mdframed}
\paragraph{Induction hypothesis H1.} For every transition sequence from a 
state reached on each recursive call \texttt{Explore($S.r$)}, for all $r \in L$, the algorithm 
explores a corresponding dependence-covering sequence.

The base case of the induction based proof of Claim~C1 which captures the essence of
Theorem~\ref{thm:emdpor-correctness}, will be proved in 
Section~\ref{sec:em-dpor-proof}. The proof strategy for the induction step is presented below.
 
\paragraph{Induction step.} 
We prove that for any sequence $w: s \xrightarrow{r_1} s_1 \xrightarrow{r_2} 
s_2 \ldots \xrightarrow{r_n} s_n$ in $\mathcal{S}_G$, 
\texttt{Explore($S$)} explores a dependence-covering sequence of $w$ from 
state $s$. Here, $s' \xrightarrow{r'} s''$ means $s'' = r'(s'')$.
 
If $r_1 \in L$ then, the algorithm explores a dependence-covering sequence $u$
of $r_2 \ldots r_n$ from state $s_1$ by induction hypothesis~H1, making $r_1.u$ a 
dependence-covering sequence of $w$ from state $s$. Assume $r_1 \not\in L$ henceforth.
We also assume that $w$ has no dependence-covering sequence starting with any
transition in $L$ from state $s$. 

We prove the inductive case by doing an exhaustive case analysis of the contents of
set $L$. Set $L$ satisfies the properties presented in one of the following five cases.
\begin{enumerate}
 \item[A.] $\exists p \in L$ such that $p$ is a non-\post\ transition and $p$
 is independent with all the transitions in $w$.
 \item[B.] $L$ contains a non-empty subset of non-\post\ transitions such that all 
 the non-\post\ transitions in $L$ are dependent with some transition in $w$, and
 no transition in $L$ is in $w$. 
 \item[C.]  $L$ contains a non-empty subset of non-\post\ transitions such that all 
 the non-\post\ transitions in $L$ are dependent with some transition in $w$, and
 the first transition in $w$ from $L$ is a non-\post\ transition.
 \item[D.] $L$ contains only \post\ transitions and no transition in $L$ is in $w$.
 \item[E.] The first transition in $w$ from $L$ is a \post\ transition. In this case
 if $L$ contains non-\post\ transitions we assume all of them to be dependent with some transition
 in $w$. Note that the presence of non-\post\ transitions in $L$ does not affect the proof
 in this case. 
\end{enumerate}

Section~\ref{sec:lemma-em-dpor} presents lemmas reasoning the induction step
for each of the five cases above. 
Lemma for case~A is proved by deriving contradiction to our assumption on non-existence
of a dependence-covering sequence of $w$ starting with any transition in $L$ from $s$.
Lemmas for cases~B, C, D and E are proved by deriving contradictions to the assumptions
made on the contents of $L$, when we assume non-existence
of a dependence-covering sequence of $w$ starting with any transition in $L$ from $s$.
This in turn proves the existence of a dependence-covering sequence of $w$ in $\mathcal{S}_R$
from state $s$.


\subsubsection{Common Construction for Cases~B, C, D and E}
\label{sec:common-construction}
As shown in Figure~\ref{fig:common-construct} we construct a transition sequence $z: s \xrightarrow{r_1'} s_1' \xrightarrow{r_2'} 
s_2' \ldots \xrightarrow{r_m'} s_m'$ in $\mathcal{S}_G$, such that (a)~$r_1' \in L$ 
and (b)~$\exists r \in nextTrans(s_m')$ where $r$ is a transition in $w$, say $r_\varpi = r$ 
for $1 \leq \varpi \leq n$ in $w$, and $r$ is dependent with a transition $r_l'$ in $z$ such that
$r_l'$ is the nearest may be co-enabled or may be reordered transition that does 
not happen before $r$.
Additionally, $r_l'$ may or may not be executed in $w$. If $r_l'$ is executed in $w$
then $index(w,r) < index(w,r_l')$. We use $z$ which is not a dependence-covering
sequence of $z$ in our proof arguments, provided $z$ is valid in $\mathcal{S}_G$. 
We reason about the validity of $z$ in each of cases~B, C, D and E separately.
In cases~D and E we generate a set of relevant \emph{non} dependence-covering
transition sequences of $w$ with the help of $z$, all of which will be used by 
the proofs related to cases~D and E.

Figure~\ref{fig:common-construct} pictorially depicts some of the 
key states, transitions, sequences and function calls required
when reasoning about cases B, C, D and E. Any other properties
of $r_l'$ specific to the case B, C, D or E considered, will be presented in 
Section~\ref{sec:lemma-em-dpor}. 
Let $Z = z.z'.r$ where $z'$ is the shortest sequence in $\mathcal{S}_G$ which enables $r$.
If there exists no such $z'$ then $Z = z$. Note that $z' = \epsilon$ if $thread(r) \in enabled(s_m')$.
Let $v$ be a suffix of $Z$ from state $s_1'$ \ie\ $v = r_2'.r_3'\ldots r_m'.z'.r$
if $Z = z.z'.r$ or $v = r_2'.r_3'\ldots r_m'$ if $Z = z$.
Since $s \in \mathcal{S}_R$ and $r_1' \in L$, state $s_1' = r_1'(s)$ is in $\mathcal{S}_R$.
Then by induction hypothesis~H1, EM-DPOR explores a dependence-covering sequence 
$u$ of $v$ from $s_1'$. 
Since $r_1' \in L$ algorithm explores $r_1'.u$. Clearly, $r_1'.u$ is a dependence-covering 
sequence of $Z = r_1'.v$ from state $s$. 

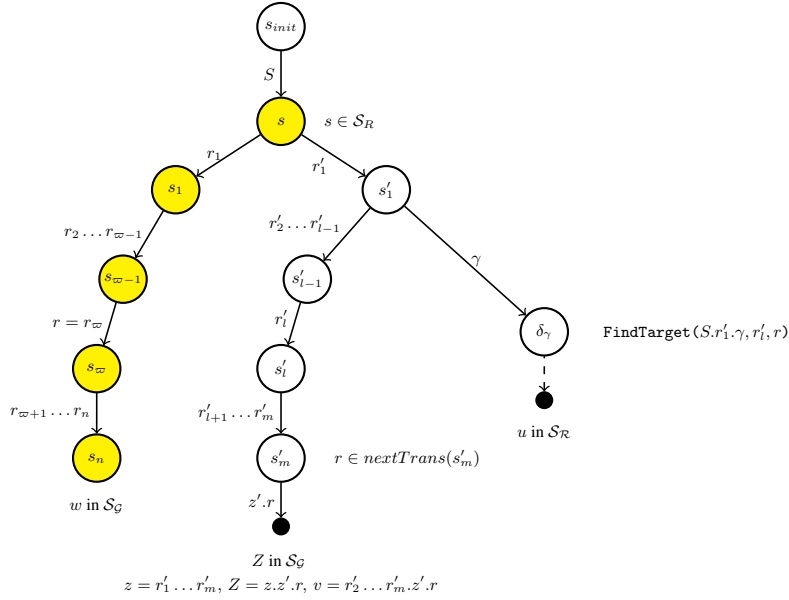
\begin{figure*}[t]
\centering
\begin{tikzpicture}[auto,node distance=16 mm,semithick, scale=.7, transform
shape]
\tikzstyle{stateNode} = [circle,draw=black,thick,inner sep=0pt,minimum 
size=9mm]
\tikzstyle{dporstateNode} = [circle,fill=yellow,draw=black,thick,inner sep=0pt,minimum 
size=9mm]
\tikzstyle{blankNode} = [thick,inner sep=0pt,minimum size=6mm]
\tikzstyle{filledNode} = [circle,fill=black,draw=black,thick,inner sep=0pt,minimum 
size=3mm]

\node[stateNode] (sinit)  {$s_{init}$};
\node[dporstateNode] (s) [below of=sinit,yshift=-.2cm] 
{$s$} edge [<-] node [left] {$S$} (sinit);
\node[blankNode] (b1) [right of=s, xshift=-.3cm] {$s \in \mathcal{S}_R$};

\node[dporstateNode] (s1) [below of=s,yshift=.3cm,xshift=-2cm] 
{$s_1$} edge [<-] node [left] {$r_1$} (s);
\node[dporstateNode] (s2) [below of=s1,yshift=-.1cm, xshift=-1cm] 
{$s_{\varpi-1}$} edge [<-] node [left] {$r_2 \ldots r_{\varpi -1}$} (s1);
\node[dporstateNode] (s3) [below of=s2,yshift=-.1cm, xshift=-.5cm] 
{$s_\varpi$} edge [<-] node [left] {$r = r_\varpi$} (s2);
\node[dporstateNode] (sn) [below of=s3,yshift=-.1cm] 
{$s_n$} edge [<-] node [left] {$r_{\varpi +1}\ldots r_n$} (s3);
\node[blankNode] (b2) [below of=sn, yshift=.7cm] {$w$ in $\mathcal{S_G}$};

\node[stateNode] (s1p) [below of=s,yshift=.3cm,xshift=2cm] 
{$s_1'$} edge [<-] node [left,yshift=-.2cm] {$r_1'$} (s);
\node[stateNode] (s2p) [below of=s1p,yshift=-.1cm, xshift=-1.5cm] 
{$s_{l-1}'$} edge [<-] node [left,yshift=.2cm] {$r_2'\ldots r_{l-1}'$} (s1p);
\node[stateNode] (slp) [below of=s2p,yshift=-.1cm, xshift=-.5cm] 
{$s_l'$} edge [<-] node [left, yshift=.1cm] {$r_l'$} (s2p);
\node[stateNode] (smp) [below of=slp,yshift=-.1cm] 
{$s_m'$} edge [<-] node [left] {$r_{l+1}'\ldots r_m'$} (slp);
\node[filledNode] (sZ) [below of=smp, yshift=.3cm] {}
edge [<-] node [left] {$z'.r$} (smp);
\node[blankNode] (b6) [right of=smp, xshift=.8cm] {$r \in nextTrans(s_m')$};
\node[blankNode] (b3) [below of=sZ, yshift=.7cm] 
{\begin{tabular}{c}
$Z$ in $\mathcal{S_G}$ \\[1mm]
$z = r_1'\ldots r_m', \, Z = z.z'.r, \, v = r_2'\ldots r_m'.z'.r$ \\
\end{tabular}
};
\node[stateNode] (sgamma) [below of=s1p,yshift=-1.1cm,xshift=3cm] 
{$\delta_\gamma$} edge [<-] node [right] {$\gamma$} (s1p);
\node[blankNode] (b4) [right of=sgamma, xshift=1.3cm] {\texttt{FindTarget($S.r_1'.\gamma, r_l', r$)}};
\node[filledNode] (su) [below of=sgamma,yshift=.3cm] {} 
edge [<-, dashed]  (sgamma);
\node[blankNode] (b5) [below of=su, yshift=1cm] {$u$ in $\mathcal{S_R}$};

\end{tikzpicture}
\caption{Illustration of key components in the proof strategy of cases B, C, D and E.
Transition sequences $w$ and $Z$ start at state $s$, and sequences $v$ and $u$ start at state $s_1'$.
A solid circle with no annotation denotes a state in $\mathcal{S}_G$. States coloured
yellow correspond to sequence $w$.}
\label{fig:common-construct}
\end{figure*}

Let $\delta_\gamma \in \mathcal{S}_R$ be the state reached by $S.r_1'.\gamma$ where $\gamma$ is a prefix of 
$u$ such that $r_l'$ is a transition in $r_1'.\gamma$, and 
transition $r$ dependent with $r_l'$ is in $nextTrans(\delta_\gamma)$. Due to the characteristics
of $r_l'$ and $r$ described in constraint~(b) given earlier on sequence $z$, 
and $r_1'.u$ being a dependence-covering sequence of $Z$ whose prefix is $z$,
\texttt{Explore($S.r_1'.\gamma$)} invokes \texttt{FindTarget($S.r_1'.\gamma,r_l',r$)} 
(line~\linefindtarget\ in Algorithm~\algoexplore). With this being a common
scenario for cases B, C, D and E, we present specific arguments for each of the cases
in their respective lemmas in Section~\ref{sec:lemma-em-dpor}, and derive 
contradictions to the assumptions made on the contents of L. 

\paragraph{Notation.} 
Given transition sequences $w_1$ and $w_2$, let $w_1 \setminus w_2$ denote transitions which are in 
sequence $w_1$ but not in sequence $w_2$. 
For a set of tasks $tks$, $threadSet(tks) = \{t \mid (t,e) \in  tks\}$, \ie\ $threadSet$
gives a set of threads corresponding to a set of tasks.
Whenever we need to reason about multiple instances of 
variables like $candidates$ and $pending$ from Algorithm~\algoexplore, \algofindtarget,
\algoreschedulepending\ or \algobacktrackeager\ in our proofs, we use numerical
subscripts to distinguish one instance from the other (\eg\ $candidates_1$ is different from $candidates_2$ and so on).
We do not add any subscripts for variable instances corresponding to the first \texttt{FindTarget}
call (\texttt{FindTarget($S.r_1'.\gamma,r_l',r$)}) from \texttt{Explore($S.r_1'.\gamma$)}.

\subsection{Supporting Lemmas}
\label{sec:lemma-em-dpor}
We use the induction hypothesis H1 and prove induction step separately for each of the 
cases A -- E introduced in section~\ref{sec:inductive-reason}.

\subsubsection{Case A}
\begin{lemma}
\label{lmm:casea}
\upshape
 EM-DPOR explores a dependence-covering sequence of $w$ from state $s$ when set $L$
 satisfies case~A.
\end{lemma}
\begin{proof}
Case A states that, $\exists p \in L$ such that $p$ is a non-\post\ transition and $p$
is independent with all the transitions in $w$. 
Then, no transition in $w$ is executed on the same thread as $p$. This is because,
$p = next(q,thread(p))$ for any state $q$ visited by a prefix of $w$. Since the
next transition of a thread at any state is unique, no transition in $w$ is executed
on $thread(p)$. 
Then, by the second condition of the dependence relation 
(Definition~\defdependencerelation), $p$ commutes with all the transitions in
$w$ and sequence $w$ is enabled at state $s' = p(s)$. Since $s \in \mathcal{S}_R$ and 
$p \in L$, $s' \in \mathcal{S}_R$. Then by induction hypothesis H1, EM-DPOR
explores a dependence-covering sequence $u$ of $w$ from $s'$. Therefore, $p.u$
is a dependence-covering sequence of $w$ at state $s$.
\end{proof}

\subsubsection{Case B}
Case~B states that the backtracking set $L$ in state $s$ contains a 
non-empty subset of non-\post\ transitions such that all 
the non-\post\ transitions in $L$ are dependent with some transition in $w$, and
no transition in $L$ is in $w$. 
With the help of the transition sequence $z$ described in Section~\ref{sec:common-construction}
we prove that \emdpor\ identifies a transition in $w$ to be reordered with a 
non-\post{} transition in $L$, due to which a transition in $w$ gets 
added to the set $L$. This establishes contradiction to the property of set $L$
which in turn proves that our primary assumption of absence of a dependence-covering
sequence of $w$ starting from a transition in set $L$, does not hold.

To suit the case under consideration, we refine the construction of sequence $z$ 
as follows.

\begin{construction}
\label{construct:caseb}
\upshape
Let $z: s \xrightarrow{r_1'} s_1' \xrightarrow{r_2'} 
s_2' \ldots \xrightarrow{r_m'} s_m'$ in $\mathcal{S}_G$ be a sequence satisfying
the following constraints:
\begin{enumerate}[label=M\arabic*.]
 \item $r_1' \in L$ is a non-\post\ transition.
 \item For all $r_i'$ in $z$, $i \neq 1$, $r_i'$ is a transition in $w$ and 
 $r_i' = r_{i-1}$. Recall that $w = r_1.r_2 \ldots r_n$.
 \item $\exists r \in nextTrans(s_m')$ such that $r = r_\varpi$ is in $w$, and
 $r$ is the first transition in $w$ to be dependent with $r_1'$.
\end{enumerate}
\end{construction}

\noindent Following properties can be inferred for a sequence $z$ adhering to Construction~\ref{construct:caseb}.
\begin{enumerate}[label=P\arabic*.]
 \item All the transitions in $z$ except $r_1'$ are in $w$.
 \item Transitions $r_1'$ and $r$ may be co-enabled \ie\ $thread(r_1') \neq thread(r)$. 
 This is because, $r_1' = next(q,thread(r_1'))$
 at any state $q$ visited by a prefix of sequence $w$, including the state $s_{\varpi-1}$ where $r = r_\varpi$
 is executed. Since both $r_1'$ and $r$ are in $nextTrans(s_{\varpi-1})$, $thread(r_1') \neq thread(r)$.
 \item A sequence $z$ satisfying M1, M2 and M3 exists as all the transitions in $L$
 are dependent with some transition in $w$. In the worst case, $z$ may only consist 
 of $r_1'$ if $r_1'$ is dependent with some transition in $nextTrans(s)$ which is executed in $w$. 
 \item $z$ is not a dependence-covering sequence of $w$ at state $s$ as the dependence
 between $r_1'$ and $r$ does not satisfy any constraints of a 
 dependence-covering sequence (Definition~\defdepcovering).
\end{enumerate}

\begin{lemma}
\label{lmm:caseb}
\upshape
EM-DPOR explores a dependence-covering sequence of $w$ from state $s$ when set $L$
satisfies case~B. 
\end{lemma}
\begin{proof}
Consider a sequence $z$ in $\mathcal{S}_G$ constructed as per Construction~\ref{construct:caseb}.
Then, as explained in the proof strategy (Section~\ref{sec:common-construction}) 
let $Z = z.z'.r$ or $Z = z$ based on the existence of shortest $z'$ that enables $r$. EM-DPOR explores a dependence-covering
sequence $u$ for $v = r_2'.r_3'\ldots r_m'.z'.r$ or $v = r_2'.r_3'\ldots r_m'$, making $r_1'.u$ a 
dependence-covering sequence of $Z$. Note that in this case $v$ only consists of transitions from $w$
and if there exists a $z'$ satisfying the criteria considered, $v$ also has transitions from $z'$ .
Also, $r_l' = r_1'$ (see Figure~\ref{fig:common-construct}).
From Section~\ref{sec:proof-strategy}, $\delta_\gamma \in \mathcal{S}_R$ is the state reached by $S.r_1'.\gamma$ 
where $\gamma$ is a prefix of $u$ such that $r \in nextTrans(\delta_\gamma)$.

By Construction~\ref{construct:caseb} and its properties, transition $r_1'$ is the nearest
dependent and may be co-enabled transition which does not happen before $r$ at state $\delta_\gamma$.
Then, \texttt{Explore($S.r_1'.\gamma$)} invokes \texttt{FindTarget($S.r_1'.\gamma,r_1',r$)}.
Line~\linediverging\ in \texttt{FindTarget} (Algorithm~\algofindtarget)
is skipped as $thread(r_1') \neq thread(r)$.
Step~2 of \texttt{FindTarget} identifies state $s$ from where $r_1'$ is executed, as the 
state to add backtracking choices to reorder $r_1'$ and $r$.
We first show that the line~\linecandidateone\ in Algorithm~\algofindtarget\ computes 
$candidates \subseteq$ $\{task(r_2'), task(r_3'), \ldots ,task(r_m'), task(r)\}$. In other words,
set $candidates$ contains no task $(t,e)$ such that transition $next(s,t)$ is in $\gamma \setminus (v \setminus z')$. 
This is because $u$ is a dependence-covering sequence of $v$ and thus there exists no transition 
$p \in u \setminus v$ dependent with a transition $p' \in v$
such that $index(S.r_1'.u,p) < index(S.r_1'.u,p')$. As a result there exists no 
$p \in \gamma \setminus v$, $\gamma$ being a prefix of $u$, such that $p$ is 
dependent with $r$. Hence there exists no $p \in \gamma \setminus v$
such that $p$ happens before $r$. Since $z'$
is a sequence to enable $r$ from state $s_m'$ reached by $z$, no transition in $z'$ 
happens before any transition of $thread(r)$ executed in $z$. Also, $r$ is not yet executed 
in $S.r_1'.\gamma$ and hence by Definition~\defhappensbefore, 
no transition in $z'$ happens before $r$ in $S.r_1'.\gamma$.
Set $candidates$ only consists of those tasks 
whose threads are enabled at state $s$, and in this case only those tasks whose
enabled transition are in sequence $v \setminus z'$ and thus in $w$. Now there are two cases.

\vspace{1.5mm}
\noindent 1. \underline{$candidates \neq \emptyset$} $\,$ Then $threadSet(candidates) \, \cap \, backtrack(s)$ 
is not an empty set.
This is due to line~\lineunexploredthread\ in Algorithm~\algofindtarget, and even if 
$threadSet(candidates)$ $\subseteq done(s)$, $done(s) \subseteq backtrack(s)$ at 
any point of execution of the algorithm. This
contradicts the assumption that set $L$ has no transition from sequence $w$. 

\vspace{1.5mm}
\noindent 2. \underline{$candidates = \emptyset$} $\,$ Then set $pending$ computed at line~\linepending\
in Algorithm~\algofindtarget\ is an empty set since $pending \subseteq candidates$. 
This results in a call to \texttt{BacktrackEager($S.r_1'.\gamma,|S| + 1, r$)}
on line~\linecallbacktrackeager. Note that $index(S.r_1'.\gamma,r_1') = |S| + 1$.

\texttt{BacktrackEager($S.r_1'.\gamma,|S| + 1, r$)} (Algorithm~\algobacktrackeager) 
temporarily copies the HB relation in $\to_{S.r_1'.\gamma}$ to $\rightsquigarrow$ 
(see line~\linecopyhb\ in \texttt{BacktrackEager}).
Then, it orders each pair of co-enabled \texttt{post} transitions in $S$ posting 
events to the same destination thread, and closes the happens-before relation 
$\rightsquigarrow$ with FIFO and transitivity due to newly added \post\ to \post\ 
mappings (lines~\ref{line-post-post-dep}--\ref{line-add-new-hb}). 
We refer to the modified happens-before relation
as \emph{extended} happens-before relation. We show that the extended happens-before 
relation does not order $r_1'$ and $r$ \ie\ $i \not\rightsquigarrow task(r)$ where $i = |S|+1$.
This is because, sequence $S.w$ executes $r$ and not $r_1'$ whereas, sequence $S.r_1'.\gamma$
executes $r_1'$ and not $r$. Hence with a common prefix $S$, sequences $S.w$ and
$S.r_1'.\gamma$ explore both the ordering between $r$ and $r_1'$. 
Thus the order of \texttt{post} operations in $S$ does not 
determine the order of $r_1'$ and $r$. However \texttt{BacktrackEager} only orders 
\texttt{post} operations in $S$ and their respective handlers in $S.r_1'.\gamma$.

Since $i \not\rightsquigarrow task(r)$, Algorithm~\algobacktrackeager\ does not 
return via line~\linereturnonhb\ and proceeds to compute set $candidates_1$ 
on line~\bactracklinecandidateone\ using extended happens-before relation. 
Due to FIFO the handlers posted to the same thread execute in the order in which
they are posted. Since $S$ is a prefix of both $S.w$ and $S.r_1'.\gamma$, the relative 
execution order of event handlers in $w$ and $r_1'.\gamma$ whose events are posted in $S$ is the same.
\texttt{BacktrackEager} only augments happens-before mappings between transitions
of handlers posted in $S$. Thus any new happens-before mappings in the extended happens-before relation, 
between a transition $p$ in $r_1'.\gamma$ and $r$ is such that $p$ is a transition in $w$, 
and $index(S.w,p) < index(S.w,r)$. This along with our earlier reasoning on the 
dependence-covering property of $S.r_1'.u$ proves that $candidates_1$ computed 
by line~\ref{bactrack-line-candidate-1} only contains 
threads whose enabled transitions at $s$ are executed in sequence $w$. If 
$candidates_1 \neq \emptyset$ then line~\backtracklineaddchoice\ of Algorithm~\algobacktrackeager\
adds a thread from $candidates_1$ to $backtrack(s)$. This contradicts the assumption 
that sequence $w$ has no transition from $L$. If $candidates_1 = \emptyset$, then 
$backtrack(s) = enabled(s)$ (line~\backtracklineaddall). 
Then, $thread(r_1) \in backtrack(s)$ which implies $r_1 \in L$. Transition $r_1$ being 
the first transition in $w$ contradicts the assumption that sequence $w$ has no transition from set $L$.
\end{proof}

\subsubsection{Case C}
Case C states that $L$ contains a non-empty subset of non-\post\ transitions such that all 
the non-\post\ transitions in $L$ are dependent with some transition in $w$, and
the first transition in $w$ from $L$ is a non-\post\ transition.
With the help of the transition sequence $z$ described in Section~\ref{sec:common-construction}
we prove that \emdpor\ identifies a transition in $w$ executed prior to the first 
transition from $L$ in $w$ to be reordered with the first transition from $L$
in $w$. We will further prove that this causes a transition in $w$ executed prior
to the first transition from $L$ in $w$, to get added to the set $L$. 
This establishes contradiction to the property of set $L$.

To suit the case under consideration, we refine the construction of sequence $z$ 
as follows.

\begin{construction}
\label{construct:casec}
\upshape
Let $z: s \xrightarrow{r_1'} s_1' \xrightarrow{r_2'} 
s_2' \ldots \xrightarrow{r_m'} s_m'$ in $\mathcal{S}_G$ be a sequence satisfying
the following constraints:
\begin{enumerate}[label=M\arabic*.]
 \item $r_1' \in L$ is a non-\post\ transition and $r_1'$ is the first transition in $w$ from $L$.
 Let $k = index(w,r_1')$.
 \item For all $r_i'$ in $z$, $i \neq 1$, $r_i'$ is a transition in $w$ and 
 $r_i' = r_{i-1}$. 
 \item $\exists r \in nextTrans(s_m')$ such that $r = r_\varpi$ is in $w$ such that
 $\varpi < k$, and $r$ is the first transition in $\alpha: r_1.r_2\ldots r_{k-1}$ to be dependent with $r_1'$.
\end{enumerate}
\end{construction}

\noindent Construction~\ref{construct:casec} differs from Construction~\ref{construct:caseb}
in constraints M1 and M3. Following properties can be inferred for a sequence $z$
constructed as per Construction~\ref{construct:casec}.
\begin{enumerate}[label=P\arabic*.]
 \item All the transitions in $z$ are in $w$.
 \item Transitions $r_1'$ and $r$ may be co-enabled \ie\ $thread(r_1') \neq thread(r)$. 
 This is because, $r_1' = next(q,thread(r_1'))$
 at any state $q$ visited by a prefix of sequence $\alpha$, including the state $s_{\varpi-1}$ where $r = r_\varpi$
 is executed. Since both $r_1'$ and $r$ are in $nextTrans(s_{\varpi-1})$, $thread(r_1') \neq thread(r)$.
 \item A sequence $z$ satisfying M1, M2 and M3 exists only if $\exists r_1' \in L$ such that $r_1'$ is dependent
 with some transition prior to its index in $w$. In such a case $z$ can at least consist 
 of $r_1'$ if $r_1'$ is dependent with some transition in $nextTrans(s)$ which is executed prior to $r_1'$ in $w$. 
 If $r_1'$ is not dependent with any transition prior to it in $w$, then $z$ does not exist. Nevertheless, 
 as will be shown in Lemma~\ref{lmm:casec}, we get a dependence-covering sequence of $w$ from 
 $s$ in $\mathcal{S}_R$ without constructing $z$ in such a case.
 \item $z$ is not a dependence-covering sequence of $w$ at state $s$ as the dependence
 between $r_1'$ and $r$ does not satisfy any constraints of a 
 dependence-covering sequence (Definition~\defdepcovering).
\end{enumerate}

\begin{lemma}
\label{lmm:casec}
\upshape
EM-DPOR explores a dependence-covering sequence of $w$ from state $s$ when set $L$
satisfies case~C. 
\end{lemma}
\begin{proof}
In the context of case C two sub-cases exist:\\
I. \textbf{\emph{Assume there exists a transition \boldmath$r_1' \in L$ executed in $w$ such that $k = index(w,r_1')$
and $r_1'$ is independent with all the transitions $r_i$ in $w$ for $1 \leq i < k$\unboldmath.}}
Then by Definition~\defdependencerelation\ transition $r_1'$ commutes 
with all such $r_i$ and thus $r_{k-1}( \ldots r_2(r_1(r_1'(s)))\ldots) = s_k$. 
Since $r_{k+1}\ldots r_{n-1} r_n$ is enabled at $s_k$, sequence $r_1' r_1 \ldots 
r_{k-1} r_{k+1} \ldots r_n$ is a dependence-covering sequence of $w$ from state $s$.

\vspace{1.5mm}
\noindent II. \textbf{\emph{Assume no transition in \boldmath$L$ executed in $w$ satisfies sub-case~I. 
Let $r_1'$ be the first transition in $w$ from $L$, and let $r_1'$ be dependent 
with some transition $r_i$ in $w$ where $1 \leq i < k$ where $k = index(w,r_1')$\unboldmath.}}
Consider a sequence $z$ in $\mathcal{S}_G$ constructed as per Construction~\ref{construct:casec}.
Sequence $z$ exists as sub-case II satisfies the pre-condition of property~P3 of 
Construction~\ref{construct:casec}.
As explained in Section~\ref{sec:common-construction}, let $Z = z.z'.r$ or $Z = z$ based on the 
existence of shortest $z'$ that enables $r$. EM-DPOR explores a dependence-covering
sequence $u$ for $v = r_2'.r_3'\ldots r_m'.z'.r$ or $v = r_2'.r_3'\ldots r_m'$, making $r_1'.u$ a 
dependence-covering sequence of $Z$. 
Note that here $r_l' = r_1'$.
From Section~\ref{sec:proof-strategy}, $\delta_\gamma \in \mathcal{S}_R$ is the state reached by $S.r_1'.\gamma$ 
where $\gamma$ is a prefix of $u$ such that $r \in nextTrans(\delta_\gamma)$.

By Construction~\ref{construct:casec} and its properties, transition $r_1'$ is the nearest
dependent and may be co-enabled transition which does not happen before $r$. 
Then, \texttt{Explore($S.r_1'.\gamma$)} invokes \texttt{FindTarget($S.r_1'.\gamma,r_1',r$)}.
Line~\linediverging\ in \texttt{FindTarget} is skipped since $thread(r_1') \neq thread(r)$. 
Step~2 of \texttt{FindTarget} identifies state $s$ from where $r_1'$ is executed, as the 
state to add backtracking choices to reorder $r_1'$ and $r$.
We show that the line~\linecandidateone\ 
computes $candidates \subseteq \{task(r_1), task(r_2)$ $\ldots task(r_{k-1})\}$ \ie\ 
$candidates$ has no task $(t,e)$ such that transition $next(s,t)$ is in $\gamma \setminus \alpha$. 
The reason for this is similar to a corresponding step in the proof of Lemma~\ref{lmm:caseb}.
Set $candidates$ only consists of those tasks 
whose threads are enabled at state $s$, and in this case only those tasks whose
enabled transition at $s$ are in $\alpha : r_1.r_2\ldots r_{k-1}$, a prefix of $w$. Now there are two cases.

\vspace{1.5mm}
\noindent 1. \underline{$candidates \neq \emptyset$} $\,$ Then $threadSet(candidates) \, \cap \, backtrack(s)$ 
is not an empty set. This is due to line~\linereturnsuccess\ in Algorithm~\algofindtarget. 
This contradicts the assumption that $r_1'$ is the first transition in $w$ from $L$, as
index of any transition in $\alpha$ is lesser than $k = index(w,r_1')$. 

\vspace{1.5mm}
\noindent 2. \underline{$candidates = \emptyset$} $\,$ Then set $pending$ computed at line~\linepending\ in 
Algorithm~\algofindtarget\ is also an empty set as $pending \subseteq candidates$. 
This results in a call to \texttt{BacktrackEager($S.r_1'.\gamma,|S| + 1, r$)}
on line~\linecallbacktrackeager\ in Algorithm~\algofindtarget. Then, we 
can use a reasoning similar to that in the proof of Lemma~\ref{lmm:caseb} to derive 
contradiction for the assumption that $r_1'$ is the first transition in $w$ from $L$.
\end{proof}

\subsubsection{Case D}
\label{sec:caseD}
 
Case D assumes all the transitions in the set $L$ to be of the type \texttt{post}
such that none of the transitions in $L$ are in the transition sequence $w$.
Unlike the lemmas related to cases~B and C proving which involved reasoning about 
non-\post{} transitions in the set $L$, case~D requires reasoning about \post{} 
transitions in the set $L$. Similar to the proof strategy of cases 
B and C we will show that \emdpor\ identifies some transition in $w$ to 
be reordered with a \post{} transition in the set $L$. Since \emdpor\ considers a \post{}
transition to be independent w.r.t. all the transitions, the \texttt{Explore} algorithm
never invokes \texttt{FindTarget} to reorder a \post{} with some other transition.
However, a recursive call to \texttt{FindTarget} may reorder a pair of \texttt{post}s
to the same event queue (see step~1 of Algorithm~\algofindtarget\ (\texttt{FindTarget}) 
and step~3c of Algorithm~\algoreschedulepending\ (\texttt{ReschedulePending})).
For this to happen the pair of \texttt{post}s will have to be somehow related 
to a pair(s) of dependent transitions which are originally identified
for reordering by \texttt{Explore}. However, establishing this relation between
a pair(s) of dependent transitions and a pair of \post{} transitions is 
non-trivial and may involve reasoning about a set of transition sequences ultimately
leading to the identification of \texttt{post}s to be reordered starting from 
a transition sequence which identifies a pair of dependent transitions to be reordered. 

With the help of the transition sequence $z$ described in Section~\ref{sec:common-construction}
we will be generating a set of transition sequences of interest, which will lead to 
the identification of a \post{} in $w$ to be reordered with a \post{} in the backtracking
set $L$. After identifying these \texttt{post}s and establishing that \texttt{FindTarget}
will be called to reorder them, we will show that 
a transition from $w$ gets added to the set $L$ by invoking arguments similar
to the proofs of Lemma~\ref{lmm:caseb} and \ref{lmm:casec}. 
This establishes contradiction to the property of set $L$, and in turn 
establishes the existence of a dependence-covering sequence of $w$ starting from 
some transition belonging to the set $L$.
\\
\\
We make the following assumptions to simplify the proof sketch. 
\begin{assumption}
\label{assume:construction}
\upshape
For any transition sequence $\alpha$ executed from state $s$ considered 
henceforth (including sequence $w$ and $z$) we make the following assumptions. Let $K_{\alpha}$ 
be the set of transitions in $\alpha$ such that for each $k \in K_{\alpha}$, $k \not\in w$.  
Then,
\begin{enumerate}
 \item Transitions in $K_{\alpha}$ do not form a deadlock cycle consisting only of transitions
 in $K_{\alpha}$.
 \item Every lock acquired inside a task is released within the same task 
 (thread  or event handler). This assumes that a lock acquire and release does not 
 span multiple event handlers.
 \item If a transition $p \in K_{\alpha}$ is disabled at a state, then it does not require a
 transition in $w$ to be executed for $p$ to eventually enable.
\end{enumerate}
\end{assumption}

Reasoning about this case requires a few new definitions which we introduce below.

\begin{definition}
 \label{def:dep-graph}
 \upshape
 A function \boldmath $DG(p,\alpha,s')$ \unboldmath which defines a dependence graph
 of a transition w.r.t. a sequence, takes a transition $p$
 and a sequence $\alpha$ executed from a state $s'$ where $p \in R_\alpha$ or 
 $p \in nextTrans(last(\alpha))$, and returns a set $P$ such that $P \subseteq R_{\alpha}$ 
 and for each transition $a \in P$, if $p \in R_\alpha$ then $index(\alpha,a) \to_{\alpha} index(\alpha,p)$
 else either $index(\alpha,a) \to_{\alpha} task(p)$ or $a$ is dependent with $p$.
\end{definition}

The following definition gives the criteria when dependence graphs of a transition
w.r.t. two different sequences are equivalent.
\begin{definition}
 \label{def:identicalDG}
 \upshape
 Predicate \boldmath$identicalDG(p,\alpha,\beta,s')$ \unboldmath takes a transition 
 $p$, and transition sequences $\alpha$ and $\beta$, both executed from the
 same state $s'$, such that $p \in R_\alpha \cup nextTrans(last(\alpha))$, 
 $p \in R_\beta \cup nextTrans(last(\beta))$, and the predicate
 evaluates to \texttt{TRUE} only if $DG(p,\alpha,s') = DG(p, \beta,s')$.
\end{definition}

\begin{definition}
\label{def:future}
\upshape
A set \boldmath $\mathit{future}$\unboldmath\ of a \post\ operation $r$, \ie\ $\mathit{future}(r)$ 
is a set of transitions such that a transition $r' \in \mathit{future}(r)$ if,
(1)~there exists a sequence $\alpha$ in $\mathcal{S}_G$ such that $r'$ is executed by 
the handler of the event posted by $r$ in $\alpha$, or
(2)~$r' \in \mathit{future}(r'')$ such that $r''$ is a \texttt{post} operation and $r'' \in \mathit{future}(r)$.
\end{definition}

\begin{definition}
\label{def:enabledfuture}
\upshape
A set \boldmath $\mathit{enabledFuture}$\unboldmath\ of a \post\ operation $r$, 
\ie\ $\mathit{enabledFuture}(r)$ is a set of transitions such that a transition 
$r' \in \mathit{enabledFuture}(r)$ if,
(1)~$r'$ $\in$ $\mathit{future}(r)$, or (2)~there exists a sequence $\alpha$ in 
$\mathcal{S}_G$ reaching a state $s'$ such that a transition $r'' \in \mathit{future}(r)$
is blocked in $s'$ and $r'$ is the first transition of a shortest sequence from 
$s'$ which enables $r''$, 
or (3)~$r' \in \mathit{enabledFuture}(r'')$ such that 
$r''$ is a \texttt{post} operation and $r'' \in \mathit{enabledFuture}(r)$.
\end{definition}

\begin{construction}
\label{construct:cased}
\upshape
Let $z: s \xrightarrow{r_1'} s_1' \xrightarrow{r_2'} 
s_2' \ldots \xrightarrow{r_m'} s_m'$ in $\mathcal{S}_G$, 
where $r_1' \in L$ and $v = r_2'\ldots r_m'$, 
be a sequence satisfying the following constraints:
\begin{enumerate}[label=M\arabic*.]
 \item Sequence $v$ consists of transitions belonging to $w$ as well as transitions
 outside $w$. For a transition $r_i' \in v$, if $r_i' \in w$ then $identicalDG(r_i',w,z,s)$. 
 For a transition $r_i' \in v$, if $r_i' \not\in w$ then $r_i'$ is the first transition 
 of a shortest sequence from the state $s_{i-1}'$, comprising only of transitions which do not belong to $w$, 
 to be executed to make an event $e$ which was dequeued in $w$ but blocked in $s_{i-1}'$ executable.
 
  
 \item There exists no extension to any prefix $\alpha$ of $z$ which results in 
 a transition sequence $\alpha.\gamma$ such that there exists a pair of dependent transitions $c$ 
 and $d$ with the following properties:
 \begin{enumerate}
  \item (i)~Either $c$ and $d$ are transitions in $w$ such that they are ordered
  differently in $\gamma$ compared to their order in $w$, or (ii)~$c \not\in w$,
  $c$ is a transition in the $\mathit{enabledFuture}$ set of a \post{} in $w$, and
  $d \in w$ such that $c$ is executed prior to $d$ in $\alpha.\gamma$, and 
  
  \item Attempting to reorder $c$ and $d$ through some other extension to $\alpha$ will 
  only result in a transition sequence $\alpha.\gamma'$ which breaks the order between another pair 
  of dependent transitions $c'$ and $d'$ such that either (i)~$c',d' \in w$, or 
  (ii)~$d' \in w$, $c'$ is a transition in the $\mathit{enabledFuture}$ set of a \post{} in $w$,
  and $c'$ is executed prior to $d'$ in $\alpha.\gamma'$.
 \end{enumerate}

 \item There exists a transition $r \in nextTrans(s_m')$ such that $r$ is a transition
 in $w$ and $r$ is dependent with a transition $r_l' \not\in w$ executed in $v$ such that 
 $r_l' \in \mathit{enabledFuture}(r_1')$.
 
 \item $z$ is a sequence with maximum transitions from $w$ while satisfying the 
 constraints M1, M2 and M3.
\end{enumerate}
\end{construction}

From the constraints given in Construction~\ref{construct:cased}, a pair of transitions in $w$ posting
to the same event queue can be reordered in sequence $z$ so long as the properties M1 and M2
are respected. We now present a lemma describing the property of transitions not in $w$
but present in $v$.

\begin{lemma}
\label{lmm:transition-property-cased}
\upshape
In a transition sequence $z = r_1'.v$ constructed by only following the constraint M1
of Construction~\ref{construct:cased}, a transition $r_i' \in v$ such that 
$r_i' \not\in w$ satisfies one of the following properties.
\begin{enumerate}
 \item $r_i'$ $\in$ $\mathit{enabledFuture}(r_1')$, or 
 \item $r_i' \in \mathit{enabledFuture}(r_j')$ where $r_j'$ is a \post{} transition in $w$.
\end{enumerate}
\end{lemma}

\begin{proof}
We prove this by inducting on the order in which transitions not belonging to
$w$ are added to $v$. 

\paragraph{Base case.} 
Let $r_i'$ be the first transition in $v$ which does not belong to $w$.
Recall that $r_1' \in L$ is the first transition of $z$, and from the 
property of case~D we know that $r_1' \not\in w$ and it is a \post{} transition. 
This makes $r_i'$ the second
transition in $z$ to not be from the sequence $w$.
As per constraint M1 of Construction~\ref{construct:cased}, a transition not 
belonging to $w$ is added to $v$ only to make an event dequeued in $w$ executable. 
Assume that $r_i'$ has been added to make an event $e$ blocked in the state $s_{i-1}'$ 
and dequeued in $w$, executable. 
Let $E$ be the set of events on $\threadof(e)$ such that for each event $e' \in E$,
either $e'$ is the executable event on $\threadof(e)$ at state $s_{i-1}'$ or 
$e'$ is an event blocked in $e$'s event queue in state $s_{i-1}'$ such that
$e'$ is dequeued prior to $e$.
If any of the events in $E$ is dequeued in $w$, then executing $r_i'$ breaks the 
property M1. This is because in such a case, the shortest sequence to make $e$ executable will also
comprise of transitions from $w$. Also, if any event $e'$ in $E$ is posted by 
a transition in sequence $S$ prior to reaching the state $s$ from where $w$ is assumed to be executed, then
due to FIFO ordering we can infer that $e'$ is dequeued prior to $e$ in sequence
$w$ as well.
By elimination, each event in $E$ is either 
(a)~the event posted by $r_1'$, since $r_1'$ is the only transition in $z$ that 
does not belong to $w$ when state $s_{i-1}'$  is reached, or (b)~an event posted 
by a transition in $w$ but was not dequeued in $w$. 
The shortest sequence to make $e$ executable will atleast comprise of the transitions
in the handlers of the events in $E$, and $r_i'$ could be the \abegin{} transition of 
the executable event on $\threadof(e)$ which too is in $E$. 
However, if a transition, say $b$, in the handler of $e' \in E$ is blocked
on some transition outside the handler of $e'$ (\eg\ if it is a transition acquiring a 
lock held by some other thread), then the shortest transition sequence to make
$e$ executable will also include transitions to enable $b$; transition $r_i'$ could
be a transition executed to eventually enable $b$. In either case, $r_i'$ 
satisfies the constraints to be in $enabledFuture(r_1')$ or $enabledFuture(r_j')$ 
where $r_j'$ is a transition posting an event in $E$. Hence proved.


\paragraph{Induction hypothesis.} All the transitions upto $k^{th}$ transition added to $v$  
which do not belong to $w$ either belong to $\mathit{enabledFuture}(r_1')$, or
belong to $\mathit{enabledFuture}(r_j')$ where $r_j'$ is a \post{} transition in $w$.

\paragraph{Induction step.} We need to show that one of the two properties
specified in the lemma holds even for the $(k+1)^{th}$ transition, say $r_i'$, that does not
belong to $w$ but is added to $v$. Then from the condition M1, $r_i'$ should be the 
first transition in a shortest sequence comprising only of transitions not in
$w$, to make an event $e$ executable such that $e$ is dequeued in $w$
but is currently blocked in the state $s_{i-1}'$. Then, either 
$r_i'$ is a transition in the executable task on $\threadof(e)$, or it is a 
transition that must be executed so as to eventually enable a transition $b \not\in w$ in the executable 
task on $\threadof(e)$ or the handler of an event prior to $e$ in $e$'s event queue. 
Otherwise, $r_i'$ can be removed to obtain a shorter sequence executing which
can make $e$ an executable event. Let us firstly reason about the case where $r_i'$ is a transition
in the executable task on $\threadof(e)$, and let $e'$ be the corresponding event
of the executable task. Then, \post{($e'$)} which is clearly executed
prior to $r_i'$, is either a transition in $w$ or a transition not in $w$. In the latter 
case due to induction hypothesis, \post{($e'$)} satisfies one of the 
two properties listed in the lemma. If \post{($e'$)} is a transition in $w$ or
\post{($e'$)} is a transition in the $\mathit{enabledFuture}$ set of a \post{} transition in $w$ 
(as per property~2 listed by the lemma), then by Definition~\ref{def:enabledfuture} transition 
$r_i'$ too belongs to the $\mathit{enabledFuture}$ set of a \post{} transition in $w$
thus satisfying condition~2. If \post{($e'$)} $\in$ $enabledFuture(r_1')$ 
(as per property~1 listed in the lemma), then $r_i'$ being in the handler of 
$e'$ satisfies the constraints to be in $enabledFuture(r_1')$.

Now consider the case where $r_i'$ has been added to eventually enable a transition $b \not\in w$ in the executable 
task on $\threadof(e)$ or the handler of an event prior to $e$ in $e$'s event queue.
Let $e'$ be the event corresponding to the handler in which the transition $b$
is executed.
Now we can show that $r_i'$ belongs to $enabledFuture(r_1')$ or 
$\mathit{enabledFuture}(r_j')$ where $r_j'$ is a \post{} transition in $w$, 
by reasoning about \post{($e'$)} similar to the first case presented above. 
\end{proof}

\begin{lemma}
\label{lmm:witness-cased}
\upshape
A transition sequence $z = r_1'.v$ which satisfies the constraints M1, M2 and M3
exists in $\mathcal{S}_G$.
\end{lemma}

\begin{proof}
A transition sequence $z = r_1'.v$ satisfying the constraint M1 in Construction~\ref{construct:cased} 
trivially exists in $\mathcal{S}_G$.
One such sequence can be constructed by concatenating $r_1'$ with a prefix of $w$ till the next 
transition to be executed in $w$ is a \abegin{} transition whose event $e$ is blocked 
on $\destof{(r_1')}$, such that the event posted by $r_1'$ must be dequeued and
handled for $e$ to become executable. If there is no such prefix then we must be
able to execute $r_1'.w$ which clearly is a dependence-covering sequence of $w$,
since $r_1'$ being a \post{} transition is independent w.r.t. all the transitions
in $w$ as per the dependence relation defined in Definition~\defdependencerelation.

Now, let $z = r_1'.v$ be a transition sequence in $\mathcal{S}_G$ satisfying the 
constraints of M1. Our main assumption is that there exists no transition
sequence starting from any transition in the backtracking set $L$ at $s$ which
is a dependence-covering sequence of $w$. Then, the sequence $z$ must reach a 
state $s_m'$ where a transition $r \in nextTrans(s_m')$
is such that $r \in w$, $identicalDG(r,w,z,s)$ $=$ \texttt{FALSE} and no extension
$\gamma$ to $z$ can result in $identicalDG(r,w,z.\gamma,s)$. Otherwise, we can 
obtain a sequence $z$ which is a dependence-covering sequence of $w$. Since all the transitions belonging to 
$w$ in the constructed sequence $z$ have their $DG$ identical to that found
in $w$, there can be only two causes for the $DG(r,z,s)$ to be different from that w.r.t. $w$ ---
(1)~a transition $r_l' \not\in w$ present in $v$ is dependent with $r$, or
(2)~$r \in nextTrans(s_m')$ must be executed so as to execute a dependent transition $r' \in w$
such that $r'$ is executed prior to $r$ in $w$. From Lemma~\ref{lmm:transition-property-cased},
if it is case~(1) then there can be two subcases: (1.a)~$r_l' \in \mathit{enabledFuture(r_1')}$, or  
(1.b)~$r_l' \in \mathit{enabledFuture}(r_j')$ where $r_j'$ executed prior to $r_l'$ in $v$ 
is a \post{} transition in $w$. 
 
Let us assume that any sequence $z = r_1'.v$ satisfying the constraints of M1 
can only satisfy the cases 1.b or 2 defined above. We will establish a 
contradiction for this assumption thus establishing the validity of constraint 
M2 and M3 given in Construction~\ref{construct:cased}. Towards this, we construct
a sequence $z$ by strengthening the constraints of M1. Let $z = r_1'.v$ be a 
transition sequence constructed as per M1 as well as a constraint that for every pair of \texttt{post}s
$p_1 \in w$ and $p_2 \in w$ executed in $z$ and posting events to the same event queue,
$index(z,p_1) < index(z,p_2)$ iff $index(w,p_1) < index(w,p_2)$. Let $z$ be the longest
such sequence. 
Then, $z$ reaches 
a state $s_m'$ where a transition $r \in nextTrans(s_m')$ is such that $r \in w$ 
and $r$ satisfies one among the cases 1.b or 2 or (3)~$r$ is a \post{} transition 
which must be executed so as to execute a \post{}  $r' \in w$ posting 
to the same event queue as $r$ 
such that $r'$ is executed prior to $r$ in $w$. If $z$ satisfies condition 3 then 
executing $r$ at $s_m'$ violates the constraint on ordering between \post{} transitions.

The reasoning for cases 2 and 3 are similar. Let us firstly assume that $r$ satisfies
either of case~2 or 3. This indicates the existence of a transition $r' \in w$
either dependent with $r$ (if case 2) or posting to the same event queue as $r$ 
(if case 3) such that $r$ must be executed to eventually execute $r'$, even though $r'$ was executed
prior to $r$ in $w$. 
Let a sequence $\gamma$ from $s_m'$ be the shortest sequence 
such that $r' \in nextTrans(last(z.\gamma))$. Let $P$ be a set of transitions 
such that a transition $p \in P$ if $index(S.z.\gamma,p) \to_{S.z.\gamma} \taskof(r')$
and $p \in \gamma$. Since $r'$ is a transition in $w$, all the transitions in the 
set $P$ too are in $w$.
We argue that either there exists a transition $p \in P$ such that $task(p)$
is blocked on $\threadof(r)$ at state $s_m'$, or attempting to reorder $r$ 
and $r'$ by executing transitions related to $r'$ (for example those in the set $P$ 
or in event handlers executed prior to handlers containing some transitions in $P$)
prior to $r$ breaks the ordering between another pair of transitions from $w$
such that these transitions are either dependent or post to the same event queue.
In the latter case we argue that attempting to reorder the new adversely ordered 
transitions in turn breaks the ordering between some other pair of transitions 
from $w$ and so on. Since the state space we consider is finite and acyclic, 
continuing this process of reordering adversely ordered transitions eventually
causes $r$ to execute prior to $r'$. 
If this is not the case then $r'$ can be executed prior to $r$ which indicates the 
existence of some other sequence longer than $z$  and satisfying the constraints 
M1 and ordering restriction between \texttt{post}s. This violates our assumption of $z$ being 
the longest such sequence.
All these pairs of transitions including $r$ and $r'$ 
were ordered as desired in sequence $w$. This indicates the presence of a pair of 
events posted to the same event queue by a pair of transitions in $z$ such that these 
events were differently ordered in $w$. This contradicts the 
constraints of $z$ which should have preserved the relative ordering between
transitions in $w$ posting to the same event queue. Hence, $z$ does not satisfy
case~2 or 3.

Now consider the case~1.b according to which $r$ is dependent with a transition
$r_l' \in \mathit{enabledFuture}(r_j')$ such that $r_l' \not\in w$ and $r_j'$  
is a \post{} transition in $w$. We choose nearest such $r_j'$ w.r.t. $r_l'$. In 
other words, if $r_l' \in \mathit{enabledFuture}(r_i')$ and $r_l' \in \mathit{enabledFuture}(r_k')$
such that $r_k' \in \mathit{enabledFuture}(r_i')$ where $\{r_i',r_k'\} \subseteq z$ are \post{}
transitions in $w$, then we choose $r_j' = r_k'$. From (i)~the constraints of Definition~\ref{def:enabledfuture}, 
(ii)~knowing that a subset of transitions in $\mathit{enabledFuture}(r_j')$ executed 
in $z$ were not executed in $w$ and (iii)~$r_j'$ being the nearest such \post{} to whose 
$\mathit{enabledFuture}$ set $r_l'$ belongs, we infer that a suffix of the handler
of the \post{} $r_j' \in w$ is not executed in $w$. Let $e$ be the event posted 
by $r_j'$. Either a set of transitions from such a suffix of $e$'s handler are 
included in $z$, or $r_l'$ is a transition executed to eventually
enable a transition from such a suffix of $e$'s handler. In 
either case since these transitions do not belong to $w$, these transitions have 
been added to $z$ so as to make an event $e'$ dequeued in $w$ executable such that $e'$ is 
blocked on $e$'s event queue when $e$ is the executable event. From the concurrency
semantics of the event-driven model considered, an event handler is executed to 
completion before the next event on the corresponding event queue is dequeued.
We know that both the events $e$ and $e'$ are posted 
to $w$. Then, 
atleast one among the handlers of $e$ and $e'$ must have been executed to completion. 
Since we have assumed that $e$'s handler is partially executed in $w$, clearly
event $e'$'s handler has been executed to completion in $w$. This implies that
$index(w,\post{($e'$)}) < index(w,\post{($e$)})$ (due to FIFO processing of events). This inference contradicts 
their ordering in $z$ since $e'$ is blocked when $e$ is executable in $z$. 
This in turn violates the constraints assumed on $z$. Thus, we have shown that 
a transition sequence $z$ constructed this way cannot satisfy case~1.b.


From the above arguments we have established the existence of a transition sequence 
$z = r_1'.v$ constructed as per the constraint M1 for which neither of cases 1.b or 2 holds.
Then, such a transition sequence must satisfy the case 1.a according to which a transition 
$r \in nextTrans(s_m')$ is 
dependent with a transition $r_l' \not\in w$ such that $r_l' \in \mathit{enabledFuture(r_1')}$. 
This shows that a transition sequence satisfying the constraints M1, M2 and M3
of Construction~\ref{construct:cased} exists in $\mathcal{S}_G$.
\end{proof}

\begin{lemma}
\label{lmm:reorder-impossible}
\upshape
Reordering $r$ and $r_l'$ identified by Construction~\ref{construct:cased} by 
reordering some pair of transitions in $v$, is either not possible or will only 
result in a sequence $v'$ (executed from state $s_1'$) such that for $v'$ and 
any extension $\gamma$ to $v'$ one of the 
following holds --- (i)~the dependence graph of a transition 
$r_j' \in R_{v'.\gamma} \cup nextTrans(last(v'.\gamma))$ such that $r_j'$ is executed in 
both $w$ and $z$, becomes non identical to $DG(r_j',w,s)$, or (ii)~there exists a 
pair of dependent transitions $p \in v'.\gamma$ and $q \in nextTrans(last(v'.\gamma))$ such that 
$q$ is executed prior to $p$ in $w$, or (iii)~there exists a 
pair of dependent transitions $p \in v'.\gamma$ and $q \in nextTrans(last(v'.\gamma))$ such that 
$q \in w$, $p \not\in w$ and $p \in \mathit{enabledFuture(p')}$ where $p'$ is a \post{} in $w$.
\end{lemma}

\begin{proof}
We prove this by contradiction. Assume that a resulting transition sequence $v'$
executes $r$ before $r_l'$ such that $identicalDG(r,w,r_1'.v',s)$ holds, and without resulting 
in any scenario listed in (i), (ii) or (iii) above. Then such a transition sequence 
$r_1'.v'$ clearly has more number of transitions from $w$ compared to $z = r_1'.v$, with their dependence
graphs consistent with that found in the context of $w$. Also, since no 
extension to $v'$ satisfying (ii) or (iii) is possible when $r$ and $r_l'$ are reordered as per our assumption,
an extension to $r_1'.v'$, say $\gamma$, must hit a state where a transition 
$r' \in nextTrans(last(r_1'.v'.\gamma))$ belonging to $w$ is dependent with a transition 
$r'' \in v'$ such that $r'' \not\in w$. If not, an extension to $r_1'.v'$ will result in a dependence-covering
sequence for $w$. Then, $r_1'.v'.\gamma$ is a transition sequence satisfying the constraints
M1, M2 and M3 of Construction~\ref{construct:cased}, and having more transitions from $w$ than $z$. This implies that $z$
did not satisfy the constraint M4 of Construction~\ref{construct:cased}. Thus,
one of the properties (i), (ii) or (iii) must hold on any such sequence $r_1'.v'$.
\end{proof}

\begin{lemma}
\label{lmm:r1P-must-reorder-post} 
\upshape
Transition $r_1'$ must be reordered w.r.t. some transition $r_{\mu}' \in v$ (where $z = r_1'.v$)
belonging to $w$ such that $r_{\mu}'$ posts an event to the same destination 
event queue as $r_1'$, so as to obtain a transition sequence $z'$ from the state $s$ which satisfies
the following properties --- (1)~either $r$ is executed prior to $r_l'$ in $z'$ 
or only $r$ is executed in $z'$, (2)~every transition $r_j'$ in $z$ which also belongs to 
$w$ is executed in $z'$ such that $identicalDG(r_j',w,z',s)$, (3)~there exists 
atleast one extension $\gamma$ to $z'$ where neither of the following hold: (a)~there exists
a pair of dependent transitions $p \in z'.\gamma$ and $q \in nextTrans(last(z'.\gamma))$ 
such that $q$ is executed prior to $p$ in $w$, or (b)~there exists a pair of dependent 
transitions $p \in z'.\gamma$ and $q \in nextTrans(last(z'.\gamma))$ such that $q \in w$, 
$p \not\in w$ and $p \in \mathit{enabledFuture(p')}$ where $p'$ is a \post{} in $w$.
\end{lemma}

\begin{proof}
From Lemma~\ref{lmm:reorder-impossible} we have established that attempting to 
reorder $r_l'$ and $r$ by reordering any pair of transitions in the sequence $v$
(executed from state $s_1'$ reached on executing $r_1'$ from state $s$)
including $r_l'$ and $r$ themselves, can only result in a sequence which 
satisfies conditions (i), (ii) or (iii) listed in Lemma~\ref{lmm:reorder-impossible}
which are clearly not consistent with the constraints 1, 2 and 3 listed in
this lemma. The transition sequence $w$ executed from state $s$ satisfies all of the
conditions 1, 2 and 3 listed above. However, transition $r_1'$ is not executed in 
$w$. Hence, some transition, say $r_{\mu}'$, belonging to $w$ and executed in $v$ must be reordered w.r.t. $r_1'$ 
so as to reorder $r$ and $r_l'$ with neither breaking the dependence graphs of  
transitions in $z$ belonging to $w$ nor resulting in scenarios described by 3(a) or 3(b)
listed in the lemma. Then, the transition $r_{\mu}'$ too must be a transition posting to 
the same event queue as the destination event queue of $r_1'$. This is because 
if $r_{\mu}'$  is a non-\post{} transition or $r_{\mu}'$  is a transition posting to some other 
queue, then reordering $r_{\mu}'$ and $r_1'$ neither alters the final global state 
reached nor the final event queue configuration which affects the order 
between event handlers, because $r_1'$ would then commute with such a $r_{\mu}'$. 
\end{proof}

\paragraph{Intuition to prove that \emdpor\ explores a dependence-covering sequence for case~D.}
%
In the cases B and C the transition $r_1'$ belonging to the backtracking set $L$ and
executed at state $s$, was a non-\post{} transition. Hence a transition in $w$, say $r$, 
which was dependent with $r_1'$ could be easily identified leading to a \emph{non} dependence-covering 
sequence of $w$. We then argued that \emdpor\ would 
attempt to reorder $r_1'$ and $r$, and add a backtracking choice at state $s$ which
would break the property assumed on the backtracking set $L$ at $s$. Thus we were
able to prove the existence of a dependence-covering sequence of $w$ starting from
a transition in the set $L$, through proof by contradiction.

In case~D however $r_1' \in L$ executed at state $s$ is a \post{} transition, which 
makes it harder to identify a transition in $w$ that must be reordered with $r_1'$ 
to obtain a dependence-covering sequence of $w$. This is because even though ordering 
between events posted to the same queue affect the ordering between dependent 
transitions, the ordering between transitions posting to the same event queue
are not directly captured in a dependence-covering sequence. Hence, the 
influence of $r_1'$ on the non-\post{} transitions in $w$ can only be identified
through the interference from non-post{} transitions in the $\mathit{enabledFuture}$ set of 
$r_1'$. Lemma~\ref{lmm:r1P-must-reorder-post} establishes that $r_1'$ must be reordered
with a transition $r_{\mu}' \in v$ posting to the same event queue as the destination
of $r_1'$ where $r_{\mu}' \in w$, so as to explore more transitions from $w$ in the 
resulting sequence but without resulting in adversarial scenarios i, ii and iii listed in 
Lemma~\ref{lmm:reorder-impossible}. This could eventually lead to a dependence-covering
sequence of $w$. In order to identify such an $r_{\mu}'$, we systematically generate 
a set $\Gamma$ of transition sequences starting from $z$ by flipping the ordering 
between certain dependent transitions and transitions posting to the same event queues
belonging to $R_z \cup nextTrans(last(S.z))$. A few pairs of dependent transitions
and transitions posting to the same event queues and seen in transition sequences of the
set $\Gamma$, are encoded as a tree called \gtree. Intuitively the \gtree\ encodes
all pairs of dependent transitions and \post{} operations explored in the subspace
reached from $s_1'$ (state reached on executing $r_1'$) such that exploring every pair 
of transitions in \gtree\ in a manner consistent to obtain a dependence-covering
sequence of $w$, requires reordering $r_{\mu}'$ with $r_1'$. We will then show that
\emdpor\ too is capable of identifying all the transition pairs of \gtree\ ultimately
leading to the identification that $r_{\mu}'$ must be reordered with $r_1'$.
\\
\\
To aid the proof we define a tree called \gtree\ which can encode certain transition
pairs which are of interest to the proof.

\begin{definition}
\label{def:gtree}
\upshape
\gtree\ is a tree with each of its nodes being a set of ordered pairs of transitions
and its root node being $\{(r_1',r_{\mu}')\}$. 
Let $\Gamma_{node} = \{(c_1,d_1),(c_2,d_2) \ldots (c_{\chi},d_{\chi})\}$
be a non-root node in the \gtree. Then, 
\begin{enumerate}[label=N\arabic*.]
\item For all $i \in [1,\chi]$, $d_i \in w$ and $c_i$ may or may
not be a transition of $w$ such that if $c_i \in w$ then $d_i$ is executed prior to $c_i$
in $w$. If $c_i \not\in w$ then either $c_i \in \mathit{enabledFuture}(r_1')$,
or $c_i \in \mathit{enabledFuture}(r_j)$ where $r_j$ is a \post{} transition in $w$.

\item Each pair $(c_i,d_i)$ is such that either $c_i$ and $d_i$ are dependent
transitions, or $c_i$ and $d_i$ are transitions posting events to the same event queue.

\item If $\Gamma_{node}$ is a leaf node of \gtree\ then it only contains pairs of 
dependent transitions as its members. If it is a non-leaf node then it contains atleast 
one pair of transitions posting to the same event queue. 

\item Transition pairs $c_i$ and $d_i$ can either be from two 
different threads or two different handlers on the same thread. If it is the latter
then $\Gamma_{node}$ is a singleton set. 
However, if $\Gamma_{node}$ has multiple transition pairs then 
each pair $(c_i,d_i)$ are such that $\threadof(c_i) \neq \threadof(d_i)$. 

\item There exists a subspace $\mathcal{S}_{cd}$ of $\mathcal{S}_G$ reachable from the state $s_1'$ (reached 
on executing $r_1' \in L$), such that in $\mathcal{S}_{cd}$ it is not possible for $d_i$ from every pair of 
transitions in $\Gamma_{node}$ to execute prior to $c_i$.
In other words, there always exists
one pair of transitions $c_i$ and $d_i$ such that $d_i$ cannot be executed prior
to $c_i$ in $\mathcal{S}_{cd}$ even when for all the other pairs $(c_j,d_j)$, $j \neq i$, 
$d_j$ executes prior to $c_j$ in $\mathcal{S}_{cd}$. Attempting to reorder $(c_i,d_i)$
within $\mathcal{S}_{cd}$ will alter the order between another pair of transitions 
in $\Gamma_{node}$, thus making the resultant transition sequence \emph{non} dependence-covering
w.r.t. $w$.

\item In the subspace $\mathcal{S}_{cd}$, the transitions $c_1$ and $d_{\chi}$ satisfy one of the following
criteria --- (a)~$c_1$ and $d_{\chi}$ are transitions of two different event 
handlers on the same thread such that $\eventof(\taskof(d_{\chi}))$ is blocked 
when $\eventof(\taskof(c_1))$ is executable, or (b)~a transition prior to $d_{\chi}$ in the task of 
$d_{\chi}$ is enabled by a transition in the event handler blocked on $\threadof(c_1)$
when $c_1$ is the next transition on that thread, or (c)~$c_1$ needs to be executed 
to enable a transition $q$ blocked in an event handler such that either $\eventof(\taskof(d_{\chi}))$
is blocked on $\threadof(q)$ when $q$ is the next transition on that thread, or
a transition prior to $d_{\chi}$ in the task of $d_{\chi}$ is enabled by a transition in the event 
handler blocked on $\threadof(q)$ when $q$ is the next transition on that thread.

\item Let $e_c$ and $e_d$ respectively be the executable event related to $c_1$
and the blocked event related to $d_{\chi}$ (as identified by N6 above) in the 
subspace $\mathcal{S}_{cd}$. The events $e_c$ and $e_d$ posted 
to the same thread are such that, either $e_c$ is not posted in the transition sequence
$S.w$ whereas $e_d$ is posted in $w$ or $e_d$ is posted prior to $e_c$ in $w$. 

\item The parent node of $\Gamma_{node}$ in \gtree\ is a node $\Gamma_{par}$ which contains 
$(\post{($e_c$)},\post{($e_d$)})$
as a transition pair. Reordering \post{($e_c$)} and \post{($e_d$)} results in a state 
space where every transition $d_i$ in the transition pairs of the $\Gamma_{node}$
can be executed prior to corresponding $c_i$ thus making them consistent w.r.t. 
ordering observed in $w$. We refer to $\Gamma_{node}$ as the \emph{child of $\Gamma_{par}$
obtained on exploring {\upshape\post{($e_c$)}} prior to {\upshape\post{($e_d$)}}}.

\item $\Gamma_{node}$ has the same number of child nodes as the number of 
pairs of \texttt{post}s in  $\Gamma_{node}$.

\item Every path in the tree from the root to a node containing atleast one pair of dependent
transitions encodes a \emph{non} dependence-covering sequence of $w$ from state $s$, 
say $v_k$, which identifies one pair of dependent transitions which either are ordered 
differently compared to their ordering in $w$ or form a new incoming dependence into 
a transition in $w$.
\end{enumerate}
\end{definition}

We can systematically construct certain interesting \emph{non} dependence-covering transition sequences
of $w$ using \gtree\ paths, each of which have a transition $b$
belonging to $w$ whose $DG$ over the constructed transition sequence does not 
match $DG(b,w,s)$.

\begin{construction}
\label{const:gtree-seq}
\upshape
A transition sequence $v_k$ is constructed using a path of \gtree\ by performing
steps I, II and III below. The order between those transitions in $v_k$ which 
are not explicitly specified by the step~II below can be arbitrary but valid w.r.t. the 
orders fixed for transitions reasoned in step~II and consistent 
w.r.t. dependence graph over $w$.\\
\underline{Step I.}~Start from the root of the \gtree. \\
\underline{Step II.}~At each node $\Gamma_{node} = \{(c_1,d_1),(c_2,d_2) \ldots (c_{\chi},d_{\chi})\}$, pick a pair of transitions
$(c_i,d_i)$ such that $c_i$ will be executed prior to $d_i$ in the sequence $v_k$ 
while the order of other transition pairs (if can be executed in $v_k$) are consistent
w.r.t. $w$ \ie\ $d_j$ is executed prior to $c_j$ for $j \neq i$. \\
\underline{Step III.}~If the pair $(c_i,d_i)$ selected are non-\post{} dependent 
transitions then the construction of $v_k$ is complete, else move to the child 
obtained on exploring $c_i$ prior to $d_i$ and repeat Step~II. 
\end{construction}

\paragraph{Observations for Construction~\ref{const:gtree-seq}.}
From the step~III of Construction~\ref{const:gtree-seq} we note that the process 
of constructing a transition sequence of interest can stop even at an intermediate node.
A transition sequence $v_k$ identified by this construction has one transition $d_i \in w$ (corresponding
to the last pair of transitions selected from a \gtree\ node) which has dependence
with a prior executed transition $c_i$ such that either $c_i \not\in w$ or
$d_i$ is executed prior to $c_i$ in $w$. Hence, this construction cleanly identifies
a pair of transitions in $v_k$ which need to be reordered so as to eventually
obtain a dependence-covering sequence of $w$.

\begin{lemma}
\label{lmm:gtree-exists}
\upshape
A \gtree\ defined by Definition~\ref{def:gtree} can be constructed in the state space
$\mathcal{S}_G$. 
\end{lemma}

\begin{proof}
We prove this by giving a sketch for constructing \gtree\ starting with the transition
sequence $z = r_1'.v$ constructed as per Construction~\ref{construct:cased}. 

Figure~\ref{fig:sequence-space} pictorially represents the \gtree.
The variables $a$ and $b$ in the root stand respectively
for $r_1'$ and $r_{\mu}'$. We annotate the only child of the root node as $\Gamma_0$. 
Except the root node, every other node $\Gamma_{[k_1\ldots k_j]}$
contain transition pairs of the form $(a_{[k_1\ldots k_j]i},b_{[k_1\ldots k_j]i})$,
such that exploring the transition $a_{[k_1\ldots k_{j-1}]k_j}$ prior
to $b_{[k_1\ldots k_{j-1}]k_j}$ in a transition sequence results in the discovery
of the subtree rooted at $\Gamma_{[k_1\ldots k_j]}$. This makes 
$(a_{[k_1\ldots k_{j-1}]k_j},b_{[k_1\ldots k_{j-1}]k_j})$ the transition pair
corresponding to $\Gamma_{[k_1\ldots k_j]}$ in its parent node. The last pair of
transitions in the set corresponding to $\Gamma_{[k_1\ldots k_j]}$ is identified
as $(a_{[k_1\ldots k_j](ab)_{k_1\ldots k_j}},b_{[k_1\ldots k_j](ab)_{k_1\ldots k_j}})$, where 
$(ab)_{k_1\ldots k_j}$ symbolically denotes the count of the number of transition
pairs in the node $\Gamma_{[k_1\ldots k_j]}$.

\begin{figure*}
\centering
\includegraphics[scale=0.42]{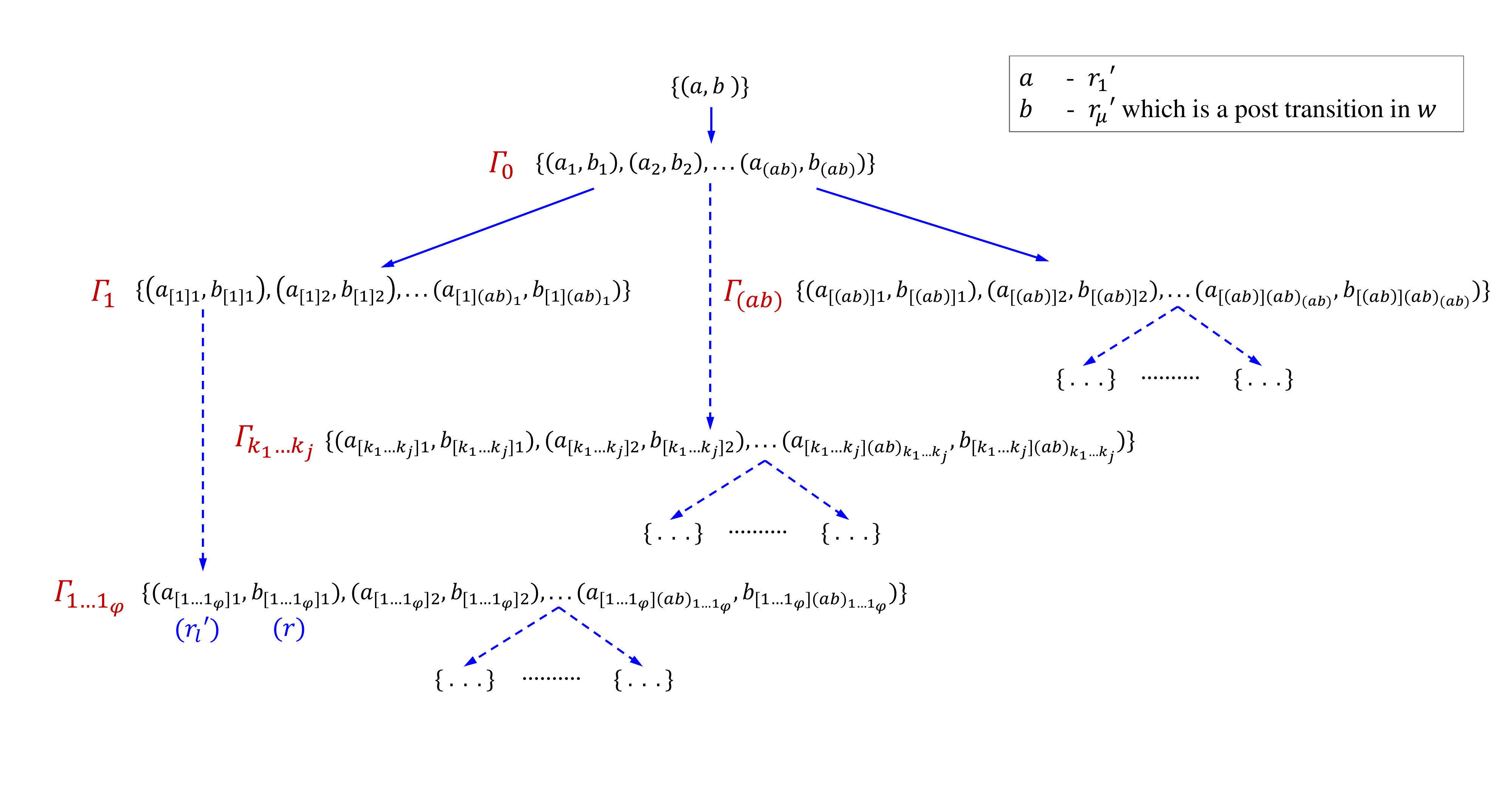}
\vspace{-1.3cm}
\caption{Tree encoding (\gtree) of the set $\Gamma$ of transition sequences 
used in identifying a \post{} transition $r_{\mu}' \in w$ executed in $v$, 
which is to be reordered with $r_1'$.}
\label{fig:sequence-space}
\end{figure*}

From the constraint M3 of Construction~\ref{construct:cased}, $r_l'$ 
is in the set $\mathit{enabledFuture}(r_1')$. Based on this and the 
Definition~\ref{def:enabledfuture} we can identify 
a chain of \post{} transitions related to $r_l'$, which is a subsequence of $z = r_1'.v$
identified henceforth as $p_{\eta} \ldots p_2.p_1$. The chain of \texttt{post}s
$p_{\eta} \ldots p_2.p_1$ is such that
(i)~$p_{\eta} = r_1'$, (ii)~for any $i \in [1,\eta-1]$, $p_i$ is either in the handler
of the event posted by $p_{i+1}$ or $p_i$ has been added to enable 
a blocked transition in the handler of the event posted by $p_{i+1}$, and 
(iii)~the transition $r_l'$ is either in the handler of the event posted by $p_1$
or $r_l'$ is a transition added to enable a blocked transition
in the handler of the event posted by $p_1$. 
We encode each of the transitions in the chain $p_{\eta} \ldots p_2.p_1$ and $r_l'$
as the first transition in the ordered pair of transitions belonging to different 
\gtree\ nodes in the leftmost branch of the \gtree\ in Figure~\ref{fig:sequence-space}.
Through this construction sketch we will reason that for any 
$i \in [1,\eta-1]$, $p_i$ is the first transition in a pair belonging to a node 
whose parent node contains the transition $p_{i+1}$ in a member pair.
Similarly, $r_l'$ is a transition in a pair belonging to a node whose parent node 
contains the transition $p_1$ in a member pair. 

In case of $r_l'$, $r$ is the transition occupying the second position in the 
pair corresponding to $r_l'$. In Figure~\ref{fig:sequence-space}, the node 
annotated $\Gamma_{1\ldots 1_{\varphi}}$ identifies the \gtree\ node containing 
the pair $(r_l',r)$. In this node we use $(a_{[1 \ldots 1_{\varphi}]1}, b_{[1 \ldots 1_{\varphi}]1})$ to denote 
$(r_l',r)$, where $\varphi$ denotes the count of ``1''s in the first part of the subscript. Indeed
we will reason that $\varphi = \eta-1$ and for $\vartheta \in [1,\eta-1]$, 
$a_{[1 \ldots 1_{\vartheta-1}] 1} = p_{\eta - \vartheta}$. 

Let us now see how to 
identify the rest of the transition pairs in the node $\Gamma_{1 \ldots 1_{\varphi}}$.
Assume $\threadof(r_l') \neq \threadof(r)$. Let $P$ be a set of transitions 
such that a transition $p \in P$ if $index(S.z,p) \to_{S.z} \taskof(r)$
and $index(S.z,p) > index(S.z,r_l')$. Clearly, all the transitions in $P$ were executed prior to $r$
in $w$ as well (by M1 in Construction~\ref{construct:cased}). Now in order to execute
$r$ prior to $r_l'$, all the transitions in the set $P$ also need to execute prior
to $r_l'$. 
However, attempting to explore transitions in set $P$ from the state $s_{l-1}'$
(from where $r_l'$ is executed in sequence $z$) will 
result in one of the three scenarios listed in Lemma~\ref{lmm:reorder-impossible}.
Concretely, this happens because of one of the following reasons.
\begin{enumerate}[label=S\arabic*.]
\item A non-\post{} transition $c$ belonging to $P$ or belonging to a task $h$
on whose thread the task of some transition in $P$ is blocked, gets shifted 
prior to $r_l'$ even though $c$ has dependence with a transition $d \in w$
such that $d$ can be executed only after $r_l'$ and $d$ is executed prior to $c$
in $z$, or
\item A \post{} transition $c$ belonging to $P$ or belonging to a task $h$
on whose thread the task of some transition in $P$ is blocked, gets shifted 
prior to $r_l'$ and gets reordered w.r.t. a \post{} transition $d \in w$
such that $d$ can be executed only after $r_l'$ and $d$ is executed prior to $c$
in $z$. The reordering of \texttt{post}s $c$ and $d$ in turn breaks the $DG$ of a 
transition $d' \in w$ making it non-identical to $DG(d',w,s)$ either by reordering
it w.r.t. a dependent transition $c' \in w$ or by exploring a dependent transition
$c' \not\in w$ prior to $d'$.
\end{enumerate}

The transition pair $(c,d)$ identified above is the transition pair
$(a_{[1 \ldots 1_{\varphi}]2}, b_{[1 \ldots 1_{\varphi}]2})$ in the set corresponding
to the node $\Gamma_{1\ldots 1_{\varphi}}$. The transition pair 
$(a_{[1 \ldots 1_{\varphi}]3}, b_{[1 \ldots 1_{\varphi}]3})$ can be identified 
by attempting to explore $b_{[1 \ldots 1_{\varphi}]2}$ prior to $a_{[1 \ldots 1_{\varphi}]2}$
by using a strategy similar to the one devised to reorder $b_{[1 \ldots 1_{\varphi}]1}$ and $a_{[1 \ldots 1_{\varphi}]1}$.
However when doing so we also need to try to explore $b_{[1 \ldots 1_{\varphi}]1}$
prior to $a_{[1 \ldots 1_{\varphi}]1}$, else we will obtain a transition sequence
similar to $z$ which is already established to be a \emph{non} dependence-covering
sequence of $w$. Similarly, attempting to reorder the recently identified pair of 
transitions while keeping the order between prior identified transition pairs in $\Gamma_{1\ldots 1_{\varphi}}$
consistent w.r.t. dependence graph of $w$, aids in identifying newer transition 
pairs to be added to $\Gamma_{1\ldots 1_{\varphi}}$. However, due to Lemma~\ref{lmm:reorder-impossible}
and the state space being finite and acyclic, 
we will soon run out of transition pairs which can be added this way. Indeed,
attempting to reorder the last transition pair consisting of 
$a_{[1 \ldots 1_{\varphi}](ab)_{1 \ldots 1_{\varphi}}}$ and $b_{[1 \ldots 1_{\varphi}](ab)_{1 \ldots 1_{\varphi}}}$
using the above technique will result in exploring a transition sequence 
where $r_l' = a_{[1 \ldots 1_{\varphi}]1}$ gets explored
prior to $r = b_{[1 \ldots 1_{\varphi}]1}$, thus re-identifying an already added 
transition pair. However, in the sequence $w$ each transition 
$b_{[1 \ldots 1_{\varphi}]i}$ was explored prior to $a_{[1 \ldots 1_{\varphi}]i}$ 
or $a_{[1 \ldots 1_{\varphi}]i}$ was not even explored.
This indicates that the ordering between dependent transitions can be made 
consistent w.r.t. $w$ by reordering a pair of event handlers
related to the transition pairs in $\Gamma_{1\ldots 1_{\varphi}}$. 
Indeed we can 
show that this can be achieved by reordering handler of the event posted by 
$p_1$ (belonging to the chain $p_{\eta} \ldots p_2.p_1$) with the handler 
in which $b_{[1 \ldots 1_{\varphi}](ab)_{1 \ldots 1_{\varphi}}}$ is executed or 
a handler that enables a transition prior to $b_{[1 \ldots 1_{\varphi}](ab)_{1 \ldots 1_{\varphi}}}$ in 
$\taskof(b_{[1 \ldots 1_{\varphi}](ab)_{1 \ldots 1_{\varphi}}})$.
Let $e$ be the event corresponding to the latter handler. Clearly, $\post{($e$)} \in w$. 
We will be able to show that $p_1$ posts to the same event queue as $e$, thus satisfying
property N6 of Construction~\ref{const:gtree-seq}. This will result
in adding the pair of transitions $(p_1,\post{($e$)})$ into the parent node of 
$\Gamma_{1\ldots 1_{\varphi}}$ thus satisfying the property N8 of a \gtree.
The parent node $\Gamma_{1\ldots 1_{\varphi-1}}$ can be populated similarly
starting with the reordering of $a_{[1 \ldots 1_{\varphi-1}]1} = p_1$ and 
$b_{[1 \ldots 1_{\varphi-1}]1} = \post{($e$)}$, and so on eventually identifying
the parent node of $\Gamma_{1\ldots 1_{\varphi-1}}$ which will contain $(p_2,\_)$
as a member. 

We note that if $\threadof(r_l') = \threadof(r)$ then $\Gamma_{1\ldots 1_{\varphi}}$
would be a singleton set consisting only the pair $(r_l',r)$, and the parent node
would be identified as the node with transition pair $(p_1,\post{($\eventof(r)$)})$.
In general, for any pair of transitions on the same thread but different handlers we identify
the parent node as the node with the pair of transitions posting the events corresponding
to these handlers, as a member.
 
If $c$ and $d$ identified by the scenario S2 introduced earlier, are \post{} transitions,
then the dependent transition pair $(c',d')$ identified by this scenario becomes a transition pair in 
one of the nodes in the subtree that can be generated by exploring $c$
prior to $d$ in a transition sequence, say $v_k$, from state $s_1'$. Transition
pairs belonging to the nodes of this subtree can be systematically identified by 
attempting to reorder $c'$ explored in $v_k$ w.r.t. $d'$ by doing as described 
in the context of reordering transitions $a_{[1 \ldots 1_{\varphi}]i}$ and $b_{[1 \ldots 1_{\varphi}]i}$
belonging to the node $\Gamma_{1\ldots 1_{\varphi}}$.
By only reordering transitions in the subspace obtained when the \post{} transition $c$
is explored prior to \post{} transition $d$, will end up breaking the $DG$ of some 
transition in $w$ 
thus resulting in \emph{non} dependence-covering
sequences of $w$. If this is not the case then it implies that the transition pairs 
in $\Gamma_{1\ldots 1_{\varphi}}$ added prior to $(c,d)$ can be explored in a manner
consistent w.r.t. $w$ which makes such a transition sequence dependence-covering
w.r.t. $w$, or contain more transitions from $w$ than in $z$ thus breaking the 
constraint M4 of Construction~\ref{construct:cased}.
\end{proof}

For each pair of transitions in the nodes of \gtree\ and the entire node itself, 
we assign a level called  \gidx\ defined as below.
\begin{definition}
\label{def:gidx}
\upshape
The \gidx{} of a tree node, say $\Gamma_i$, referred as \gidx{($\Gamma_i$)}
is assigned a level same as the \gidx{} of a transition pair in $\Gamma_i$
which has the highest \gidx{} among all the transition pairs in $\Gamma_i$.
The \gidx{} of a pair $(c,d)$ of dependent
transitions is considered to be $0$ and referred as \gidx{($(c,d)$)}. 
Let $(c,d)$ be a pair of \post{} transitions in a \gtree\ node, say $\Gamma_{par}$, such that 
$\Gamma_{kid}$ be the child node of $\Gamma_{par}$ discovered on executing
$c$ prior to $d$ in a transition sequence from state $s_1'$. Then,
\gidx{($(c,d)$)} $=$ $\gidx{($\Gamma_{kid}$)} + 1$. 
\end{definition}

\paragraph{Note.} 
\gtree\ essentially identifies all the pairs of transitions (dependent or posting to
the same event queue) which will have to be systematically identified for reordering by 
\emdpor\ (by invoking \texttt{FindTarget}) in order to discover a dependence-covering 
sequence of $w$, assuming that \emdpor\ initially explored only the members of $L$ from state $s$. 
%
Also during the process, \emdpor\ needs to adequately set up data structures such as $backtrack$, $done$ and 
$RP$ sets at the explored states in the subspace reachable from $s_1'$ so as to eventually 
invoke \texttt{FindTarget($r_1'.v'$,$r_1'$,$r_{\mu}'$)} where $r_1'.v'$ is a 
transition sequence constructed from \gtree\ using Construction~\ref{const:gtree-seq}. 
After establishing this we can use the arguments used to prove lemmas
corresponding to cases B and C (Lemma~\ref{lmm:caseb} and \ref{lmm:casec}) to show that some transition
executed in $w$ prior to $r_{\mu}'$ or $r_{\mu}'$ itself gets added to the backtracking set $L$ at state $s$,
thus contradicting the property assumed for $L$ as per the case~D.
Note that all the transition sequences
which can be constructed by running the Construction~\ref{const:gtree-seq} on
the \gtree\ in Figure~\ref{fig:sequence-space}, have 
$r_1'$ as their first transition. Hence by induction hypothesis H1, \emdpor\ 
explores dependence-covering sequences of all these transition sequences. The 
challenge however is to show that \texttt{FindTarget} gets invoked to reorder $r_1'$ and 
$r_{\mu}'$. We achieve this by proving the following property by inducting on
the \gidx{} levels of the \gtree\ nodes.

\begin{lemma}
\label{lmm:gtree-emdpor-property}
\upshape
For each node $\Gamma_{node} = \{(c_1,d_1),(c_2,d_2) \ldots (c_k,d_k)\}$ in the \gtree\ generated from 
the sequence $z$, \emdpor\ explores a sequence $r_1'.u$ whose prefix reaches a state  
$s'$ such that the following properties hold.
\begin{enumerate}[label=P\arabic*.]
 \item For every transition pair $(c_i,d_i)$ in $\Gamma_{node}$, for $i \in [1,k]$, 
 $c_i = next(s',\threadof(c_i))$ and either $\threadof(c_i) \in done(s')$
 or $\threadof(c_i)$ is not enabled in $s'$. 
 \item There exists a pair $(c_i,d_i)$ in $\Gamma_{node}$, for $i \in [1,k]$,
 such that $c_i$ is executed at $s'$ and \texttt{FindTarget($r_1'.u$,$c_i$,$d_i$)} 
 is invoked such that for $j \in [i+1,k]$, $index(r_1'.u,d_j) < index(r_1'.u,c_j)$. 
 \item For each transition pair $(c_i,d_i)$ in $\Gamma_{node}$, for $i \in [1,k]$, where $c_i$ and $d_i$ are \post{}
 transitions, either $(c_i,d_i) \in RP(s')$ or \texttt{FindTarget} has been invoked 
 to reorder $c_i$ executed at $s'$ and the later executed transition $d_i$. 
\end{enumerate}
\end{lemma}

\begin{proof}
We prove the above property by inducting on the \gidx{} level of nodes.

\paragraph{Base case (\gidx{} $=$ $0$).} We present an outline on how to reason about this
case. Only leaf nodes of \gtree\ belong to \gidx{} 
level $0$. Let $\Gamma_{node} = \{(c_1,d_1),(c_2,d_2) \ldots (c_k,d_k)\}$ be a leaf node.
Let $v_i$ be a transition sequence constructed over \gtree\ starting from the root 
and ending with a suffix where $c_i$ is explored prior to dependent transition $d_i \in w$.
Let (dependence-covering sequence of) $v_i$ be the first transition sequence related to  $\Gamma_{node}$ to be explored
by \emdpor. Exploration of $v_i$ by \emdpor\ is guaranteed due to induction 
hypothesis H1, since the first transition of $v_i$ is $r_1'$ which is a transition in 
the set $L$.
We will then have to show that $c_i$ and $d_i$ will be identified as racing transitions
by the Algorithm~\texttt{Explore} leading to invocation of \texttt{FindTarget($S.v_i$,$d_i$,$c_i$)}.
Exploring backtracking choices thus added results in exploring $c_{i+1}$ prior to
$d_{i+1}$. Again these will be identified as racing transitions and so on. Ultimately,
\texttt{FindTarget($S.v_{i-1}$,$c_{i-1}$,$d_{i-1}$)} gets invoked when threads of all 
the other $c_j$ transitions are either in $done$ set at the state, say $s'$, from where 
$c_{i-1}$ is executed or disabled in $s'$. This proves property P1 and P2. Property P3 is not relevant for 
the base case because a leaf node does not contain any pair of \post{} transitions.

\paragraph{Induction hypothesis.} For a node $\Gamma_{node}$ such that \gidx{($\Gamma_{node}$)}
$=$ $\theta$, the properties P1, P2 and P3 hold. 

\paragraph{Induction step (\gidx{} $=$ $\theta+1$).} 
Let $\Gamma_{node} = \{(c_1,d_1),(c_2,d_2) \ldots (c_k,d_k)\}$ be a node in \gtree\
such that \gidx{($\Gamma_{node}$)} $=$ $\theta + 1$. This indicates that the 
highest \gidx{} of any pair of transitions in $\Gamma_{node}$ is $\theta + 1$.
Let $(c_i,d_i) \in \Gamma_{node}$ be a pair of \post{} transitions with its 
corresponding child node being $\Gamma_{i} = \{(c_{[i]1},d_{[i]1}),(c_{[i]2},d_{[i]2}) 
\ldots (c_{[i](cd)_i},d_{[i](cd)_i})\}$. Then by our assumption on the \gidx{}
of $(c_i,d_i)$ and the definition of \gidx, we can establish that \gidx{($\Gamma_i$)}
can be atmost $\theta$. Then, with the help of induction hypothesis we can show 
that \texttt{FindTarget} gets invoked to reorder a transition $c_{[i]j}$ and later
executed transition $d_{[i]j}$ with suitable constraints over $done$ and $RP$
sets (as established by P1 and P3),
resulting in the invocation of \texttt{ReschedulePending} by the Step~3 of \texttt{FindTarget}
(see line~\linerecursivefindtarget). This in turn invokes \texttt{FindTarget} to reorder
the \post{} transitions $c_i$ and $d_i$. 
Note that it is important for the $RP$ set to be adequately set up since the HB relation
computed by \emdpor\ adds edges based on $RP$ set as well (see Definition~\ref{def:happens-before}).
After backtracking choices are 
computed at the state from where $c_i$ is executed, $(c_i,d_i)$ get added to the $RP$
set at that state. On eventually reordering $c_i$ and $d_i$, $c_{i+1}$ and $d_{i+1}$
get reordered. Based on whether these are \post{} or non-\post{} transitions we
can apply suitable reasoning to show how $done$ and $RP$ sets get populated.
Ultimately, we can show that \texttt{FindTarget} gets invoked to reorder the 
last pair of transitions belonging to the set $\Gamma_{node}$ with the corresponding
$done$ and $RP$ sets appropriately set up as required. 
\end{proof}

\begin{lemma}
\label{lmm:cased}
\upshape
\emdpor\ explores a dependence-covering sequence for $w$ from state $s$ when set 
$L$ satisfies case D.
\end{lemma}

\begin{proof}
From Lemma~\ref{lmm:gtree-exists}, \gtree\ exists and a \gtree\ can be generated
using the transition sequence $z$ constructed as per Construction~\ref{construct:cased}
which too has been proven to exist (see Lemma~\ref{lmm:witness-cased}).
From the definition of a \gtree\ the transition pair $(r_1',r_{\mu}')$ is the
only element of the singleton set at the root of the \gtree\ generated from $z$. 
Then, from the property P2 established by Lemma~\ref{lmm:gtree-emdpor-property}, 
\texttt{FindTarget($S.r_1'.u,r_1',r_{\mu}'$)} is invoked for some transition 
sequence $r_1'.u$ which is a dependence-covering sequence of a sequence constructed 
using Construction~\ref{const:gtree-seq}. Since $r_{\mu}' \in w$ we can use the 
arguments used to prove Lemma~\ref{lmm:caseb} to show that some transition
executed in $w$ prior to $r_{\mu}'$ or $r_{\mu}'$ itself gets added to the backtracking 
set $L$ at state $s$, thus contradicting the property assumed for $L$ as per the case~D.
This in turn proves the existence of a dependence-covering sequence of $w$.
\end{proof}

\subsubsection{Case E}
\label{sec:caseE}
Case E assumes a subset of transitions in the backtracking set $L$ to be present
in $w$ such that the first transition in $w$ from the set $L$ is of type \post.
This case additionally assumes that if there are non-\post{} transitions in the 
set $L$ then all such transitions are dependent with some transition in $w$.
This is because if there exists a non-\post{} transition independent w.r.t. all the transitions 
in $w$ then this case becomes equivalent to case~A which has already been shown 
to result in a dependence-covering sequence of $w$ (see Lemma~\ref{lmm:casea}).
Note that we had considered a variant of case~E in case~C where we had assumed 
the first transition in $w$ from the set $L$ to be a non-\post{} transition.
However, a \post{} transition being the first transition in $w$ from $L$
makes the reasoning of this case very similar to that used to establish contradiction
to the property of $L$ in case~D.


\begin{construction}
\label{construct:casee}
\upshape
Let $z: s \xrightarrow{r_1'} s_1' \xrightarrow{r_2'} 
s_2' \ldots \xrightarrow{r_m'} s_m'$ in $\mathcal{S}_G$, 
where $r_1'$ is the first transition in $w$ from the set $L$ and $v = r_2'\ldots r_m'$, 
be a sequence satisfying the following constraints:
\begin{enumerate}[label=M\arabic*.]
 \item Sequence $v$ consists of transitions belonging to $w$ as well as transitions
 outside $w$. For a transition $r_i' \in v$, if $r_i' \in w$ and $r_i' \not\in \mathit{enabledFuture}(r_1')$ 
 then $identicalDG(r_i',w,z,s)$. For a transition $r_i' \in v$, if $r_i' \in w$ 
 and $r_i' \in \mathit{enabledFuture}(r_1')$  then either $identicalDG(r_i',w,z,s)$
 or $r_i'$ is the first transition of a shortest sequence from the state $s_{i-1}'$ 
 comprising only of transitions which do not belong to $w$ or belong to $\mathit{enabledFuture}(r_1')$, 
 to be executed to make an event $e$ blocked in $s_{i-1}'$ executable 
 such that $e$ was dequeued in $w$. 
 Finally, for a transition $r_i' \in v$, if $r_i' \not\in w$ 
 then $r_i'$ is the first transition of a shortest sequence from the state $s_{i-1}'$ 
 comprising only of transitions which do not belong to $w$ or belong to $\mathit{enabledFuture}(r_1')$, 
 to be executed to make an event $e$ blocked in $s_{i-1}'$ executable 
 such that $e$ was dequeued in $w$. 
  
 \item There exists no extension to any prefix $\alpha$ of $z$ which results in 
 a transition sequence $\alpha.\gamma$ such that there exists a pair of dependent transitions $c$ 
 and $d$ with the following properties:
 \begin{enumerate}
  \item (i)~Either $c$ and $d$ are transitions in $w$ such that $c \not\in \mathit{enabledFuture}(r_1')$,
  $index(w,d) < index(w,c)$ but $index(\gamma,c) < index(\gamma,d)$, 
  or (ii)~$c \not\in w$, $c \not\in \mathit{enabledFuture}(r_1')$, $c$ is a transition 
  in the $\mathit{enabledFuture}$ set of a \post{} in $w$, and
  $d \in w$ such that $c$ is executed prior to $d$ in $\alpha.\gamma$, and 
  
  \item Attempting to reorder $c$ and $d$ through some other extension to $\alpha$ will 
  only result in a transition sequence $\alpha.\gamma'$ which breaks the order between another pair 
  of dependent transitions $c'$ and $d'$ such that either (i)~$c',d' \in w$,
  $c' \not\in \mathit{enabledFuture}(r_1')$, $index(w,d') < index(w,c')$ but $index(\gamma',c) < index(\gamma',d)$, or 
  (ii)~$d' \in w$, $c' \not\in \mathit{enabledFuture}(r_1')$, $c'$ is a transition 
  in the $\mathit{enabledFuture}$ set of a \post{} in $w$,
  and $c'$ is executed prior to $d'$ in $\alpha.\gamma'$.
 \end{enumerate}

 \item There exists a transition $r \in nextTrans(s_m')$ such that $r$ is a transition
 in $w$, $r \not\in \mathit{enabledFuture}(r_1')$ and $r$ is dependent with a transition $r_l'$ 
 executed in $v$ such that (i)~$r_l' \in \mathit{enabledFuture}(r_1')$
 and (ii)~if $r_l' \in w$ then $r$ is executed prior to $r_l'$ in $w$.
 
 \item $z$ is a sequence with maximum transitions from $w$ while satisfying the 
 constraints M1, M2 and M3.
\end{enumerate}
\end{construction}

We now present the lemma which establishes that a dependence-covering sequence
of $w$ from state $s$ gets explored even when the set $L$ satisfies the 
property stated in case~E. The proof sketch of this lemma is similar 
to that outlined for Lemma~\ref{lmm:cased} which reasons about the case~D.
However, the proof for this case will use a transition sequence $z$ constructed
as per Construction~\ref{construct:casee} to generate the \gtree. 
\begin{lemma}
\label{lmm:casee}
\upshape
EM-DPOR explores a dependence-covering sequence for $w$ from state $s$ when set 
$L$ satisfies case E.
\end{lemma}
%

\subsection{Main Result}
\label{sec:em-dpor-proof}

\begin{theorem}
\label{thm:emdpor-correctness}
\upshape
In a finite and acyclic state space $\mathcal{S}_G$, whenever \texttt{Explore} (Algorithm~\algoexplore) 
backtracks from a state $s$ to a state prior to $s$ in the search stack, EM-DPOR has explored
a dependence-covering sequence for any sequence $w$ in $\mathcal{S}_G$ from $s$, \ie\ the set of 
transitions explored from a state $s$ is a dependence-covering set in $s$.
\end{theorem}
\begin{proof}
The proof for this theorem is by induction on the order in which states visited by EM-DPOR
are backtracked, as explained in the proof strategy in Section~\ref{sec:proof-strategy}.

\paragraph{Base case.} The first backtracked state is a state with no 
transitions enabled. Such a state is reached as Algorithm~\texttt{Explore} performs
a depth first search on the state space of $\mathcal{S}_G$ which is finite and acyclic.
The induction hypothesis H1 vacuously holds for such a state with no outgoing transitions.

\paragraph{Induction hypothesis (same as induction hypothesis H1 in Section~\ref{sec:proof-strategy}).} 
Let $S$ be a sequence from $s_{init} \in \mathcal{S}_G$ reaching state $s$, explored by Algorithm \texttt{Explore} of \emdpor. 
Let $L$ be the set of transitions explored by \emdpor\ from the state $s$.
Then, for every transition sequence from a state reached on each recursive call \texttt{Explore($S.r$)}, for all $r \in L$, the algorithm 
explores a corresponding dependence-covering sequence.

\paragraph{Induction step.}
Lemmas~\ref{lmm:casea}, \ref{lmm:caseb}, \ref{lmm:casec}, \ref{lmm:cased} and \ref{lmm:casee}
prove the induction step for the exhaustive cases based on the contents of the set
$L$, introduced in Section~\ref{sec:proof-strategy}.

Thus EM-DPOR explores a dependence-covering sequence for any sequence $w$ in $\mathcal{S}_G$
from a state $s$ reached on \texttt{Explore($S$)}, which in turn establishes that the set of transitions $L$ explored from $s$
is a dependence-covering set as per Definition~\deflazypersistentset.
\end{proof}
\begin{figure*}[t]
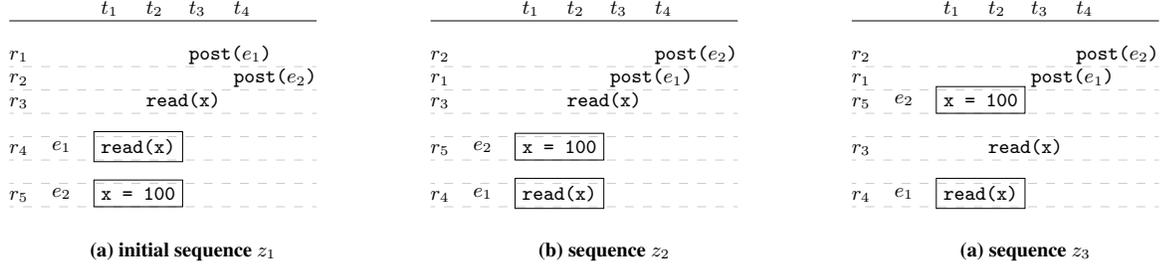

\begin{tabular*}{\textwidth}{l@{\;\;}c@{\;\;}r}
\begin{minipage}{\dimexpr0.33\textwidth-2\tabcolsep}
\centering
\scalebox{0.85}{
\begin{tabular}{@{}r@{\;\;\;\;\;\;\;\;\;\;\;\;\;}l@{\;}l@{\;}l@{\;}l@{}}
 & \tn{t1}{$t_1$}  & \tn{t2}{\;\;\;\;$t_2$} & \tn{t3}{$t_3$} & \tn{t4}{\;\;\;\;$t_4$} \\
\hline
\\
$r_1$ &&& \multicolumn{2}{@{}l}{\tn{n1}{\post{($e_1$)}}} \\\myhline
$r_2$ &&&& \tn{n2}{\;\;\;\;\post{($e_2$)}} \\\myhline
$r_3$ && \multicolumn{3}{@{}l}{\tn{n3}{\;\;\;\;\readX{(x)}}} \\\myhline
\\\myhline
$r_4$ &  \multicolumn{2}{@{}l}{\tn{n4}{\readX{(x)}}} & \\\myhline
\\\myhline
$r_5$ & \multicolumn{2}{@{}l}{\tn{n5}{\texttt{x = 100}}} & \\\myhline
\\
\end{tabular}
\tikz[remember picture,overlay] \node[fit=(n4), draw] {};
\tikz[remember picture,overlay] \node[fit=(n5), draw] {};
\tikz[remember picture,overlay] \node[left of=n4,xshift=-.2cm] (e1) {$e_1$};
\tikz[remember picture,overlay] \node[left of=n5,xshift=-.2cm] (e2) {$e_2$};
\tikz[remember picture,overlay] \node (dummycaptionA) at ($(t1)!0.5!(t4)$) {};
\tikz[remember picture,overlay] \node[below of=dummycaptionA,xshift=.2cm,yshift=-2.8cm] 
(captionA) {\textbf{(a) initial sequence $z_1$}};
}
\end{minipage}%
&
\begin{minipage}{\dimexpr0.33\textwidth-2\tabcolsep}
\centering
\scalebox{0.85}{
\begin{tabular}{@{}r@{\;\;\;\;\;\;\;\;\;\;\;\;\;}l@{\;}l@{\;}l@{\;}l@{}}
 & \tn{t1}{$t_1$}  & \tn{t2}{\;\;\;\;$t_2$} & \tn{t3}{$t_3$} & \tn{t4}{\;\;\;\;$t_4$} \\
\hline
\\
$r_2$ &&&& \tn{n2}{\;\;\;\;\post{($e_2$)}} \\\myhline
$r_1$ &&& \multicolumn{2}{@{}l}{\tn{n1}{\post{($e_1$)}}} \\\myhline
$r_3$ && \multicolumn{3}{@{}l}{\tn{n3}{\;\;\;\;\readX{(x)}}} \\\myhline
\\\myhline
$r_5$ & \multicolumn{2}{@{}l}{\tn{n5}{\texttt{x = 100}}} & \\\myhline
\\\myhline
$r_4$ &  \multicolumn{2}{@{}l}{\tn{n4}{\readX{(x)}}} & \\\myhline
\\
\end{tabular}
\tikz[remember picture,overlay] \node[fit=(n4), draw] {};
\tikz[remember picture,overlay] \node[fit=(n5), draw] {};
\tikz[remember picture,overlay] \node[left of=n4,xshift=-.2cm] (e1) {$e_1$};
\tikz[remember picture,overlay] \node[left of=n5,xshift=-.2cm] (e2) {$e_2$};
\tikz[remember picture,overlay] \node (dummycaptionA) at ($(t1)!0.5!(t4)$) {};
\tikz[remember picture,overlay] \node[below of=dummycaptionA,xshift=.2cm,yshift=-2.8cm] 
(captionA) {\textbf{(b) sequence $z_2$}};
}
\end{minipage}
&
\begin{minipage}{\dimexpr0.33\textwidth-2\tabcolsep}
\centering
\scalebox{0.85}{
\begin{tabular}{@{}r@{\;\;\;\;\;\;\;\;\;\;\;\;\;}l@{\;}l@{\;}l@{\;}l@{}}
 & \tn{t1}{$t_1$}  & \tn{t2}{\;\;\;\;$t_2$} & \tn{t3}{$t_3$} & \tn{t4}{\;\;\;\;$t_4$} \\
\hline
\\
$r_2$ &&&& \tn{n2}{\;\;\;\;\post{($e_2$)}} \\\myhline
$r_1$ &&& \multicolumn{2}{@{}l}{\tn{n1}{\post{($e_1$)}}} \\\myhline
$r_5$ & \multicolumn{2}{@{}l}{\tn{n5}{\texttt{x = 100}}} & \\\myhline
\\\myhline
$r_3$ && \multicolumn{3}{@{}l}{\tn{n3}{\;\;\;\;\readX{(x)}}} \\\myhline
\\\myhline
$r_4$ &  \multicolumn{2}{@{}l}{\tn{n4}{\readX{(x)}}} & \\\myhline
\\
\end{tabular}
\tikz[remember picture,overlay] \node[fit=(n4), draw] {};
\tikz[remember picture,overlay] \node[fit=(n5), draw] {};
\tikz[remember picture,overlay] \node[left of=n4,xshift=-.2cm] (e1) {$e_1$};
\tikz[remember picture,overlay] \node[left of=n5,xshift=-.2cm] (e2) {$e_2$};
\tikz[remember picture,overlay] \node (dummycaptionA) at ($(t1)!0.5!(t4)$) {};
\tikz[remember picture,overlay] \node[below of=dummycaptionA,xshift=.2cm,yshift=-2.8cm] 
(captionA) {\textbf{(a) sequence $z_3$}};
}
\end{minipage}
\end{tabular*}%
\vspace{4mm}
\caption{An example illustrating challenges in reordering dependent 
\texttt{read} - \texttt{write} operations.}
\label{fig:readwrite-ex1}
\end{figure*}

\begin{figure}[t]
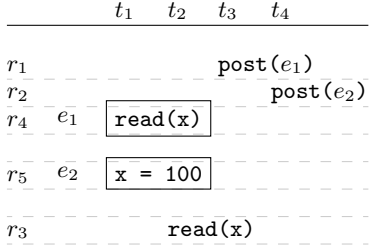

\centering
\begin{tabular}{@{}r@{\;\;\;\;\;\;\;\;\;\;\;\;\;}l@{\;}l@{\;}l@{\;}l@{}}
 & \tn{t1}{$t_1$}  & \tn{t2}{\;\;\;\;$t_2$} & \tn{t3}{$t_3$} & \tn{t4}{\;\;\;\;$t_4$} \\
\hline
\\
$r_1$ &&& \multicolumn{2}{@{}l}{\tn{n1}{\post{($e_1$)}}} \\\myhline
$r_2$ &&&& \tn{n2}{\;\;\;\;\post{($e_2$)}} \\\myhline
$r_4$ &  \multicolumn{2}{@{}l}{\tn{n4}{\readX{(x)}}} & \\\myhline
\\\myhline
$r_5$ & \multicolumn{2}{@{}l}{\tn{n5}{\texttt{x = 100}}} & \\\myhline
\\\myhline
$r_3$ && \multicolumn{3}{@{}l}{\tn{n3}{\;\;\;\;\readX{(x)}}} \\\myhline
\\
\end{tabular}
\tikz[remember picture,overlay] \node[fit=(n4), draw] {};
\tikz[remember picture,overlay] \node[fit=(n5), draw] {};
\tikz[remember picture,overlay] \node[left of=n4,xshift=-.2cm] (e1) {$e_1$};
\tikz[remember picture,overlay] \node[left of=n5,xshift=-.2cm] (e2) {$e_2$};

\caption{An interesting sequence $z$ corresponding to Example~\ref{ex:readwrite-1}, 
not explored by EM-DPOR when \texttt{read}s 
to same variable are considered independent.}
\label{fig:readwrite-ex1-missing}
\end{figure}

\section{Optimizations to \emdpor}
\label{app:emdpor-opt}

This section presents two main optimizations that we have applied to \emdpor\
(see Section~\ref{sec:algo}) to 
further prune the exploration of redundant states and transitions. Both of these
optimizations refine the set of pairs of dependent transitions and thus reduce the 
number of pairs of transitions considered dependent. We present 
modifications to \emdpor\ so as to not miss exploring interesting transition 
sequences when using the refined notion of dependence.

\subsection{Eliminate \texttt{read} - \texttt{read} Dependence}
\label{sec:read-write-dep}
\emdpor\ algorithm presented in Section~\ref{sec:emdpor-algo} needs to consider each pair of 
\texttt{read} operations to the  same shared variable as dependent, to not miss 
some interesting interleavings. This may result in exploring many redundant 
transition sequences.
A minor variation to Algorithm~\algoexplore\ (Algorithm \texttt{Explore}) while keeping
the algorithms \texttt{FindTarget}, \texttt{ReschedulePending} and 
\texttt{BacktrackEager} as is solves this problem. 
With this variation, \emdpor\ considers a \texttt{read} operation to be dependent 
only with a conflicting \texttt{write} operation. 

Before presenting the modifications to Algorithm~\texttt{Explore} in Section~\ref{sec:explore-modified}, 
we discuss some examples for which applying \emdpor\ presented in Section~\ref{sec:algo} as is, considering a 
pair of \texttt{read} operations to the same shared variable to be 
\emph{independent}, does not explore all possible partial orders of dependent 
transitions.

\subsubsection{Problematic Cases}
\label{sec:readwrite-examples}

\begin{example}
\label{ex:readwrite-1}
\upshape
Consider an execution trace $z_1$ given in Figure~\ref{fig:readwrite-ex1}(a), 
of an Android program. Among the threads $t_1$, $t_2$ and $t_3$
and $t_4$, only $t_1$ is associated with an event queue. Sequence 
$z_1$ has two pairs of \emph{may be co-enabled} or \emph{may be 
reordered} dependent transitions: $(r_3,r_5)$ and $(r_4,r_5)$, assuming every 
pair of \texttt{read} transitions to be independent.

Assume \emdpor\ to initially explore the sequence $z_1$ given in 
Figure~\ref{fig:readwrite-ex1}(a). On exploring a prefix of $z_1$ upto $r_4$, 
line~\linecond\ of Algorithm~\algoexplore\ (\texttt{Explore}) identifies $r_4$
and $r_5$ to be nearest pair of dependent and \emph{may be reordered} 
transitions executing on different handlers on the same thread. 
\texttt{FindTarget} invoked to compute backtracking choices to reorder $r_4$ 
and $r_5$ identifies $r_1$ and $r_2$ to be the corresponding diverging \texttt{post}s to be 
reordered. Thus, thread $t_4$ is added to backtracking set at the state $pre(z_1,r_1)$, \ie\ the state from which
$r_1$ is executed in sequence $z_1$. This eventually reorders $r_1$ and $r_2$ and results in 
exploring sequence $z_2$ (Figure~\ref{fig:readwrite-ex1}(b)).  On executing a 
prefix of $z_2$, transitions $r_3$ and $r_5$ are identified to be nearest 
dependent and co-enabled transitions. \texttt{FindTarget} reorders these two, 
eventually exploring sequence $z_3$ (Figure~\ref{fig:readwrite-ex1}(c)). \emdpor\ 
does not explore any other partial orders over $r_3$, $r_4$ and $r_5$ after sequence 
$z_3$. 

We note that, on seeing sequence $z_1$ \emdpor\ does not attempt to reorder 
$r_3$ and $r_5$, as $r_3$ is not the nearest reorderable dependent transition 
corresponding to $r_5$ in sequence $z_1$. As a result, 
Algorithm \texttt{Explore} considering a \texttt{read} to be only dependent 
with a conflicting \texttt{write} operation, misses exploring a sequence similar 
to $z$ (Figure~\ref{fig:readwrite-ex1-missing}), where $r_3$ reads the write 
performed by $r_5$ while $r_4$ does not.
\end{example}

\begin{figure}[h]
\centering
\begin{tikzpicture}[auto,node distance=16 mm,semithick, scale=.7, transform
shape]
\tikzstyle{stateNode} = [circle,draw=black,thick,inner sep=0pt,minimum 
size=7mm]
\tikzstyle{blankNode} = [thick,inner sep=0pt]
\tikzstyle{branchNode} = [isosceles triangle,draw,shape border 
rotate=90,isosceles triangle stretches=true, minimum height=20mm,minimum 
width=12mm,inner sep=0,anchor=north]

\node[stateNode] (sinit)  {$s_{init}$};

\node[stateNode] (s1) [below of=sinit,yshift=.3cm,xshift=-1.5cm] 
{$s_1$} edge [<-] (sinit);
\node[stateNode] (si) [below of=s1,yshift=.3cm, xshift=-1cm] 
{$s_i$} edge [<-] (s1);
\node[blankNode] (siB) [below of=si,yshift=.3cm] 
{} edge [<-] (si);
\node at (siB) [isosceles triangle, shape border rotate=+90,
draw, isosceles triangle stretches=true, minimum height=20mm,minimum 
width=12mm,inner sep=0, anchor=north] (siBranch) {};
\node[stateNode] (sj) [below of=s1,yshift=.3cm,xshift=1cm] 
{$s_j$} edge [<-] (s1);
\node[blankNode] (sjB) [below of=sj,yshift=.3cm] 
{} edge [<-] (sj);
\node at (sjB) [isosceles triangle, shape border rotate=+90,
draw, isosceles triangle stretches=true, minimum height=20mm,minimum 
width=12mm,inner sep=0, anchor=north] (sjBranch) {};

\node[blankNode] (dummynode1) at ($(si)!0.5!(sj)$) {$\ldotp$ $\ldotp$ 
$\ldotp$ $\ldotp$};

\node[blankNode] (arrowNode1) [left of=sinit, xshift=.85cm] {};
\node[blankNode] (arrowNode2) [left of=s1, xshift=.85cm] {} edge [<-,very 
thick,blue] (arrowNode1);
\node[blankNode] (arrowNode3) [left of=si, xshift=.85cm] {} edge [<-,very 
thick,blue] (arrowNode2);
\node[blankNode] (arrowNode4) [right of=si, xshift=-1.15cm] {};
\node[blankNode] (arrowNode5) [below of=s1, yshift=1.05cm] {} edge [<-,very 
thick, blue] (arrowNode4);
\node[blankNode] (arrowNode6) [left of=sj, xshift=1.1cm] {} edge [<-,very 
thick, blue] (arrowNode5);

\node[stateNode] (s2) [below of=sinit,yshift=.3cm,xshift=1.3cm] 
{$s_2$} edge [<-] (sinit);
\node[blankNode] (s2B) [below of=s2,yshift=.3cm] 
{} edge [<-] (s2);
\node at (s2B) [isosceles triangle, shape border rotate=+90,
draw, isosceles triangle stretches=true, minimum height=33mm,minimum 
width=12mm,inner sep=0, anchor=north] (s2Branch) {};

\node[stateNode] (s3) [below of=sinit,yshift=.3cm,xshift=3cm] 
{$s_n$} edge [<-] (sinit);
\node[blankNode] (s3B) [below of=s3,yshift=.3cm] 
{} edge [<-] (s3);
\node at (s3B) [isosceles triangle, shape border rotate=+90,
draw, isosceles triangle stretches=true, minimum height=33mm,minimum 
width=12mm,inner sep=0, anchor=north] (s3Branch) {};

\node[blankNode] (dummynode2) at ($(s2)!0.5!(s3)$) {$\ldotp$ $\ldotp$ 
$\ldotp$ $\ldotp$};
\end{tikzpicture}
\caption{Systematic exploration of branches in DFS based dynamic POR.}
\label{fig:dfs-branches}
\end{figure}
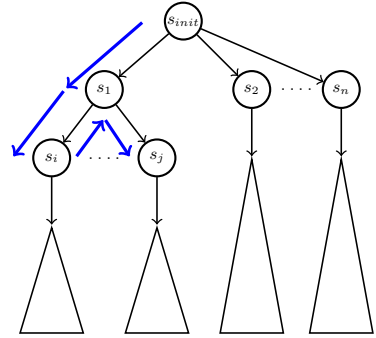

\begin{figure*}[t]
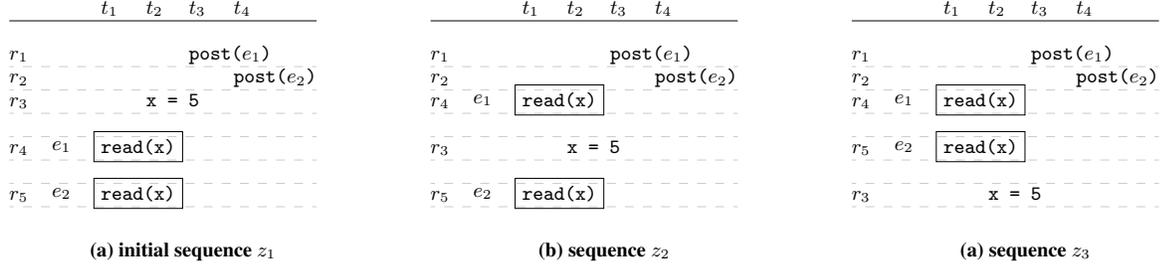

\begin{tabular*}{\textwidth}{l@{\;\;}c@{\;\;}r}
\begin{minipage}{\dimexpr0.33\textwidth-2\tabcolsep}
\centering
\scalebox{0.85}{
\begin{tabular}{@{}r@{\;\;\;\;\;\;\;\;\;\;\;\;\;}l@{\;}l@{\;}l@{\;}l@{}}
 & \tn{t1}{$t_1$}  & \tn{t2}{\;\;\;\;$t_2$} & \tn{t3}{$t_3$} & \tn{t4}{\;\;\;\;$t_4$} \\
\hline
\\
$r_1$ &&& \multicolumn{2}{@{}l}{\tn{n1}{\post{($e_1$)}}} \\\myhline
$r_2$ &&&& \tn{n2}{\;\;\;\;\post{($e_2$)}} \\\myhline
$r_3$ && \multicolumn{3}{@{}l}{\tn{n3}{\;\;\;\;\texttt{x = 5}}} \\\myhline
\\\myhline
$r_4$ &  \multicolumn{2}{@{}l}{\tn{n4}{\readX{(x)}}} & \\\myhline
\\\myhline
$r_5$ & \multicolumn{2}{@{}l}{\tn{n5}{\readX{(x)}}} & \\\myhline
\\
\end{tabular}
\tikz[remember picture,overlay] \node[fit=(n4), draw] {};
\tikz[remember picture,overlay] \node[fit=(n5), draw] {};
\tikz[remember picture,overlay] \node[left of=n4,xshift=-.2cm] (e1) {$e_1$};
\tikz[remember picture,overlay] \node[left of=n5,xshift=-.2cm] (e2) {$e_2$};
\tikz[remember picture,overlay] \node (dummycaptionA) at ($(t1)!0.5!(t4)$) {};
\tikz[remember picture,overlay] \node[below of=dummycaptionA,xshift=.2cm,yshift=-2.8cm] 
(captionA) {\textbf{(a) initial sequence $z_1$}};
}
\end{minipage}%
&
\begin{minipage}{\dimexpr0.33\textwidth-2\tabcolsep}
\centering
\scalebox{0.85}{
\begin{tabular}{@{}r@{\;\;\;\;\;\;\;\;\;\;\;\;\;}l@{\;}l@{\;}l@{\;}l@{}}
 & \tn{t1}{$t_1$}  & \tn{t2}{\;\;\;\;$t_2$} & \tn{t3}{$t_3$} & \tn{t4}{\;\;\;\;$t_4$} \\
\hline
\\
$r_1$ &&& \multicolumn{2}{@{}l}{\tn{n1}{\post{($e_1$)}}} \\\myhline
$r_2$ &&&& \tn{n2}{\;\;\;\;\post{($e_2$)}} \\\myhline
$r_4$ &  \multicolumn{2}{@{}l}{\tn{n4}{\readX{(x)}}} & \\\myhline
\\\myhline
$r_3$ && \multicolumn{3}{@{}l}{\tn{n3}{\;\;\;\;\texttt{x = 5}}} \\\myhline
\\\myhline
$r_5$ & \multicolumn{2}{@{}l}{\tn{n5}{\readX{(x)}}} & \\\myhline
\\
\end{tabular}
\tikz[remember picture,overlay] \node[fit=(n4), draw] {};
\tikz[remember picture,overlay] \node[fit=(n5), draw] {};
\tikz[remember picture,overlay] \node[left of=n4,xshift=-.2cm] (e1) {$e_1$};
\tikz[remember picture,overlay] \node[left of=n5,xshift=-.2cm] (e2) {$e_2$};
\tikz[remember picture,overlay] \node (dummycaptionA) at ($(t1)!0.5!(t4)$) {};
\tikz[remember picture,overlay] \node[below of=dummycaptionA,xshift=.2cm,yshift=-2.8cm] 
(captionA) {\textbf{(b) sequence $z_2$}};
}
\end{minipage}
&
\begin{minipage}{\dimexpr0.33\textwidth-2\tabcolsep}
\centering
\scalebox{0.85}{
\begin{tabular}{@{}r@{\;\;\;\;\;\;\;\;\;\;\;\;\;}l@{\;}l@{\;}l@{\;}l@{}}
 & \tn{t1}{$t_1$}  & \tn{t2}{\;\;\;\;$t_2$} & \tn{t3}{$t_3$} & \tn{t4}{\;\;\;\;$t_4$} \\
\hline
\\
$r_1$ &&& \multicolumn{2}{@{}l}{\tn{n1}{\post{($e_1$)}}} \\\myhline
$r_2$ &&&& \tn{n2}{\;\;\;\;\post{($e_2$)}} \\\myhline
$r_4$ &  \multicolumn{2}{@{}l}{\tn{n4}{\readX{(x)}}} & \\\myhline
\\\myhline
$r_5$ & \multicolumn{2}{@{}l}{\tn{n5}{\readX{(x)}}} & \\\myhline
\\\myhline
$r_3$ && \multicolumn{3}{@{}l}{\tn{n3}{\;\;\;\;\texttt{x = 5}}} \\\myhline
\\
\end{tabular}
\tikz[remember picture,overlay] \node[fit=(n4), draw] {};
\tikz[remember picture,overlay] \node[fit=(n5), draw] {};
\tikz[remember picture,overlay] \node[left of=n4,xshift=-.2cm] (e1) {$e_1$};
\tikz[remember picture,overlay] \node[left of=n5,xshift=-.2cm] (e2) {$e_2$};
\tikz[remember picture,overlay] \node (dummycaptionA) at ($(t1)!0.5!(t4)$) {};
\tikz[remember picture,overlay] \node[below of=dummycaptionA,xshift=.2cm,yshift=-2.8cm] 
(captionA) {\textbf{(a) sequence $z_3$}};
}
\end{minipage}
\end{tabular*}%
\vspace{4mm}
\caption{An example illustrating challenges in reordering dependent 
\texttt{read} - \texttt{write} operations even when a \texttt{write} is attempted
to be reordered with multiple prior \texttt{read}s.}
\label{fig:readwrite-ex2}
\end{figure*}

\begin{figure}[t]
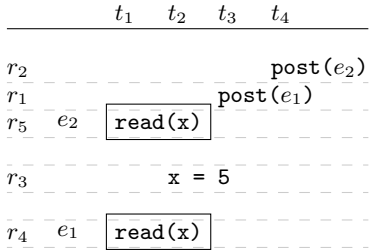

\centering
\begin{tabular}{@{}r@{\;\;\;\;\;\;\;\;\;\;\;\;\;}l@{\;}l@{\;}l@{\;}l@{}}
 & \tn{t1}{$t_1$}  & \tn{t2}{\;\;\;\;$t_2$} & \tn{t3}{$t_3$} & \tn{t4}{\;\;\;\;$t_4$} \\
\hline
\\
$r_2$ &&&& \tn{n2}{\;\;\;\;\post{($e_2$)}} \\\myhline
$r_1$ &&& \multicolumn{2}{@{}l}{\tn{n1}{\post{($e_1$)}}} \\\myhline
$r_5$ & \multicolumn{2}{@{}l}{\tn{n5}{\readX{(x)}}} & \\\myhline
\\\myhline
$r_3$ && \multicolumn{3}{@{}l}{\tn{n3}{\;\;\;\;\texttt{x = 5}}} \\\myhline
\\\myhline
$r_4$ &  \multicolumn{2}{@{}l}{\tn{n4}{\readX{(x)}}} & \\\myhline
\\
\end{tabular}
\tikz[remember picture,overlay] \node[fit=(n4), draw] {};
\tikz[remember picture,overlay] \node[fit=(n5), draw] {};
\tikz[remember picture,overlay] \node[left of=n4,xshift=-.2cm] (e1) {$e_1$};
\tikz[remember picture,overlay] \node[left of=n5,xshift=-.2cm] (e2) {$e_2$};

\caption{An interesting sequence $z$ corresponding to Example~\ref{ex:readwrite-2}, 
not explored by EM-DPOR when \texttt{read}s 
to same variable are considered independent.}
\label{fig:readwrite-ex2-missing}
\end{figure}

\paragraph{Analysis of Example~\ref{ex:readwrite-1}.} A DFS  based explorer 
explores all paths originating at a state in the state space before backtracking to a 
prior state in the search stack and exploring other branches. \emdpor\ is a POR 
algorithm which prunes some redundant transition sequences explored by a 
na\"{i}ve DFS based state space explorer. Hence, \emdpor\ should explore all the 
interesting interleaving of dependent transitions originating, say at some state 
$s_i$, before backtracking to a prior state in the stack, say $s_1$, and 
exploring other branches. This is because, after backtracking to state $s_1$ 
from $s_i$, the subspace rooted at $s_i$ will not be visited again. Thus, any interleaving of 
dependent transitions that could be explored only from $s_i$ will be missed, if 
not explored before backtracking to $s_1$. This is pictorially depicted in 
Figure~\ref{fig:dfs-branches}. In Figure~\ref{fig:dfs-branches} triangles 
represent state space reachable from the states to which the triangles are 
connected. Even though not shown in the figure, some of the states may overlap. 
The thick directed arrows depict the way in which state exploration proceeds. 
In case of Example~\ref{ex:readwrite-1}, \emdpor\ backtracked from a state even before 
exploring all the \emph{non-redundant} interleaving of dependent transitions 
reachable from that state. This is the cause of missing some interesting 
sequences.

We can solve this issue with \emdpor\ without having to consider every pair of 
\texttt{read} operations to the same memory location as dependent, as follows. 
In the context of line~\linecond\ of Algorithm~\texttt{Explore}, when the next 
transition on a thread $t$ in sate $s$ is a \texttt{write}, its
dependent transition can either be a \texttt{read} or a \texttt{write} to the 
same shared variable. Instead of invoking \texttt{FindTarget} to reorder a 
\texttt{write} operation $r'$ with its nearest executed dependent transition, 
we identify all the \emph{may be co-enabled} or \emph{may be reordered} 
dependent transitions upto nearest executed \texttt{write} operation, and 
compute backtracking choices to reorder all these identified dependent 
transitions with the \texttt{write} operation $r'$.

\begin{example}
\label{ex:readwrite-2}
\upshape
Consider an execution trace $z_1$ given in Figure~\ref{fig:readwrite-ex2}(a), 
of an event-driven multi-threaded program. Among the threads $t_1$, $t_2$, $t_3$
and $t_4$, only $t_1$ is associated with an event queue. Sequence 
$z_1$ has two pairs of \emph{may be co-enabled} dependent transitions: 
$(r_3,r_4)$ and $(r_3,r_5)$, assuming any pair of \texttt{read} operations to 
be independent.

Assume EM-DPOR initially explores sequence $z_1$ in 
Figure~\ref{fig:readwrite-ex1}(a). Algorithm~\texttt{Explore} identifies 
transition pairs $(r_3,r_4)$ and $(r_3,r_5)$ as dependent, identifies 
backtracking choices using \texttt{FindTarget}, and eventually explores 
sequences $z_2$ and $z_3$ (Figure~\ref{fig:readwrite-ex2-missing}(b) and (c) 
respectively). However, EM-DPOR does not explore any more interleaving of 
transitions $r_3$, $r_4$ and $r_5$, even if we use the modification discussed 
in the analysis presented for Example~\ref{ex:readwrite-1} (this modification 
computes backtracking choices to reorder $r_3$ with both $r_4$ and $r_5$ in 
sequence $z_3$, but is ineffective in this case). As a result, EM-DPOR misses 
exploring a transition sequence similar to $z$ (Figure~\ref{fig:readwrite-ex2-missing}), 
where $r_4$ reads the write performed by $r_3$ while $r_5$ does not.
\end{example}

\begin{figure*}[t]
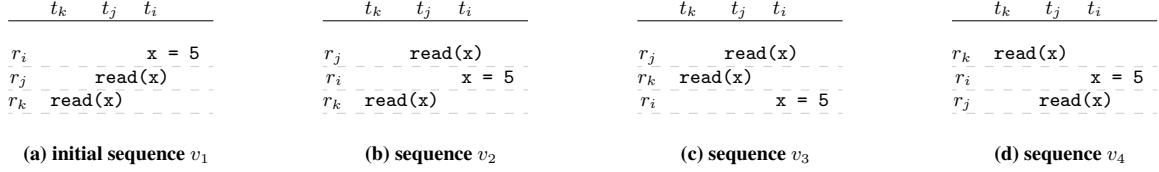

\begin{tabular*}{\textwidth}{l@{\;\;}c@{\;\;}c@{\;\;}r}
\begin{minipage}{\dimexpr0.25\textwidth-2\tabcolsep}
\centering
\scalebox{0.85}{
\begin{tabular}{@{}r@{\;\;\;\;}l@{\;}l@{\;}l@{}}
 & \tn{ti}{$t_k$}  & \tn{tj}{\;\;\;\;\;$t_j$} & \tn{tk}{\;$t_i$} \\
\hline
\\
$r_i$ && \multicolumn{2}{@{}l}{\tn{n1}{\;\;\;\;\;\;\;\;\;\;\;\;\;\texttt{x = 5}}} \\\myhline
$r_j$ & \multicolumn{3}{@{}l}{\tn{n3}{\;\;\;\;\;\;\;\;\readX{(x)}}} \\\myhline
$r_k$ &  \multicolumn{2}{@{}l}{\tn{n4}{\readX{(x)}}} &\\\myhline
\\
\end{tabular}
\tikz[remember picture,overlay] \node (dummycaptionA) at ($(ti)!0.5!(tk)$) {};
\tikz[remember picture,overlay] \node[below of=dummycaptionA,xshift=.2cm,yshift=-1.3cm] 
(captionA) {\textbf{(a) initial sequence $v_1$}};
}
\end{minipage}%
&
\begin{minipage}{\dimexpr0.25\textwidth-2\tabcolsep}
\centering
\scalebox{0.85}{
\begin{tabular}{@{}r@{\;\;\;\;}l@{\;}l@{\;}l@{}}
 & \tn{ti}{$t_k$}  & \tn{tj}{\;\;\;\;\;$t_j$} & \tn{tk}{\;$t_i$} \\
\hline
\\
$r_j$ & \multicolumn{3}{@{}l}{\tn{n3}{\;\;\;\;\;\;\;\;\readX{(x)}}} \\\myhline
$r_i$ && \multicolumn{2}{@{}l}{\tn{n1}{\;\;\;\;\;\;\;\;\;\;\;\;\;\texttt{x = 5}}} \\\myhline
$r_k$ &  \multicolumn{2}{@{}l}{\tn{n4}{\readX{(x)}}} &\\\myhline
\\
\end{tabular}
\tikz[remember picture,overlay] \node (dummycaptionA) at ($(ti)!0.5!(tk)$) {};
\tikz[remember picture,overlay] \node[below of=dummycaptionA,xshift=.2cm,yshift=-1.3cm] 
(captionA) {\textbf{(b) sequence $v_2$}};
}
\end{minipage}
&
\begin{minipage}{\dimexpr0.25\textwidth-2\tabcolsep}
\centering
\scalebox{0.85}{
\begin{tabular}{@{}r@{\;\;\;\;}l@{\;}l@{\;}l@{}}
 & \tn{ti}{$t_k$}  & \tn{tj}{\;\;\;\;\;$t_j$} & \tn{tk}{\;$t_i$} \\
\hline
\\
$r_j$ & \multicolumn{3}{@{}l}{\tn{n3}{\;\;\;\;\;\;\;\;\readX{(x)}}} \\\myhline
$r_k$ &  \multicolumn{2}{@{}l}{\tn{n4}{\readX{(x)}}} &\\\myhline
$r_i$ && \multicolumn{2}{@{}l}{\tn{n1}{\;\;\;\;\;\;\;\;\;\;\;\;\;\texttt{x = 5}}} \\\myhline
\\
\end{tabular}
\tikz[remember picture,overlay] \node (dummycaptionA) at ($(ti)!0.5!(tk)$) {};
\tikz[remember picture,overlay] \node[below of=dummycaptionA,xshift=.2cm,yshift=-1.3cm] 
(captionA) {\textbf{(c) sequence $v_3$}};
}
\end{minipage}
&
\begin{minipage}{\dimexpr0.25\textwidth-2\tabcolsep}
\centering
\scalebox{0.85}{
\begin{tabular}{@{}r@{\;\;\;\;}l@{\;}l@{\;}l@{}}
 & \tn{ti}{$t_k$}  & \tn{tj}{\;\;\;\;\;$t_j$} & \tn{tk}{\;$t_i$} \\
\hline
\\
$r_k$ &  \multicolumn{2}{@{}l}{\tn{n4}{\readX{(x)}}} &\\\myhline
$r_i$ && \multicolumn{2}{@{}l}{\tn{n1}{\;\;\;\;\;\;\;\;\;\;\;\;\;\texttt{x = 5}}} \\\myhline
$r_j$ & \multicolumn{3}{@{}l}{\tn{n3}{\;\;\;\;\;\;\;\;\readX{(x)}}} \\\myhline
\\
\end{tabular}
\tikz[remember picture,overlay] \node (dummycaptionA) at ($(ti)!0.5!(tk)$) {};
\tikz[remember picture,overlay] \node[below of=dummycaptionA,xshift=.2cm,yshift=-1.3cm] 
(captionA) {\textbf{(d) sequence $v_4$}};
}
\end{minipage}
\end{tabular*}%
\vspace{4mm}
\caption{A scenario analogous to that presented in Figure~\ref{fig:readwrite-ex2}
but in the context of a multi-threaded program.}
\label{fig:readwrite-ex2-analog}
\end{figure*}

\paragraph{Analysis of Example~\ref{ex:readwrite-2}} A scenario in case of a pure 
multi-threaded program analogous to that in Figure~\ref{fig:readwrite-ex2}, is 
shown in Figure~\ref{fig:readwrite-ex2-analog}. Transitions $r_i$, $r_j$ and 
$r_k$ executed on threads $t_i$, $t_j$ and $t_k$ respectively in 
Figure~\ref{fig:readwrite-ex2-analog} correspond to transitions $r_3$, 
$r_4$ and $r_5$ respectively in Figure~\ref{fig:readwrite-ex2}. 
Relative order of transitions $r_i$, $r_j$ and $r_k$ in 
Figure~\ref{fig:readwrite-ex2-analog}(a), (b) and (c) correspond to relative 
order of $r_3$, $r_4$ and $r_5$ in Figure~\ref{fig:readwrite-ex2}(a), (b) 
and (c) respectively. On exploring sequence $v_1$ 
(Figure~\ref{fig:readwrite-ex2-analog}), DPOR (even EM-DPOR) adds thread $t_k$ 
to backtracking set at state prior to executing $r_i$, when computing 
backtracking choices to reorder $r_i$ and $r_k$. This leads to exploring 
sequence $v_4$ (Figure~\ref{fig:readwrite-ex2-analog}(d)) whose analogue is not 
explored by EM-DPOR in case of Example~\ref{ex:readwrite-2} when considering 
\texttt{read} operations to be independent. Thus, for EM-DPOR to explore 
sequence $z$ (Figure~\ref{fig:readwrite-ex2-missing}), \texttt{FindTarget} 
called to reorder $r_3$ and $r_5$ in sequence $z_1$, should be able to identify 
the presence of \texttt{read}s to same variable between the transitions $r_3$ and $r_5$ and coming 
from other handlers on the same thread. \texttt{FindTarget} should then reorder \texttt{post}s of such 
handlers with \post\ of $e_2$. However, this involves modifications to 
Algorithm \texttt{FindTarget}. Instead of modifying \texttt{FindTarget}, we 
provide minor modifications to Algorithm \texttt{Explore} to identify relevant 
event handlers to be reordered in such scenarios. In case of 
Example~\ref{ex:readwrite-2}, our modification identifies events $e_1$ 
and $e_2$ for reordering on exploring sequence $z_3$ 
(Figure~\ref{fig:readwrite-ex2}(c)) upto the transition $r_5$. 

In addition to the modification we discussed under analysis for 
Example~\ref{ex:readwrite-1}, we do the following in Algorithm \texttt{Explore}.
After computing backtracking choices to reorder a \texttt{write} transition 
$r'$ with its nearest executed reorderable dependent transition $r$, 
we assume a temporary happens-before mapping from $r$ to $r'$. We then invoke  
\texttt{FindTarget} to reorder $r'$ with other conflicting \readX{} transitions upto 
the nearest executed conflicting \texttt{write}. Invoking 
\texttt{FindTarget} assuming such a happens-before relation from $r$ to $r'$, enables 
\texttt{FindTarget} to add tasks corresponding to $r$ into the set $candidates$ 
computed by the steps of \texttt{FindTarget}
(refer Algorithm~\algofindtarget). In scenarios similar 
to sequence $z_3$ in Example~\ref{ex:readwrite-2}, this enables 
\texttt{FindTarget} to reach Step~3, invoke \texttt{ReschedulePending} (line~\linerecursivefindtarget\
in Algorithm~\algofindtarget) and identify \post\ operations of relevant 
event handlers for reordering. The modified version of 
Algorithm~\texttt{Explore} is presented as Algorithm~\ref{algo:explore-optimized}.

{
\SetEndCharOfAlgoLine{}
\SetInd{6pt}{0.01pt}
\begin{algorithm*}[t!]
\small
\SetKwFunction{findTarget}{FindTarget}
\SetKwFunction{Explore}{Explore}

\SetKwInput{KwInput}{Input}
\KwInput{a transition sequence \boldmath$w$\unboldmath$: r_1 \ldots r_n$ and a set 
\boldmath$rp$ \unboldmath of posts to be reordered}

Let $s$ $=$ $last(w)$; \phantom{xx} $RP(s) = rp$\;
\nllabel{mod-explore-line-init}

\ForEach{thread $t$}{
\nllabel{mod-explore-loop-start}
\Indp
\If{$\exists i = \max(\{ i \in dom(w) \mid r_i \text{ is dependent and }
  (\text{may be co-enabled or reordered}$\\ 
  \nonl\hspace{4mm}$\text{ with } next(s,t) ) \text{ and }  
  i\; {\not\to}_{w}\; task(next(s,t))\})$
\nllabel{mod-explore-line-cond}
}{
   \Indp
   \tcp{Identify backtracking point and choice to reorder $r_i$ and $next(s,t)$}
   \texttt{FindTarget}$(w, r_i, next(s,t))$\;
   \nllabel{mod-explore-line-findtarget}
   
    \If{$opType(next(s,t)) =$ \upshape{WRITE} \nllabel{mod-explore-line-ifwrite}}{
       \Indp
       Add a happens-before edge between $r_i$ and $next(s,t)$\; 
       \nllabel{mod-explore-read-write-temp-hb}
       Let $i' = max(\{i' \in dom(w) \mid opType(r_{i'}) = \text{ WRITE}$\\ 
       \nonl\hspace{4mm}$\text{and } var(r_{i'}) = var(next(s,t))\} \cup \{-1\})$\;
       \nllabel{mod-explore-line-recent-write}
       \ForEach{$j \in dom(w) \mid r_j \text{ is dependent and } (\text{may be co-enabled}$\\ 
       \nonl\hspace{4mm}$\text{or reordered with } next(s,t)) \text{ and } i\; {\not\to}_{w}\; task(next(s,t)) \text{ and } j \geq i'$ \nllabel{mod-explore-foreach-dep}}{
           \Indp
           \texttt{FindTarget}$(w, r_j, next(s,t))$ 
           \nllabel{mod-explore-line-findtarget-2}
       }
       \nllabel{mod-explore-foreach-dep-end}
       Remove the happens-before edge between $r_i$ and $next(s,t)$\;
       \nllabel{mod-explore-remove-hb}
   } \nllabel{mod-explore-line-ifwrite-end}
}
}
\nllabel{mod-explore-loop-end}

\If{$\exists t \in enabled(s)$}{
\nllabel{mod-explore-line-dfs-start}
\Indp
  Let $backtrack(s)$ $=$ $\{t\}$ and $done(s)$ $=$ $\emptyset$\;
  \tcp{Perform selective depth-first traversal}
  \While{$\exists t \in (backtrack(s) \setminus done(s))$}{
\Indp
    Let $r = next(s,t)$; Execute transition $r$\;
    \If{$r$ is a \emph{\texttt{post}} operation \nllabel{mod-explore-line-ifpost} }{
\Indp
    \If{$\exists k = max(\{k \in dom(w) \mid r_k \in reorderedPosts(r,w.r)\})$  \nllabel{mod-explore-line-maxpost}}
    {  
    \Indp
    Add thread $t$ to $backtrack(pre(w,k))$\;
    \nllabel{mod-explore-line-addpost-to-backtrack}
    }
    $rp$ = $RP(s)\setminus\{(r,\_) \in RP(s)\}$\;
    \nllabel{mod-explore-line-remover-from-rp}
    }
    \nllabel{mod-explore-line-endifpost}
    Add $t$ to $done(s)$; $\texttt{Explore}(w \cdot r)$\;
    \nllabel{mod-explore-line-done}
  }
}
\nllabel{mod-explore-line-dfs-end}
\caption{\texttt{Explore}}
\label{algo:explore-optimized}
\end{algorithm*}
}

\subsubsection{Modifications to Algorithm \texttt{Explore}}
\label{sec:explore-modified}
Algorithm \texttt{Explore} (given in Algorithm~\ref{algo:explore-optimized}) modified to consider any pair of \texttt{read}s to the same 
variable to be independent, is similar to Algorithm~\algoexplore\ presented in Section~\ref{sec:algo}
except for lines~\ref{mod-explore-line-ifwrite}-\ref{mod-explore-line-ifwrite-end} in 
Algorithm~\ref{algo:explore-optimized}. The line numbers referred henceforth correspond
to Algorithm~\ref{algo:explore-optimized}. Function
$opType(r)$ finds the type of visible operation in transition $r$. 
After invoking \texttt{FindTarget} on line~\ref{mod-explore-line-findtarget} to compute
backtracking choices and backtracking state to reorder $next(s,t)$ with the nearest 
\emph{may be co-enabled or reordered} dependent transition, 
lines~\ref{mod-explore-read-write-temp-hb}-\ref{mod-explore-remove-hb} are executed only if 
$next(s,t)$ has a \texttt{write} as visible operation. 

Line~\ref{mod-explore-read-write-temp-hb}
adds a temporary happens-before mapping from nearest dependent transition $r_i$ to $next(s,t)$.
Note that the happens-before relation defined in Definition~\defhappensbefore\ does not 
allow such a mapping. However, we can achieve the mapping $i \to_w task(next(s,t))$ 
by assuming each transition to be prefixed by a $NOP$ operation which does not alter the
state. We execute the $NOP$ operation and add a happens-before mapping from $r_i$ to this
$NOP$, which results in the required $i\to_{w.NOP}task(next(s,t))$. 
Line~\ref{mod-explore-line-recent-write} computes the index of the most recent 
conflicting \texttt{write} and stores it in $i'$. Absence or prior \texttt{write}s
to the variable accessed by transition $next(s,t)$, assigns $-1$ to $i'$. 
Lines~\ref{mod-explore-foreach-dep}--\ref{mod-explore-foreach-dep-end} compute backtracking
choices and backtracking states to reorder $next(s,t)$ with all the prior 
\emph{may be co-enabled or reordered} dependent transitions upto $r_i'$ with no 
happens-before mapping between them. Line~\ref{mod-explore-remove-hb} removes the HB 
mapping between $r_i'$ and $NOP$ corresponding to $next(s,t)$. 

Addition of temporary HB mapping and computing backtracking information
for all the relevant dependent transitions when $next(s,t)$ is a \texttt{write}
operation, solves the issues explained through Examples~\ref{ex:readwrite-1}-\ref{ex:readwrite-2}.
In case of Example~\ref{ex:readwrite-2}, line~\ref{mod-explore-line-findtarget} of Algorithm~\ref{algo:explore-optimized}
invokes \texttt{FindTarget} to reorder transition $r_5$ and $r_3$ when exploring sequence $z_3$ (see Figure~\ref{fig:readwrite-ex2}).
Then, line~\ref{mod-explore-read-write-temp-hb} adds a temporary happens-before mapping
from $r_5$ to $r_3$. Lines~\ref{mod-explore-foreach-dep} and \ref{mod-explore-line-findtarget-2} 
identify $r_4$ as a relevant dependent transition to be reordered with $r_3$ and
invoke \texttt{FindTarget}. Due to happens-before mapping from $r_5$ to $r_3$, Step~2
of \texttt{FindTarget} (see Algorithm~\ref{algo:findtarget} in Section~\ref{sec:algo}) 
compute set $candidates = \{(t_2,\bot), (t_1,e_2)\}$.
Since both threads $t_1$ and $t_2$ corresponding to $candidates$ are in $done$ set
at the state from where $r_4$ is executed, Step~3 is reached. Step~3 of \texttt{FindTarget}
computes $pending = \{(t_1,e_2)\}$ and invokes \texttt{ReschedulePending} which
reorders events $e_1$ and $e_2$, eventually exploring the sequence given in Figure~\ref{fig:readwrite-ex2-missing}.
Thus, modified Algorithm \texttt{Explore} enables EM-DPOR to consider a \texttt{read}
operation to be dependent only with conflicting \texttt{write} operations, and thus
avoids exploring some redundant transition sequences reaching same final states.

\begin{figure*}[t]
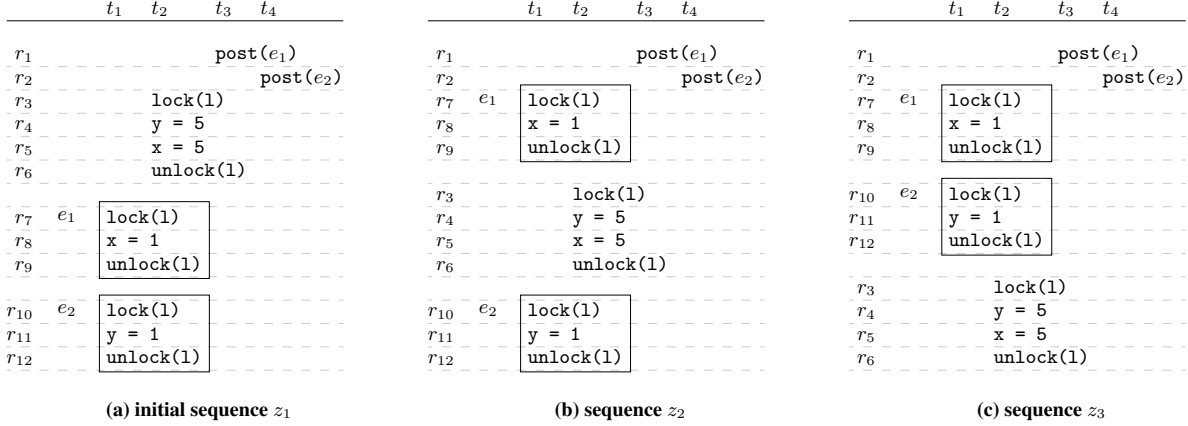

\begin{tabular*}{\textwidth}{l@{\;\;}c@{\;\;}r}
\begin{minipage}{\dimexpr0.33\textwidth-2\tabcolsep}
\centering
\scalebox{0.85}{
\begin{tabular}{@{}r@{\;\;\;\;\;\;\;\;\;\;\;\;\;}l@{\;}l@{\;}l@{\;}l@{}}
 & \tn{t1}{$t_1$}  & \tn{t2}{\;\;\;\;$t_2$} & \tn{t3}{$t_3$} & \tn{t4}{\;\;\;\;$t_4$} \\
\hline
\\
$r_1$ &&& \multicolumn{2}{@{}l}{\tn{n1}{\post{($e_1$)}}} \\\myhline
$r_2$ &&&& \tn{n2}{\;\;\;\;\post{($e_2$)}} \\\myhline
$r_3$ && \multicolumn{3}{@{}l}{\tn{n3}{\;\;\;\;\acquire{(l)}}} \\\myhline
$r_4$ && \multicolumn{3}{@{}l}{\tn{n4}{\;\;\;\;\texttt{y = 5}}} \\\myhline
$r_5$ && \multicolumn{3}{@{}l}{\tn{n5}{\;\;\;\;\texttt{x = 5}}} \\\myhline
$r_6$ && \multicolumn{3}{@{}l}{\tn{n6}{\;\;\;\;\release{(l)}}} \\\myhline
\\\myhline
$r_7$ &  \multicolumn{2}{@{}l}{\tn{n7}{\acquire{(l)}}} & \\\myhline
$r_8$ &  \multicolumn{2}{@{}l}{\tn{n8}{\texttt{x = 1}}} & \\\myhline
$r_9$ &  \multicolumn{2}{@{}l}{\tn{n9}{\release{(l)}}} & \\\myhline
\\\myhline
$r_{10}$ & \multicolumn{2}{@{}l}{\tn{n10}{\acquire{(l)}}} & \\\myhline
$r_{11}$ & \multicolumn{2}{@{}l}{\tn{n11}{\texttt{y = 1}}} & \\\myhline
$r_{12}$ & \multicolumn{2}{@{}l}{\tn{n12}{\release{(l)}}} & \\\myhline
\\
\end{tabular}
\tikz[remember picture,overlay] \node[fit=(n7)(n8)(n9), draw] {};
\tikz[remember picture,overlay] \node[fit=(n10)(n11)(n12), draw] {};
\tikz[remember picture,overlay] \node[left of=n7,xshift=-.2cm] (e1) {$e_1$};
\tikz[remember picture,overlay] \node[left of=n10,xshift=-.2cm] (e2) {$e_2$};
\tikz[remember picture,overlay] \node (dummycaptionA) at ($(t1)!0.5!(t4)$) {};
\tikz[remember picture,overlay] \node[below of=dummycaptionA,xshift=.2cm,yshift=-5.3cm] 
(captionA) {\textbf{(a) initial sequence $z_1$}};
}
\end{minipage}%
&
\begin{minipage}{\dimexpr0.33\textwidth-2\tabcolsep}
\centering
\scalebox{0.85}{
\begin{tabular}{@{}r@{\;\;\;\;\;\;\;\;\;\;\;\;\;}l@{\;}l@{\;}l@{\;}l@{}}
 & \tn{t1}{$t_1$}  & \tn{t2}{\;\;\;\;$t_2$} & \tn{t3}{$t_3$} & \tn{t4}{\;\;\;\;$t_4$} \\
\hline
\\
$r_1$ &&& \multicolumn{2}{@{}l}{\tn{n1}{\post{($e_1$)}}} \\\myhline
$r_2$ &&&& \tn{n2}{\;\;\;\;\post{($e_2$)}} \\\myhline
$r_7$ &  \multicolumn{2}{@{}l}{\tn{n7}{\acquire{(l)}}} & \\\myhline
$r_8$ &  \multicolumn{2}{@{}l}{\tn{n8}{\texttt{x = 1}}} & \\\myhline
$r_9$ &  \multicolumn{2}{@{}l}{\tn{n9}{\release{(l)}}} & \\\myhline
\\\myhline
$r_3$ && \multicolumn{3}{@{}l}{\tn{n3}{\;\;\;\;\acquire{(l)}}} \\\myhline
$r_4$ && \multicolumn{3}{@{}l}{\tn{n4}{\;\;\;\;\texttt{y = 5}}} \\\myhline
$r_5$ && \multicolumn{3}{@{}l}{\tn{n5}{\;\;\;\;\texttt{x = 5}}} \\\myhline
$r_6$ && \multicolumn{3}{@{}l}{\tn{n6}{\;\;\;\;\release{(l)}}} \\\myhline
\\\myhline
$r_{10}$ & \multicolumn{2}{@{}l}{\tn{n10}{\acquire{(l)}}} & \\\myhline
$r_{11}$ & \multicolumn{2}{@{}l}{\tn{n11}{\texttt{y = 1}}} & \\\myhline
$r_{12}$ & \multicolumn{2}{@{}l}{\tn{n12}{\release{(l)}}} & \\\myhline
\\
\end{tabular}
\tikz[remember picture,overlay] \node[fit=(n7)(n8)(n9), draw] {};
\tikz[remember picture,overlay] \node[fit=(n10)(n11)(n12), draw] {};
\tikz[remember picture,overlay] \node[left of=n7,xshift=-.2cm] (e1) {$e_1$};
\tikz[remember picture,overlay] \node[left of=n10,xshift=-.2cm] (e2) {$e_2$};
\tikz[remember picture,overlay] \node (dummycaptionA) at ($(t1)!0.5!(t4)$) {};
\tikz[remember picture,overlay] \node[below of=dummycaptionA,xshift=.2cm,yshift=-5.3cm] 
(captionA) {\textbf{(b) sequence $z_2$}};
}
\end{minipage}
&
\begin{minipage}{\dimexpr0.33\textwidth-2\tabcolsep}
\centering
\scalebox{0.85}{
\begin{tabular}{@{}r@{\;\;\;\;\;\;\;\;\;\;\;\;\;}l@{\;}l@{\;}l@{\;}l@{}}
 & \tn{t1}{$t_1$}  & \tn{t2}{\;\;\;\;$t_2$} & \tn{t3}{$t_3$} & \tn{t4}{\;\;\;\;$t_4$} \\
\hline
\\
$r_1$ &&& \multicolumn{2}{@{}l}{\tn{n1}{\post{($e_1$)}}} \\\myhline
$r_2$ &&&& \tn{n2}{\;\;\;\;\post{($e_2$)}} \\\myhline
$r_7$ &  \multicolumn{2}{@{}l}{\tn{n7}{\acquire{(l)}}} & \\\myhline
$r_8$ &  \multicolumn{2}{@{}l}{\tn{n8}{\texttt{x = 1}}} & \\\myhline
$r_9$ &  \multicolumn{2}{@{}l}{\tn{n9}{\release{(l)}}} & \\\myhline
\\\myhline
$r_{10}$ & \multicolumn{2}{@{}l}{\tn{n10}{\acquire{(l)}}} & \\\myhline
$r_{11}$ & \multicolumn{2}{@{}l}{\tn{n11}{\texttt{y = 1}}} & \\\myhline
$r_{12}$ & \multicolumn{2}{@{}l}{\tn{n12}{\release{(l)}}} & \\\myhline
\\\myhline
$r_3$ && \multicolumn{3}{@{}l}{\tn{n3}{\;\;\;\;\acquire{(l)}}} \\\myhline
$r_4$ && \multicolumn{3}{@{}l}{\tn{n4}{\;\;\;\;\texttt{y = 5}}} \\\myhline
$r_5$ && \multicolumn{3}{@{}l}{\tn{n5}{\;\;\;\;\texttt{x = 5}}} \\\myhline
$r_6$ && \multicolumn{3}{@{}l}{\tn{n6}{\;\;\;\;\release{(l)}}} \\\myhline
\\
\end{tabular}
\tikz[remember picture,overlay] \node[fit=(n7)(n8)(n9), draw] {};
\tikz[remember picture,overlay] \node[fit=(n10)(n11)(n12), draw] {};
\tikz[remember picture,overlay] \node[left of=n7,xshift=-.2cm] (e1) {$e_1$};
\tikz[remember picture,overlay] \node[left of=n10,xshift=-.2cm] (e2) {$e_2$};
\tikz[remember picture,overlay] \node (dummycaptionA) at ($(t1)!0.5!(t4)$) {};
\tikz[remember picture,overlay] \node[below of=dummycaptionA,xshift=.2cm,yshift=-5.3cm] 
(captionA) {\textbf{(c) sequence $z_3$}};
}
\end{minipage}
\end{tabular*}%
\vspace{4mm}
\caption{An example illustrating challenges in reordering dependent 
\texttt{lock} operations when using modified dependence relation..}
\label{fig:mod-lock-dep}
\end{figure*}

\begin{figure}[t]
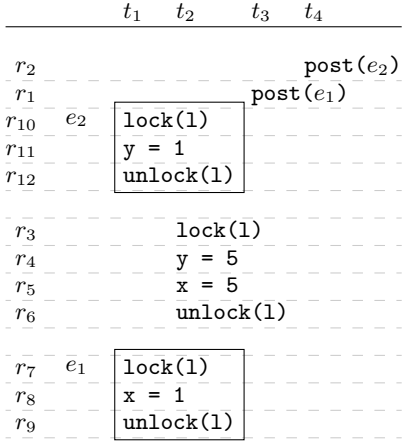

\centering
\begin{tabular}{@{}r@{\;\;\;\;\;\;\;\;\;\;\;\;\;}l@{\;}l@{\;}l@{\;}l@{}}
 & \tn{t1}{$t_1$}  & \tn{t2}{\;\;\;\;$t_2$} & \tn{t3}{$t_3$} & \tn{t4}{\;\;\;\;$t_4$} \\
\hline
\\
$r_2$ &&&& \tn{n2}{\;\;\;\;\post{($e_2$)}} \\\myhline
$r_1$ &&& \multicolumn{2}{@{}l}{\tn{n1}{\post{($e_1$)}}} \\\myhline
$r_{10}$ & \multicolumn{2}{@{}l}{\tn{n10}{\acquire{(l)}}} & \\\myhline
$r_{11}$ & \multicolumn{2}{@{}l}{\tn{n11}{\texttt{y = 1}}} & \\\myhline
$r_{12}$ & \multicolumn{2}{@{}l}{\tn{n12}{\release{(l)}}} & \\\myhline
\\\myhline
$r_3$ && \multicolumn{3}{@{}l}{\tn{n3}{\;\;\;\;\acquire{(l)}}} \\\myhline
$r_4$ && \multicolumn{3}{@{}l}{\tn{n4}{\;\;\;\;\texttt{y = 5}}} \\\myhline
$r_5$ && \multicolumn{3}{@{}l}{\tn{n5}{\;\;\;\;\texttt{x = 5}}} \\\myhline
$r_6$ && \multicolumn{3}{@{}l}{\tn{n6}{\;\;\;\;\release{(l)}}} \\\myhline
\\\myhline
$r_7$ &  \multicolumn{2}{@{}l}{\tn{n7}{\acquire{(l)}}} & \\\myhline
$r_8$ &  \multicolumn{2}{@{}l}{\tn{n8}{\texttt{x = 1}}} & \\\myhline
$r_9$ &  \multicolumn{2}{@{}l}{\tn{n9}{\release{(l)}}} & \\\myhline
\\
\end{tabular}
\tikz[remember picture,overlay] \node[fit=(n7)(n8)(n9), draw] {};
\tikz[remember picture,overlay] \node[fit=(n10)(n11)(n12), draw] {};
\tikz[remember picture,overlay] \node[left of=n7,xshift=-.2cm] (e1) {$e_1$};
\tikz[remember picture,overlay] \node[left of=n10,xshift=-.2cm] (e2) {$e_2$};

\caption{An interesting sequence $z$ corresponding to Example~\ref{ex:mod-lock-dep}, 
not explored by EM-DPOR.}
\label{fig:lock-mod-missing}
\end{figure}

\subsection{Eliminate Dependence Between Non-conflicting \texttt{lock} Operations}
\label{sec:lock-dep}
Dependence relation for an event-driven program with a state space $\mathcal{S}_G$ 
and given in Definition~\defdependencerelation,
considers every pair of \texttt{lock} operations on the same lock object to be 
dependent. This holds even for lock acquires on different event handlers on the 
same thread. This is because such operations disable each other in the event-parallel transition 
system $\mathcal{P}_G$ introduced in Section~\ref{sec:relations}
and are thus considered dependent in $\mathcal{P}_G$ (see Definition~\defepmtdependence).
When identifying dependent transitions in different event handlers on the same 
thread, the dependence relation (see Definition~\defdependencerelation) defined 
for $\mathcal{S}_G$ uses the dependences identified over $\mathcal{P}_G$.
Hence, lock acquires on a common lock by different event handlers on the same thread 
will be considered dependent in $\mathcal{S}_G$ as well.
We assume each lock acquired within
an event handler to be released within the same event handler. This is a reasonable
assumption as some widely used programming language features like Java's \texttt{synchronized} 
construct for nested acquire and release of lock objects support this assumption.
Also, acquiring and releasing locks in different event handlers can be hard to reason
and problematic if the event handler acquiring the lock is not guaranteed to always
precede the handler releasing the lock. With this assumption, any pair of \texttt{lock}
operations on the same lock object executed on different event handlers on the same
thread can never contend or deadlock with each other, as (1)~operations executing on the same
thread are never co-enabled, (2)~each handler is assumed to execute to completion 
before the execution of another handler, and (3)~a lock acquired in an event
handler is released within the same event handler as per our assumption. 
Hence, lock acquires from different event handlers on the same thread cannot 
simultaneously involve in interesting states like deadlocks. We thus consider \texttt{lock}
\emph{operations executed on different event handlers on the same thread, even if acquiring
the same lock object, to be independent}. Consequently, we consider any \texttt{unlock}
operation $r$ to be independent with subsequent \texttt{lock} operations in other event handlers
on the same thread as $r$. 

However in theory, considering operations acquiring the same lock $l$ in two different handlers
$h$ and $h'$ on the same thread to be independent is problematic --- especially if the same 
shared variable is accessed (\texttt{read-write} / \texttt{write-write}) by some transitions,
say $r$ and $r'$, in the critical sections protected by the lock $l$ in $h$ and $h'$
respectively, resulting in exploring different states on different ordering of $h$ and $h'$. 
This is because $r$ and $r'$ accessed within critical sections protected by the same lock
in $h$ and $h'$ respectively, are trivially considered independent in $\mathcal{P}_G$
(see Definition~\defepmtdependence) 
as they are never co-enabled. Hence, $r$ and $r'$ may be  considered independent in
$\mathcal{S}_G$ too. 
However, exploring different ordering of $h$ and $h'$
is essential to explore possibly different states due to conflicting accesses
$r$ and $r'$. In such scenarios, considering lock acquires corresponding to
critical sections of $r$ and $r'$ to be dependent enables a POR technique to reorder
$h$ and $h'$ even though the actual conflicting transitions $r$ and $r'$ are not
marked dependent. This will not be possible with our selective \texttt{lock}
dependence proposed above. However in practice, considering all pairs of lock acquires on
different handlers on the same thread to be independent does not result in 
aforementioned problem. This is because, \emdpor\ over-approximates the set of pairs of dependent transitions by considering pairs of
transitions (a)~making conflicting accesses to shared variables with or without
holding a protective lock, or (b)~enabling/disabling each other, to be dependent.
Hence, \texttt{lock} operations on the same object executed on different handlers
on the same thread need not be considered dependent in practice, to enable
EM-DPOR to reorder their respective handlers in case they access the same shared
variable in their critical sections. In the rest of the section we refer to this  
over-approximated dependence relation but additionally considering all the
pairs of \texttt{lock} operations and \texttt{unlock-lock} operations executed 
on different handlers of the same thread to be independent, as \emph{modified dependence relation}.

The modified dependence relation preserves dependence between pairs of \texttt{lock} 
operations and \texttt{unlock-lock} operations on same lock objects and executed 
on different threads. This is because, a lock
acquire disables all other co-enabled \texttt{lock} operations contending for the 
same lock, and \texttt{unlock} enables \texttt{lock} operations waiting for the
same lock; making such transitions dependent due to condition~2 in Definition~\defdependencerelation. 
Similar to the proof for Theorem~\ref{th:preserve-deadlock}, we 
can prove that a dependence-covering state space $\mathcal{S}_R$ of an Android program $A$ obtained by the modified dependence 
relation, preserves all deadlock cycles seen in the original state space $\mathcal{S}_G$
of $A$. This is because, a dependence-covering sequence $u$ of a transition sequence $w$
must preserve the relative order between all the pairs of \texttt{lock} operations
acquiring or contending for the same lock object and executing on different threads 
in $w$, because \texttt{lock} operations are considered dependent. 
Thus, if $w \in \mathcal{S}_G$ reaches a deadlock cycle $\langle DC,\rho \rangle$
then $u$ being its dependence-covering sequence reaches the same deadlock cycle, as $u$
must preserve the relative order of acquiring locks among threads involved in the
deadlock cycle.

\subsubsection*{Modifications to EM-DPOR to Incorporate Modified Dependence Relation}
\label{sec:emdpor-lock-mod}
EM-DPOR should be able to explore all valid interleaving of operations acquiring the same
lock object and executed on different threads, even when using modified dependence 
relation. Example~\ref{ex:mod-lock-dep} demonstrates that achieving this requires 
some modifications to EM-DPOR similar to those introduced in 
Algorithm~\ref{algo:explore-optimized} described in Section~\ref{sec:explore-modified}.

\begin{example}
\label{ex:mod-lock-dep}
\upshape
Consider an execution trace $z_1$ explored by EM-DPOR and given in Figure~\ref{fig:mod-lock-dep}(a), 
of an event-driven multi-threaded program. Among the threads $t_1$, $t_2$, $t_3$
and $t_4$, only $t_1$ is attached with an event queue. EM-DPOR is assumed to
use modified dependence relation, thus making transition pairs $(r_7,r_{10})$ and
$(r_9,r_{10})$ independent. Sequence $z_1$ has two pairs of \emph{may be co-enabled} 
dependent transitions with no happens-before mapping between them: 
$(r_3,r_7)$ and $(r_3,r_{10})$. Note that EM-DPOR does not invoke \texttt{FindTarget}
on transition pairs $(r_5,r_8)$ and $(r_4,r_{11})$ as they are ordered by happens-before
due to happens-before mapping between $r_6$ - $r_7$ and $r_6$ - $r_{10}$ 
respectively.

On exploring sequence $z_1$ Algorithm~\texttt{Explore} identifies 
transition pairs $(r_3,r_7)$ and $(r_3,r_{10})$ as dependent, identifies 
backtracking choices using \texttt{FindTarget}, and eventually explores 
sequences $z_2$ and $z_3$ (Figure~\ref{fig:mod-lock-dep}(b) and (c) 
respectively). However, EM-DPOR does not explore any more interleaving of 
transitions $r_3$ - $r_7$ - $r_{10}$ and thus misses exploring a sequence similar 
to $z$ (Figure~\ref{fig:lock-mod-missing}), where the locking order of lock \texttt{l}
is different compared to that explored by sequences $z_1$, $z_2$ and $z_3$.  Also, 
$z$ reaches a new state (compared to states reached by $z_1$, $z_2$ and $z_3$) 
where variables \texttt{x} and \texttt{y} are assigned values
$1$ and $5$ respectively. 
Even if Algorithm \texttt{Explore} is modified to compute backtracking
choices to reorder a \texttt{lock} operation with all the prior executed 
\emph{may be co-enabled} \texttt{lock} operations with no happens-before relation
(instead of only the nearest \texttt{lock} operation), EM-DPOR will not be 
able to explore sequence $z$. 
\end{example}

\paragraph{Analysis of Example~\ref{ex:mod-lock-dep}} The scenario represented 
in Example~\ref{ex:mod-lock-dep} is similar to that in Example~\ref{ex:readwrite-2}.
Specifically, in Example~\ref{ex:mod-lock-dep} \texttt{lock} in transition $r_3$
is dependent with transitions $r_7$ and $r_{10}$ executed in different handlers
on the same thread while $r_7$ and $r_{10}$ are mutually independent, similar to 
the way \texttt{write} in transition $r_3$ is dependent with \texttt{read} operations
in $r_4$ and $r_5$ executed in different handlers on the same thread in 
Example~\ref{ex:readwrite-2}. Hence, we propose modifications to Algorithm 
\texttt{Explore} similar to those explained in Section~\ref{sec:read-write-dep}.

In the initial phase of Algorithm \texttt{Explore} which invokes \texttt{FindTarget},
we do the following if $next(s,t)$ contains a \texttt{lock} operation. We compute
backtracking choices (by invoking \texttt{FindTarget}) to reorder $next(s,t)$ with 
the nearest \emph{may be co-enabled} (\ie\ not executed on thread $t$) \texttt{lock} operation, say $r$, acquiring 
the same lock object. If $r$ is executed in an event handler (\ie\ $r$ is executed 
on a thread with an event queue), we add a temporary happens-before mapping from 
$r$ to $next(s,t)$. We then compute backtracking choices and backtracking states to reorder
$next(s,t)$ with all the prior \texttt{lock} operations executed in various handlers
on $r$'s thread which do not have a happens-before mapping with $next(s,t)$,
till we find a \texttt{lock} operation, say $r'$, executed on a thread other than $\threadof(r)$ such
that index of $r'$ is lesser than the index of the \texttt{lock} operations on $r$'s
thread reordered with $next(s,t)$. After computing backtracking choices to reorder
$next(s,t)$ with all the relevant \texttt{lock} operations, we remove the temporary
happens-before mapping between $r$ and $next(s,t)$ and continue with the remaining
steps in \texttt{Explore}.

\end{document}